\documentclass[11pt]{article}
\usepackage{graphicx}
\usepackage{subcaption}

\usepackage{titlesec}
\usepackage{natbib}
\usepackage[T1]{fontenc}
\usepackage[utf8]{inputenc}
\usepackage{lmodern}
\usepackage[utf8]{inputenc}
\usepackage{enumitem}
\usepackage{hyperref}
\usepackage{comment}
\usepackage{authblk}
\usepackage{amsmath}
\usepackage{amssymb}
\usepackage{amsthm}

\newtheorem{theorem}{Theorem}[section]
\usepackage[ruled,vlined]{algorithm2e}
\SetKwRepeat{Do}{do}{while}%
\SetKwInput{KwInit}{Initialization}%
\SetKwIF{If}{ElseIf}{Else}{if}{then}{else if}{else}{endif}%
\usepackage{soul}
\usepackage{booktabs}
\usepackage{color}
\usepackage{epsfig}
\usepackage{epstopdf}
\usepackage[section]{placeins}
\usepackage{graphicx}

\usepackage{amsfonts}
\usepackage{mathrsfs}
\usepackage{soul}
\usepackage{multirow}
\usepackage{array}
\usepackage{natbib}

\newsavebox{\theorembox}
\newsavebox{\lemmabox}
\newsavebox{\claimbox}
\newsavebox{\factbox}
\newsavebox{\corollarybox}
\newsavebox{\examplebox}
\newsavebox{\remarkbox}
\newsavebox{\assbox}
\newsavebox{\propositionbox}
\newsavebox{\problembox}
\newsavebox{\defbox}

\newtheorem{lemma}{Lemma}[section]

\savebox{\theorembox}{\noindent\bf Theorem}
\savebox{\lemmabox}{\noindent\bf Lemma}
\savebox{\factbox}{\noindent\bf Fact}
\savebox{\corollarybox}{\noindent\bf Corollary}
\savebox{\examplebox}{\noindent\bf Example}
\savebox{\remarkbox}{\noindent\bf Remark}
\savebox{\assbox}{\noindent\bf Assumption}
\savebox{\propositionbox}{\noindent\bf Proposition}
\savebox{\problembox}{\noindent\bf Problem}
\savebox{\defbox}{\noindent\bf Definition}

\def\blackslug{\hbox{\hskip 1pt \vrule width 4pt height 8pt depth 1.5pt
\hskip 1pt}}

\newcommand{\vs}{\vspace*{3mm}}

\newcommand{\ba}{ {\boldsymbol a} }
\newcommand{\bA}{ {\boldsymbol A} }
\newcommand{\bb}{ {\boldsymbol b} }
\newcommand{\bB}{ {\boldsymbol B} }
\newcommand{\bc}{ {\boldsymbol c} }

\newcommand{\bD}{ {\boldsymbol D} }

\newcommand{\boldf}{ {\boldsymbol f} }

\newcommand{\bg}{ {\boldsymbol g} }
\newcommand{\bG}{ {\boldsymbol G} }
\newcommand{\bh}{ {\boldsymbol h} }
\newcommand{\bH}{ {\boldsymbol H} }

\newcommand{\bI}{ {\boldsymbol I} }

\newcommand{\bL}{ {\boldsymbol L} }

\newcommand{\br}{ {\boldsymbol r} }
\newcommand{\bR}{ {\boldsymbol R} }
\newcommand{\bs}{ {\boldsymbol s} }

\newcommand{\bu}{ {\boldsymbol u} }

\newcommand{\bv}{ {\boldsymbol v} }

\newcommand{\bw}{ {\boldsymbol w} }
\newcommand{\bW}{ {\boldsymbol W} }
\newcommand{\bx}{ {\boldsymbol x} }

\newcommand{\by}{ {\boldsymbol y} }

\newcommand{\bz}{ {\boldsymbol z} }

\newcommand{\balpha}{ {\boldsymbol \alpha} }

\newcommand{\bbeta}{ {\boldsymbol \beta} }
\newcommand{\bgam}{ {\boldsymbol \gamma} }
\newcommand{\bgamma}{ {\boldsymbol \gamma} }

\newcommand{\bepsilon}{ {\boldsymbol \epsilon} }
\newcommand{\bphi}{ {\boldsymbol \phi} }
\newcommand{\bPhi}{ {\boldsymbol \Phi} }

\newcommand{\bmu}{ {\boldsymbol \mu} }

\newcommand{\bet}{ {\boldsymbol \eta} }

\newcommand{\bsigma}{ {\boldsymbol \sigma} }

\newcommand{\bSigma}{ {\boldsymbol \Sigma} }
\newcommand{\btheta}{ {\boldsymbol \theta} }
\newcommand{\bTheta}{ {\boldsymbol \Theta} }

\newcommand{\bxi}{ {\boldsymbol \xi} }

\newcommand{\bzero}{ {\boldsymbol 0} }

\usepackage{hyperref}
\usepackage{comment}
\usepackage{authblk}
\usepackage{amsthm}

\def\blot{\quad {$\vcenter{\vbox{\hrule height.4pt
             \hbox{\vrule width.4pt height.9ex \kern.9ex \vrule
width.4pt}
             \hrule height.4pt}}$}}

\usepackage{multirow}

\usepackage{cleveref}

\begin{document}

\title{MCMC-Free Uncertainty Quantification for Deep Generative Inference for Spatially Varying Coefficient Models at Scale}

\author{
Yeseul Jeon$^{1}$,
Aaron Scheffler$^{2}$, and
Rajarshi Guhaniyogi$^{3,\dagger}$\\
{\small $^{1}$Department of Mathematics and Statistics, 
University of North Carolina at Charlotte}\\
{\small $^{2}$Department of Epidemiology and Biostatistics, 
University of California, San Francisco}\\
{\small $^{3}$Department of Statistics, Texas A\&M University}\\
{\small $^{\dagger}$Corresponding author:
\texttt{rajguhaniyogi@tamu.edu}}
}

\date{}
\maketitle

\begin{abstract}
Varying coefficient (VC) regression models have become indispensable tools in spatial data analysis, providing unmatched flexibility in capturing complex, nonlinear relationships and spatially-varying effects of predictors on responses. Although hierarchical Bayesian approaches offer a rigorous probabilistic framework for uncertainty quantification in VC modeling, their practical application to large-scale spatial datasets remains severely hindered by the scalability limitations of Markov chain Monte Carlo (MCMC) algorithms. In response, the past decade has seen substantial advances in developing more efficient hierarchical Bayesian VC models, primarily through substituting traditional Gaussian processes (GP) with computationally efficient stochastic surrogates for estimating unknown coefficient functions. This article introduces a fundamentally different approach: the \emph{Geostatistical Variational Auto-Encoder} (GeoVAE), a hierarchical deep generative framework built specifically for joint estimation of multiple spatially varying coefficient functions. GeoVAE departs from both classical GP-based Bayesian models and standard variational auto-encoders (VAEs) in two key ways. First, it constructs a coefficient-specific auto-encoder for each coefficient function, allowing each to capture its own spatial resolution and smoothness. Second, a hierarchical synthesis layer integrates information across these auto-encoders to a shared auto-encoder, explicitly modeling cross-coefficient dependencies arising from the shared spatial domain. This design enables GeoVAE to recover complex spatial patterns at computational cost that scales favorably with sample size, without MCMC or explicit GP covariance representations. We characterize the advantages and limitations of GeoVAE relative to hierarchical Bayesian spatial models, demonstrating its strong potential for large-scale spatial analysis.
\end{abstract}

\baselineskip18pt

\section{Introduction}\label{sec:intro}
Varying coefficient models (VCMs) offer a flexible extension to traditional linear regression frameworks \citep{gelfand2003spatial} by allowing regression coefficients to vary as smooth functions of an index variable, thereby accommodating nonlinear dependence between the response and covariates in a parsimonious and interpretable manner. In the Bayesian paradigm, VCMs seamlessly integrate the interpretability of parametric approaches with the adaptability of nonparametric methods, whilst delivering principled uncertainty quantification through the posterior distribution. These features render  Bayesian VCMs particularly attractive for applications involving temporal, spatial, and spatiotemporal data.

Let $\mathcal{D} \subseteq \mathbb{R}^d$ denote the $d$-dimensional index set, with indexing variable $\bu$, where $d = 1$ for temporal data, $d = 2$ for spatial data. This article mainly focuses on the spatial outcome and predictors such that $d=2$. At each $\bu \in \mathcal{D}$, a spatial outcome $y(\bu) \in \mathbb{R}$ and $J$ predictors $x_1(\bu), \ldots, x_J(\bu) \in \mathbb{R}$ are related through the VCM:
\begin{equation} 
\label{eq:VCM_basic}
    y(\bu) = \beta_0(\bu) + \sum_{j=1}^{J} x_j(\bu)\,\beta_j(\bu) + \epsilon(\bu),
\end{equation}
where $\bbeta(\bu) = (\beta_1(\bu), \ldots, \beta_J(\bu))^\top$ are the spatially varying coefficient functions, $\beta_0(\bu)$ is the spatially varying intercept, and $\epsilon(\bu) \sim \mathcal{N}(0, \tau^2)$ captures measurement error. VCMs with spatially varying coefficients are well suited to capturing smoothly varying predictor effects, and have found successful applications in a wide variety of spatial modeling \citep{gelfand2003spatial, mu2018estimation, guhaniyogi2025bayesian, guhaniyogi2022distributed, 
kim2021generalized}. spatially varying coefficient (SVC) models serve as principled process-based alternatives to geographically weighted regression \citep{fotheringham2009geographically, wheeler2021geographically} for modeling spatially nonstationary behavior in the mean, with documented utility across ecological and environmental applications \citep{zhan2026mapping}.

The canonical Bayesian approach to model \eqref{eq:VCM_basic} assigns a multivariate Gaussian process (GP) prior to $\bbeta(\bu)$, with a 
cross-covariance kernel encoding both within-coefficient spatial correlation and across-coefficient dependencies \citep{gelfand2003spatial}. Whilst this formulation is conceptually elegant and theoretically well-understood, it is severely constrained in practice by the computational demands of GP inference. Updating $\bbeta(\bu)$ at $N$ spatial locations within each Gibbs sampling iteration requires the inversion and storage of a $JN \times JN$ cross-covariance matrix, incurring a storage cost of $\mathcal{O}(J^2 N^2)$ and a computational cost of $\mathcal{O}(J^3 N^3)$. Consequently, standard MCMC-based inference becomes computationally prohibitive for large $N$ and 
even moderate $J$.


\subsection{Existing Approaches}
The literature on scalable Bayesian inference for large spatial datasets has expanded considerably in recent years; see \citet{heaton2019case} for a comprehensive review. Broadly, scalable approaches have proceeded broadly along two lines: model-based approaches, including low-rank dimension reduction \citep{cressie2011statistics, guhaniyogi2011adaptive, guhaniyogi2013modeling}; multi-resolution representations \citep{nychka2015multiresolution, guhaniyogi2020large}; sparsity-inducing GP constructions \citep{katzfuss2021general, sauer2023vecchia}; algorithmic approaches, including meta-kriging strategies \citep{andros2024robust, guhaniyogi2023distributed, guhaniyogi2025bayesian}. Despite the impressive advances these methods represent, they all remain fundamentally anchored to the hierarchical Bayesian paradigm: they reduce or approximate GP computations while retaining an MCMC-based inference engine, and therefore continue to face scalability challenges for very large datasets.

A parallel and rapidly growing literature has sought to harness the 
representational power of deep neural networks for spatial and spatiotemporal modelling. Graph neural networks (GNNs) and their spatiotemporal extensions \citep{Kipf2017, Hamilton2017, Wu2019} have demonstrated strong empirical performance in learning over irregular spatial graphs, whilst convolutional neural networks (CNNs) have been adapted to capture local spatial dependencies in gridded data \citep{LeCun1998, Reichstein2019}. Physics-informed neural 
networks \citep{Raissi2019} and neural ordinary differential equations 
\citep{Chen2018NODE} further extend this paradigm to spatiotemporal processes governed by physical laws. More recently, transformer-based architectures with spatial attention mechanisms \citep{Vaswani2017,Gao2022} have been applied to large-scale geospatial prediction tasks. Despite their considerable predictive power, these approaches offer limited uncertainty quantification, and are not designed for the interpretable estimation of spatially varying regression coefficients.

Variational autoencoders (VAEs) \citep{Kingma2014} and generative adversarial networks (GANs) \citep{Goodfellow2014} have emerged as powerful frameworks for learning latent representations of complex high-dimensional data. In spatial and geoscientific applications, deep generative models have been used for probabilistic precipitation nowcasting~\citep{Ravuri2021}, stochastic downscaling of atmospheric fields and precipitation forecasts \citep{Leinonen2020, Harris2022},
and geological model generation and geostatistical inversion \citep{Laloy2018}. Conditional generative models have also been developed for surrogate modeling and uncertainty quantification in systems governed by partial differential equations \citep{Zhu2019}. Gaussian 
process-based latent variable models \citep{Lawrence2005, Titsias2010} provide a principled probabilistic bridge between GP priors and latent variable frameworks, and have been extended to deep GP architectures \citep{Damianou2013, Salimbeni2017} for flexible nonparametric modeling of spatial fields. Conditional neural processes \citep{Garnelo2018, Gordon2020}. Gaussian neural processes~\citep{Bruinsma2021} further model correlations among predictions, providing a framework for meta-learning predictive stochastic processes. Whilst these generative approaches are highly expressive  and computationally scalable, they are not tailored to the problem of jointly estimating multiple spatially varying coefficient functions in a regression framework, nor do they directly address the inferential goals of VCMs, such as quantifying the spatially heterogeneous impact of individual predictors on a response variable.

\subsection{Our Contribution: GeoVAE}
We address this gap by introducing the \emph{Geostatistical Variational Autoencoder} (GeoVAE), a hierarchical deep generative framework purpose-built for the scalable and joint estimation of spatially varying coefficient functions in \eqref{eq:VCM_basic}. GeoVAE is a cohesive integration of five distinct  methodological innovations, each targeting a specific challenge in SVC estimation and together forming a unified, end-to-end inferential framework.

\textbf{(A) Coefficient-specific encoders.} Rather than learning a single shared latent representation for all coefficient functions, GeoVAE assigns a dedicated encoder to each of the $J$ coefficients. Each encoder independently maps the observed response vector to a coefficient-specific latent distribution, allowing the spatial structure of each coefficient, its intrinsic resolution, smoothness, and local variation, to be learned separately and at its own natural scale.

\textbf{(B) Hierarchical synthesis across coefficients.} A hierarchical synthesis layer consolidates information across all $J$ coefficient-specific encoders, explicitly capturing the cross-coefficient dependencies that arise from the shared spatial domain and common outcome structure. Via a shared latent pooling mechanism, the estimation of each coefficient is informed by the spatial patterns learned across the remaining coefficients, providing a flexible \emph{neural analogue of the multivariate GP cross-covariance structure} without requiring its explicit construction or inversion.

\textbf{(C) Coefficient-specific decoders.} Each of the $J$ coefficient-specific decoders reconstructs the $j$th varying coefficient function by drawing on both the individual latent representation from its corresponding encoder and the shared cross-coefficient representation produced by the hierarchical synthesis layer. The decoder outputs the basis-coefficient vector associated with the $j$th varying coefficient function, enabling flexible and spatially adaptive reconstruction.

\textbf{(D) End-to-end outcome reconstruction and calibrated uncertainty.} The GeoVAE variational objective is constructed so that the synthesized latent representations feed directly into the reconstruction of the observed spatial outcome. Specifically, the pooled cross-coefficient latent vector is passed through the decoders to reconstruct the response at each observed location, ensuring that the variational objective is directly tied to predictive fidelity. This end-to-end formulation simultaneously drives the optimization of the coefficient-specific encoders, the hierarchical synthesis layer, and the decoders within a single unified training procedure. The resulting joint estimation strategy guarantees that all learned coefficient functions are mutually consistent, spatially coherent, and aligned with predictive accuracy, offering a principled and computationally efficient alternative to multivariate GP-based Bayesian inference. To account for the modest underestimation of predictive uncertainty that can arise from stochastic gradient variational inference, we incorporate a post-hoc calibration step based on split conformal prediction. This procedure constructs predictive intervals that correct the nominal coverage of the variational predictive distribution, ensuring that the reported uncertainty estimates closely attain the desired frequentist coverage guarantees without requiring any modification to the core variational training procedure.

\textbf{(E) Scalable variational inference.} By replacing the MCMC inference engine with stochastic gradient variational Bayes, GeoVAE avoids the construction and inversion of $N\times N$ spatial covariance matrices entirely. With $H$ basis functions per coefficient, the dominant response-reconstruction cost is approximately $O(NH(J+1))$ per full-data optimization step, scaling favorably in $N$ and rendering large-scale spatial inference feasible in regimes where GP-based MCMC remains computationally challenging.


The remainder of the paper proceeds as follows. Section~\ref{sec:spatial_basis_expansion} develops the GeoVAE modeling framework, including the spatially varying coefficient model and its basis representation. Section~\ref{sec:hierarchical_geovae} presents the variational inference procedure in detail, covering the coefficient-specific encoder-decoder pairs and the hierarchical synthesis layer through which a shared encoder pools information across coefficient-specific latent representations. Section~\ref{sec:theoretical_study} establishes theoretical guarantees for the proposed framework, deriving consistency of the estimated varying coefficient surfaces and predictive consistency under a set of regularity conditions. Section~\ref{sec:simulation} reports simulation studies assessing the finite-sample performance of GeoVAE across a range of settings, and Section~\ref{sec:realdata} demonstrates the methodology on a remote sensing vegetation data application. Section~\ref{sec:conclusion} concludes with a discussion and directions for future research. Appendix offers proofs of the theoretical results presented in Section~\ref{sec:theoretical_study}.

\section{Methodology}\label{sec:method}

\subsection{Spatially Varying Coefficient Model with Basis Representation}\label{sec:spatial_basis_expansion}
Directly estimating $\beta_j(\bu)$ at every observed location would result in an ill-posed, high-dimensional problem. To obtain a structured and computationally tractable formulation, we instead
represent each coefficient function through a finite-dimensional basis expansion. Specifically, for $j = 0, \ldots, J$, we write
\begin{equation}
\label{eq:basis_representation}
\beta_j(\bu)
= \sum_{k=1}^{K_j} B_{jk}(\bu)\,\alpha_{jk}
= \mathbf{B}_j(\bu)^\top \boldsymbol{\alpha}_j,
\end{equation}
where
$\mathbf{B}_j(\bu) = \bigl(B_{j1}(\bu), \ldots, B_{jK_j}(\bu)\bigr)^\top \in \mathbb{R}^{K_j}$ is the vector of pre-specified spatial basis functions evaluated at $\bu$, and $\boldsymbol{\alpha}_j = (\alpha_{j1}, \ldots, \alpha_{jK_j})^\top
\in \mathbb{R}^{K_j}$ is the corresponding vector of basis coefficients, which are shared across all spatial locations and serve as the unknown parameters to
be estimated. The basis functions $\{B_{jk}(\cdot)\}$ can be chosen to encode the anticipated smoothness of $\beta_j(\bu)$; common choices include B-spline basis functions \citep{shen2015adaptive}, radial basis functions \citep{buhmann2000radial}, thin-plate splines \citep{wood2003thin}, or Gaussian kernels centered at a set of knots distributed over $\mathcal{D}$ \citep{guhaniyogi2011adaptive}. Through the expansion  \eqref{eq:basis_representation},
the infinite-dimensional estimation problem for the continuously defined surface $\beta_j(\bu)$ is reduced to estimating the finite-dimensional parameter vector $\boldsymbol{\alpha}_j \in\mathbb{R}^{K_j}$, thereby substantially reducing the effective
dimension of the parameter space, and offering similar results with a sufficiently large choice of $K_j$.

Let the data be observed at $N$ spatial locations $\bu_1,...,\bu_N$. To unify the intercept and predictor effects within a single framework, we set $x_0(\mathbf{u}_i) = 1$ for all $i=1,...,N$ and stack all basis coefficient vectors into a single global
parameter vector
$\boldsymbol{\alpha}
= \bigl(\boldsymbol{\alpha}_0^\top,\,
         \boldsymbol{\alpha}_1^\top,\,\ldots,\,
         \boldsymbol{\alpha}_J^\top\bigr)^\top
\in \mathbb{R}^{K},\:
K = \sum_{j=0}^{J} K_j.$
For each observation $i$, we construct a
\emph{basis-covariate design vector} by scaling the basis evaluations for each coefficient by the corresponding covariate value
$\mathbf{D}_i
= \Big[
    x_0(\bu_i)\mathbf{B}_0(\bu_i)^\top,\;
    x_1(\bu_i)\mathbf{B}_1(\bu_i)^\top,\;
    \ldots,$\\
    $
    x_J(\bu_i)\mathbf{B}_J(\bu_i)^\top
  \Big]
\in \mathbb{R}^{1 \times K}.$
By construction,
$\mathbf{D}_i \boldsymbol{\alpha}
= \sum_{j=0}^J x_j(\bu_i)\,
  \mathbf{B}_j(\bu_i)^\top \boldsymbol{\alpha}_j
= \sum_{j=0}^J x_j(\bu_i)\,\beta_j(\bu_i)$. Stacking these row vectors over all $N$ observations yields the global design matrix
$\mathbf{D}
= \bigl[\mathbf{D}_1^\top :\, \cdots :\,
         \mathbf{D}_N^\top\bigr]^\top
\in \mathbb{R}^{N \times K},$
so that the VC model \eqref{eq:VCM_basic} admits the
compact matrix form
\begin{equation}
\label{eq:matrix_svc}
\mathbf{y} = \mathbf{D}\boldsymbol{\alpha} + \boldsymbol{\epsilon},
\qquad
\boldsymbol{\epsilon} \sim \mathcal{N}(\mathbf{0},\, \tau^2 \mathbf{I}_n),
\end{equation}
where $\mathbf{y} = (y(\bu_1), \ldots, y(\bu_N))^\top$. Equation~\eqref{eq:matrix_svc} is formally a standard Gaussian linear regression model in the stacked basis coefficient vector $\boldsymbol{\alpha}$. Crucially, all spatial structure of the original VC model is fully embedded within the design matrix $\mathbf{D}$: once $\mathbf{D}$ is constructed from the observed locations and covariate values, standard linear model machinery applies directly to inference on $\boldsymbol{\alpha}$. We also denote $\bD^{(j)}=[x_j(\bu_1)\bB_j(\bu_1):\cdots:x_j(\bu_N)\bB_j(\bu_N)]^T\in\mathbb{R}^{N\times K_j}$ as the design matrix corresponding to the coefficient vector $\balpha_j$ representing the basis coefficient for the $j$th predictor.

\subsection{Hierarchical VAE-Based Estimation of the Basis Coefficients}
\label{sec:hierarchical_geovae}
The predominant approach in the VCM literature for spatial statistics is to place a multivariate Gaussian prior \citep{gelfand2003spatial}, or a scale mixture thereof \citep{PalaciosSteel2006}, on the coefficient functions, thereby jointly capturing inter-coefficient dependencies while accommodating high-dimensionality. We depart fundamentally from this paradigm and instead propose a hierarchical variational autoencoder (VAE) framework for estimating the basis coefficients $\balpha_0,...,\balpha_J$ that determine the spatially varying coefficient functions. While the basis expansion in Section~\ref{sec:spatial_basis_expansion} reduces the infinite-dimensional estimation problem to a finite-dimensional one, the basis coefficient vectors associated with different predictors may encode heterogeneous spatial structures. At the same time, since all coefficient functions jointly explain the same outcome, they are likely to share common latent information that, if properly exploited, can improve the estimation of each individual coefficient.

To accommodate both coefficient-specific and shared structure, we introduce a two-level hierarchical latent representation. At the first level, a separate encoder is constructed for each coefficient function $\beta_j(\bu)$, producing a coefficient-specific latent vector $\bz_j\in\mathbb{R}^{d_j}$
that captures the spatial structure unique to the $j$th coefficient. At the second level, the collection of coefficient-specific latent variables is jointly mapped to a shared global latent vector $\bz_g \in \mathbb{R}^{d_g}$, which distills structure that is common across all coefficient functions. Each basis coefficient vector $\balpha_j$ is then generated from both its own coefficient-specific representation $\bz_j$ and the shared representation $\bz_g$, enabling each coefficient to retain its individual latent structure while borrowing strength from the remaining coefficients through the shared latent layer.  Figure~\ref{fig:architecture} provides a schematic overview of the full GeoVAE architecture.

\begin{figure}[htbp]
\centering
\includegraphics[width=\linewidth]{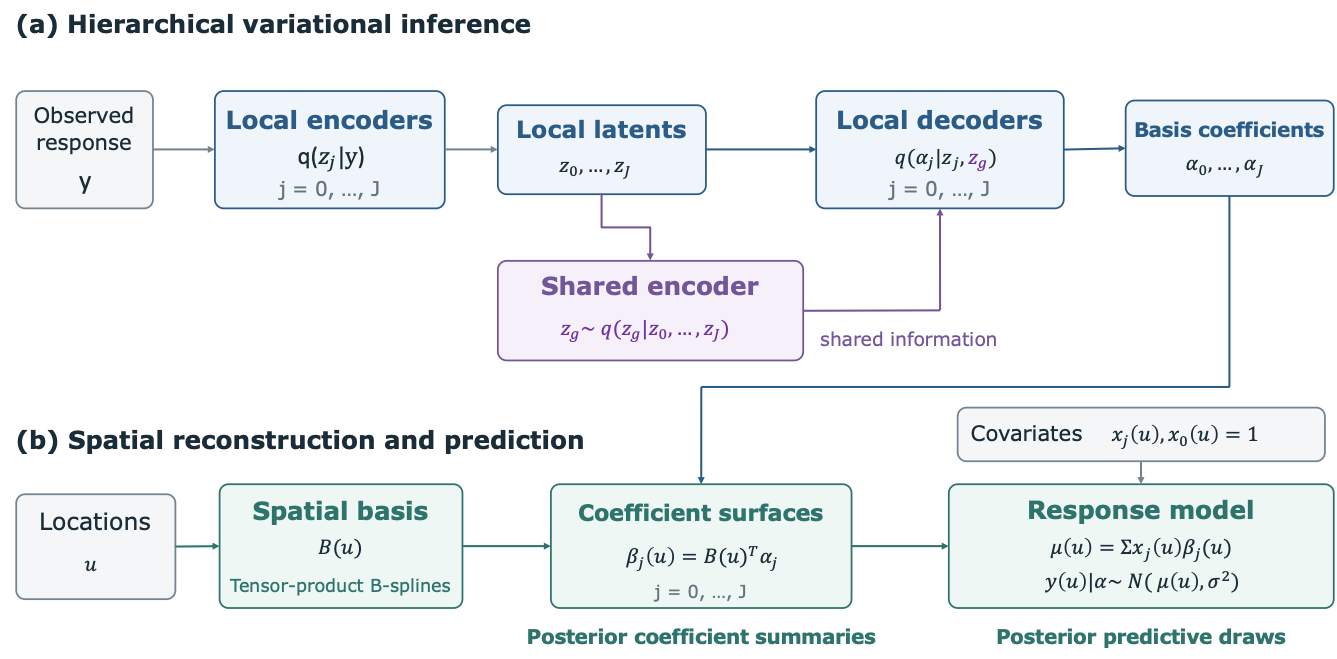}
\caption{
Architecture of GeoVAE.
(a) Hierarchical variational inference combining coefficient-specific and shared latent representations to estimate the basis coefficients.
(b) Spatial reconstruction of coefficient surfaces using basis functions and response prediction using covariates.
}
\label{fig:architecture}
\end{figure}

\subsubsection{Coefficient-Specific Encoders}
For each $j=0,...,J$, the $j$th coefficient-specific encoder maps the observed response vector $\by \in \mathbb{R}^N$ to a coefficient-specific approximate posterior distribution over the latent variable $\bz_j$. We adopt a mean-field Gaussian variational family:
\begin{equation}
q_{\bphi_j}(\bz_j\mid \by))
=
N\left(
\bmu_{\bphi_j}(\by),
\operatorname{diag}
\left\{
\bsigma_{\bphi_j}^2(\by)
\right\}
\right),
\label{eq:local_encoder}
\end{equation}
where $\bmu_{\bphi_j}(\by)\in\mathbb{R}^{d_j}$ and $\bsigma_{\bphi_j}^2(\by)\in\mathbb{R}_{+}^{d_j}$ are the variational mean and standard deviation vectors, both parameterized as functions of $\by$ through the $j$th encoder network with parameters $\bphi_j$. To enable gradient-based optimization, we employ the reparameterization formula and draw samples as
\begin{equation}
\bz_j
=
\bmu_{\bphi_j}(\by)
+
\bsigma_{\bphi_j}(\by)
\odot
\bxi_j,
\qquad
\bxi_j\sim N(\bzero,\bI_{d_j}),
\label{eq:local_reparameterization}
\end{equation}
where $\odot$ denotes the element-wise (Hadamard) product, and samples are drawn independently for $j=0,...,J$ conditional on the observed data.

It is important to note that all $J+1$ coefficient-specific encoders receive the same observed response vector $\by$ as input, but possess entirely separate network parameters $\bphi_0,...,\bphi_J$. Consequently, each encoder produces a distinct latent representation tailored to its respective coefficient function. The association between the $j$th latent branch and the $j$th coefficient function is induced through the corresponding branch-specific decoder and the joint response reconstruction term in the variational objective, described in Section~\ref{sec:elbo}.

\subsubsection{Hierarchical Synthesis via Shared Encoder}
Although the latent variables $\bz_0,...,\bz_J$ are constructed independently, the corresponding coefficient functions jointly explain the same response. To capture shared structure across all coefficient functions, we introduce a second-level shared encoder. Let $\bz=(\bz_0^T,...,\bz_J^T)^T\in\mathbb{R}^{\sum_{j=0}^Jd_j}$ denote the concatenation of all coefficient-specific latent variables. The shared encoder maps $\bz$ to a global latent variable $\bz_g\in\mathbb{R}^{d_g}$
via the approximate variational posterior
\begin{equation}
q_{\bphi_g}(\bz_g\mid \bz)
=
N\left(
\bmu_{\bphi_g}(\bz),
\operatorname{diag}
\left\{
\bsigma_{\bphi_g}^2(\bz)
\right\}
\right),
\label{eq:global_encoder}
\end{equation}
where $\bmu_{\bphi_g}(\bz)\in\mathbb{R}^{d_g}$ and $\bsigma_{\bphi_g}^2(\bz)\in\mathbb{R}_{+}^{d_g}$ are parameterized through a shared encoder network with parameters $\bphi_g$. A sample from the shared latent distribution is obtained via reparameterization as
\begin{equation}
\bz_g = \bmu_{\bphi_g}(\bz) + \bsigma_{\bphi_g}(\bz)
\odot \bxi_g, \qquad
\bxi_g\sim N(\bzero,\bI_{d_g}).
\label{eq:global_reparameterization}
\end{equation}
The shared latent variable $\bz_g$ therefore provides a learned, compact summary of structure that is common across all coefficient-specific representations. This allows the model to preserve effect-specific features through each $\bz_j$ while enabling cross-coefficient information sharing through $\bz_g$.

\subsubsection{Coefficient-Specific Decoders}
Each coefficient-specific decoder generates the basis coefficient vector $\balpha_j$ from both the coefficient-specific latent variable $\bz_j$ and the shared latent variable $\bz_g$. Define the combined latent input for the $j$th decoder as the concatenation $\bgamma_j = (\bz_j^\top, \bz_g^\top)^\top\in\mathbb{R}^{d_j+d_g}$ for $j=0,\ldots,J$. The $j$th coefficient-specific decoder specifies a Gaussian conditional distribution over $\balpha_j$ as
\begin{equation}
p_{\btheta_j}(\balpha_j\mid \bgam_j)
=
N\left(
\bmu_{\btheta_j}(\bgamma_j),
\operatorname{diag}
\left\{
\bsigma_{\btheta_j}^2(\bgamma_j)
\right\}
\right),
\label{eq:local_decoder}
\end{equation}
where $\bmu_{\btheta_j}(\bgamma_j)\in\mathbb{R}^{K_j}$ and $\bsigma_{\btheta_j}^2(\bgamma_j)\in\mathbb{R}^{K_j}_{+}$ are the decoder mean and standard deviation vectors, parameterized through a fully connected network with parameters $\btheta_j$. Samples are obtained via reparameterization as
\begin{equation}
\balpha_j
=
\bmu_{\btheta_j}(\bgamma_j)
+
\bsigma_{\btheta_j}(\bgamma_j)
\odot
\bxi_{\balpha_j},
\qquad
\bxi_{\balpha_j}\sim N(\bzero,\bI_{K_j}).
\label{eq:a_reparameterization}
\end{equation}

The use of both $\bz_j$ and $\bz_g$ in \eqref{eq:local_decoder} is central to the hierarchical construction. The coefficient-specific latent variable $\bz_j$ preserves information unique to the $j$th varying coefficient, whereas $\bz_g$ carries information synthesized
across all coefficient functions. Consequently, the basis coefficient vector
$\balpha_j$ is estimated using both coefficient-specific and shared information. Once $\balpha_j$ is generated, the corresponding coefficient surface is recovered via basis representation $\bbeta_j(\bu) = B_j(\bu)^\top \balpha_j, ~ j=0,\ldots,J.$ Conditional on $\balpha_0,\ldots,\balpha_J$, the fitted response at location $\bu_i$ is $
= \sum_{j=0}^{J} x_j(\bu_i)\bB_j(\bu_i)^\top \balpha_j=\bD_i\balpha,$ and the observation model is
\begin{equation}
p(\by\mid \balpha_0,\ldots,\balpha_J,D)
=
N\left(
\sum_{j=0}^{J}\bD^{(j)} \balpha_j,
\tau^2I_N
\right).
\label{eq:hier_observation}
\end{equation}

\subsubsection{Deep Neural Network Representation of the Encoders and Decoders}\label{sec:nn_representation}
We now specify the neural network architectures used to parameterize the Gaussian mean and variance functions appearing in the coefficient-specific encoders~\eqref{eq:local_encoder}, the shared encoder~\eqref{eq:global_encoder}, and the coefficient-specific decoders~\eqref{eq:local_decoder}. We proceed in three steps: we first define fully connected networks as a general building block, then introduce one-dimensional convolutional networks for processing the spatial response vector, and finally assemble these components into the complete encoder-decoder parameterization.

\noindent\underline{\textbf{Fully connected networks.}} The unknown mean and variance functions of each variational distribution are modeled using \emph{fully connected neural networks} (FCNNs). For an input vector $\bv$, the hidden layer representations for an FCNN are computed as
\begin{align}
\bh^{(o)}=
\rho
\left(
\bW^{(o)}
\bh^{(o-1)}
+
\bb^{(o)}
\right),
~ \bh^{(0)}= \bv, ~ o=1,\ldots,O,
\label{eq:full_connected}
\end{align}
where $O$ is the number of hidden layers, $\bW^{(o)}$ and $\bb^{(o)}$ are the weight matrix and bias vector at layer $o$, and $\rho(\cdot)$ is the activation function. Given a desired output dimension $r$, two separate linear output heads map the final hidden representation $\bh^{(O)}$
to a mean vector and an unconstrained log-variance vector
\begin{equation}
    \bmu_r(\bv) = \bA_{\mu} \bh^{(O)} + \ba_{\mu} \in \mathbb{R}^r, \qquad
    \br_r(\bv) = \bA_{\sigma} \bh^{(O)} + \ba_{\sigma} \in \mathbb{R}^r,
    \label{eq:output_heads}
\end{equation}
where $\bA_{\mu}, \bA_{\sigma} \in \mathbb{R}^{r \times d_O}$ and $\ba_\mu$, $\ba_\sigma\in\mathbb{R}^r$ are learnable parameters, and $d_O$ is the width of the final hidden layer. The corresponding standard deviation vector is obtained by applying the element-wise soft-plus or exponential transformation to ensure positivity defined as $\bsigma_r(\bv) = \exp \left\{\frac{1}{2}\br_r(\bv)\right\}\in \mathbb{R}_{+}^r$ where the exponential is applied element-wise ensuring that every component of $\bsigma_r(\bv)$ is strictly positive. We denote the resulting Gaussian parameter map compactly as
$\mathcal G_r(\bv;\boldsymbol\eta) = \left\{\bmu_\bet(\bv), \bsigma_\bet(\bv) \right\},$ where $\bet$ collects all hidden-layer weights, biases, and output-head parameters of the FCNN.

\noindent\underline{\textbf{Coefficient-specific 1D convolutional networks.}} Since $\by$ consists of scalar responses at $N$ spatial locations, we treat it as a one-dimensional sequence and process it using a \emph{one-dimensional convolutional neural network} (1D-CNN). This allows each encoder to extract spatially local features from the response without requiring the spatial coordinates as explicit inputs. Specifically, let $\pi$ denote a fixed ordering of the observed spatial locations, constructed once from the observed location coordinates and retained throughout model fitting. Define the sequentially ordered response vector $\by^{\mathrm{seq}} = \lbrace y(\bs_{\pi(1)}),\ldots, y(\bs_{\pi(N)}) \rbrace \in\mathbb R^{1\times N}.$ For each $j=0,...,J$, the 1D-CNN of the $j$th encoder processes $\by^{\mathrm{seq}}$ through $L$ convolutional layers. Setting $\bH_j^{(0)}=\by^{\mathrm{seq}}$, the feature maps are computed recursively as
\begin{equation}
\bH_j^{(\ell)}
=
\rho
\left(
\bW_{j\ell}^{C}
*
\bH_j^{(\ell-1)}
+
\bb_{j\ell}^{C}
\right),
\qquad
\ell=1,\ldots,L,
\label{eq:feature}
\end{equation}
where $*$ denotes one-dimensional convolution, $\bW_{jl}^C$ and $\bb_{jl}^C$
are the convolution filter weights and biases at layer $l$, and $\rho(\cdot)$ is an element-wise nonlinear activation function (e.g., ReLU). The output of the final convolutional layer,  $\bH_j^{(L)}\in\mathbb R^{C_L\times T_L}$, where $C_L$ is the number of output channels and $T_L$ is the sequence length after convolution, is then reduced to a fixed-length feature vector using \emph{adaptive average pooling} (AAP) over a target sequence length $L_{\textrm{pool}}$
\begin{equation}
\bc_j
=
\operatorname{vec}
\left[
\operatorname{AAP}{L{\mathrm{pool}}}
\lbrace
\bH_j^{(L)}
\rbrace
\right]
\in
\mathbb R^{C_L L_{\mathrm{pool}}}.
\label{eq:cnn}
\end{equation}
Crucially, the dimension of $\bc_j$ is $C_L L_{\mathrm{pool}}$, which is independent of the number of observed locations $N$. This ensures that the downstream fully connected layers have fixed input dimensions regardless of dataset size.

\noindent\underline{\textbf{Network parameterizations for coefficient-specific encoders and decoders, and shared encoder.}}
Using the convolutional feature vector $\bc_j$ defined in~\eqref{eq:cnn} as input to the fully connected Gaussian parameter map $\mathcal{G}$, the variational parameters of the coefficient-specific encoders, the shared encoder, and the coefficient-specific decoders are parameterized as follows:
\begin{equation}
\begin{aligned}
    \left\{ \bmu_{\bphi_j}(\by),\, \bsigma_{\bphi_j}(\by) \right\} 
    &= \mathcal{G}_{d_j}\!\left(\bc_j;\, {\boldsymbol \phi}_j\right), \qquad j = 0, \ldots, J, \\[4pt]
    \left\{ \bmu_{\bphi_g}(\bz_g),\, \bsigma_{\bphi_g}(\bz_g) \right\} 
    &= \mathcal{G}_{d_g}\!\left(\bz_g;\, \bphi_g\right), \\[4pt]
    \left\{ \bmu_{\btheta_j}(\bgamma_j),\, \bsigma_{\btheta_j}(\bgamma_j) \right\} 
    &= \mathcal{G}_{K_j}\!\left(\bgamma_j;\, \btheta_j\right), \qquad j = 0, \ldots, J,
\end{aligned}
\label{eq:encoder_decoder_networks}
\end{equation}
where $\bphi_j$ denotes the FCNN parameters of the $j$th coefficient-specific encoder (applied to the pooled convolutional feature $\bc_j$), $\bphi_g$ denotes the parameters of the shared encoder (applied to the concatenated latent vector $\bz_g$), and $\btheta_j$ denotes the parameters of the $j$th coefficient-specific decoder (applied to the combined latent input $\bgamma_j = (\bz_j^\top,\, \bz_g^\top)^\top$).

Let $\bPhi=\{\bphi_1,...,\bphi_J\}$ be the set of all coefficient specific encoder parameters and $\bTheta=\{\btheta_1,...,\btheta_J\}$ be the coefficient-specific decoder parameters. The complete set of parameters of the GeoVAE model is 
$\boldsymbol{\Theta} \cup \bPhi \cup \left\{ \bphi_g, \tau^2 \right\}.$ All parameters are estimated jointly through the variational objective described in Section~\ref{sec:elbo}.


\subsubsection{Joint Variational Objective}\label{sec:elbo}

We assign standard Gaussian priors to the coefficient-specific and shared latent variables $p(\bz_j)=N(\bzero,\bI_{d_j}),~ j=0,\ldots,J,~p(\bz_g)=N(\bzero,\bI_{d_g})$. These priors act as regularizers, encouraging compact and well-structured latent representations while preventing the encoders from memorizing the input data. The joint variational distribution can be written as
\begin{equation}
q_{\bPhi}(\bz_0,\ldots,\bz_J,\bz_g\mid \by)
=
q_{\bphi_g}(\bz_g\mid \bz)
\prod_{j=0}^{J}
q_{\bphi_j}(\bz_j\mid \by).
\label{eq:joint_variational}
\end{equation}
This factorization reflects the two-level latent structure of the model. At the lower level, each coefficient-specific encoder $q_{\bphi_j}$ maps the observed response $\by$ directly to its own local latent code $\bz_j$. At the upper level, the shared encoder $q_{\bphi_g}$ takes the collection of local codes $\bz$ as input and produces the global latent code $\bz_g$, which captures variation that is common across all coefficient functions. Conditional on the local codes, the global code is independent of the raw data $\by$, which is consistent with the generative structure of the model.

All model components are estimated jointly by maximizing a variational lower bound on the log of the likelihood in Equation~\eqref{eq:hier_observation}. Equivalently, we minimize the following negative weighted evidence lower bound (ELBO)
\begin{align}
&\mathcal{L}(\bPhi,\bTheta,\bphi_g,\tau^2)
=\;
\mathbb{E}_{q_{\bPhi}}
\mathbb{E}_{\prod_{j=0}^{J}
p_{\btheta_j}(\balpha_j\mid \bz_j,\bz_g)}
\left[
\frac{1}{2\tau^2}
\left\|
\by-\sum_{j=0}^{J}\bD^{(j)} \balpha_j
\right\|_2^2
\right]
+
\lambda_{\mathrm{local}}
\sum_{j=0}^{J}
\operatorname{KL}
\left\{
q_{\bphi_j}(\bz_j\mid \by)
\,
\|\,N(\bzero,\bI_{d_j})
\right\}\nonumber\\
&\qquad\qquad+
\lambda_{\mathrm{global}}
\mathbb{E}_{
\prod_{j=0}^{J}q_{\bphi_j}(\bz_j\mid \by)
}
\left[
\operatorname{KL}
\left\{
q_{\bphi_g}(\bz_g\mid \bz)
\,
\|\,N(\bzero,\bI_{d_g})
\right\}
\right],
\label{eq:hierarchical_elbo}
\end{align}
up to terms that do not depend on the model parameters. The first term in (\ref{eq:hierarchical_elbo}) is the response \emph{reconstruction loss}. It requires the coefficient-specific basis coefficients generated from the hierarchical latent representation to jointly reconstruct the observed response through the original SVC model. Thus, although separate encoders and decoders are used for different coefficient functions, the branches are not estimated independently; they are trained jointly through the common response likelihood.

The second term, referred to as the \emph{local KL regularization}, penalizes each coefficient-specific posterior $q_{\bphi_j}(\bz_j|\by)$ for deviating from its standard Gaussian prior. This encourages each local latent space to be smooth and regularly structured, preventing individual encoders from learning degenerate or overly complex representations. The weight $\lambda_{\mathrm{local}}$ controls the strength of this regularization. The third term applies the same type of regularization to the shared latent variable $\bz_g$. Because $\bz_g$ is itself a function of the stochastic local latent variables $\bz$, its KL divergence must be averaged over the distribution of those local codes, which explains the expectation appearing in this term. The weight $\lambda_{\mathrm{global}}$ controls how strongly the shared representation is regularized relative to the local ones.


For a diagonal Gaussian distribution $q=N\{\bmu,\operatorname{diag}(\bsigma^2)\}$, with $\bmu=(\mu_1,..,\mu_R)^\top$ and $\bsigma=(\sigma_1,...,\sigma_R)^\top$, the KL divergence from the standard Gaussian distribution has the closed-form expression
\begin{equation}
\operatorname{KL}
\left\{
q\,\|\,N(\bzero,\bI)
\right\}
=
\frac{1}{2}
\sum_{r}
\left\{
\mu_r^2+\sigma_r^2-\log(\sigma_r^2)-1
\right\}.
\label{eq:hier_kl}
\end{equation}
Using this closed-form expression avoids the need for Monte Carlo estimation of the KL terms, which reduces variance in the gradient estimates and stabilizes training.


\subsubsection{Estimation of the Spatially Varying Coefficients}

After training, uncertainty in the estimated coefficient surfaces is obtained by Monte Carlo sampling through both levels of the latent representation and the coefficient-specific decoders. For $m=1,\ldots,M$, we first independently generate $\bz_j^{(m)} \sim q_{\phi_j}(\bz_j\mid \by),~ j=0,\ldots,J.$ We then construct $\bz =
( \bz_0^{(m)\top},
\ldots, \bz_J^{(m)\top})^\top$ and generate the shared representation $\bz_g^{(m)} \sim q_{\phi_g}(\bz_g\mid \bz^{(m)})$

Conditional on the coefficient-specific and shared latent draws, we generate $\balpha_j^{(m)} \sim p_{\theta_j}
(\balpha_j\mid \bz_j^{(m)},\bz_g^{(m)}),~ j=0,\ldots,J.$ The corresponding Monte Carlo realization of the $j$th coefficient surface is $\bbeta_j^{(m)}(\bu) = \bB_j(\bu)^\top \balpha_j^{(m)}.$ The estimated coefficient surface is summarized by
\begin{equation}
\widehat{\beta}_j(\bu)
=
\frac{1}{M}
\sum_{m=1}^{M}
\bB_j(\bu)^\top \balpha_j^{(m)},
\qquad
j=0,\ldots,J.
\label{eq:beta_mean_hier}
\end{equation}






\subsubsection{Prediction and Conformal Calibration at New Spatial Locations}
\label{sec:prediction}

Let $\{\bu_\ell^*\}_{\ell=1}^{n^*}$ denote $n^*$ locations where prediction is sought. For the $j$th coefficient, define a prediction design matrix $\bD^{(j)*} \in \mathbb{R}^{n^* \times K_j}$ with rows
$\left[\bD^{(j)*}\right]_{\ell\cdot} = x_j(\bu_\ell^*)\,\bB_j(\bu_\ell^*)^\top, \: \ell = 1, \ldots, n^*.$
Given the $m$th Monte Carlo samples $\{\balpha_j^{(m)}\}$, for each $j$,
$\beta_j^{(m)}(\bu_\ell^*) = \bB_j(\bu_\ell^*)^\top \balpha_j^{(m)},$ 
$\bmu_*^{(m)} = \sum_{j=0}^{J} \bD^{(j)*}\balpha_j^{(m)}.$
Sample the predictive response:
\begin{equation}
    \by_*^{(m)} = \bmu_*^{(m)} + \bepsilon_*^{(m)}, \qquad
    \bepsilon_*^{(m)} \sim \mathcal{N}_{n^*}(\bzero, \tau^2\bI_{n^*}).
    \label{eq:predictive_draw_hier}
\end{equation}
The posterior predictive mean at new locations is
$\widehat{\mathbb{E}}\left(\by_* \mid \by \right) = \frac{1}{M} \sum_{m=1}^{M} \bmu_*^{(m)}.$ For a nominal coverage level $1-\gamma$, the uncalibrated posterior predictive interval at $\bu_l^*$ is obtained as the empirical $\gamma/2$ and $1-\gamma/2$ quantiles of predictive samples $\{y_{*, \ell}^{(m)}\}_{m=1}^M$, denoted $L_l^*$ and $U_l^*$ respectively.

\noindent\underline{\textbf{Split conformal calibration.}} While the Monte Carlo predictive intervals described above capture model uncertainty propagated through the latent representation and decoders, they are not guaranteed to achieve exact frequentist coverage in finite samples. To address this, we apply a \emph{split conformal prediction} procedure~\citep{vovk2005algorithmic, lei2018distribution} to calibrate the response-level predictive intervals.

The procedure works as follows. We set aside a held-out calibration set $\mathcal{I}_{\mathrm{cal}}$ of size $n_{\mathrm{cal}}$ from the training data. The GeoVAE model is estimated on the remaining observations, without using the calibration responses. For each calibration observation $i\in\mathcal{I}_{\mathrm{cal}}$, we compute the uncalibrated posterior predictive limits $L_i^{\mathrm{cal}}$ and $U_i^{\mathrm{cal}}$ using the Monte Carlo procedure described above. We then compute a \emph{nonconformity score} for each calibration observation, measuring the degree to which the uncalibrated interval fails to cover the true response:
\begin{align*}
R_i =
\max \left\{L_i^{\mathrm{cal}} - y_i,\ y_i - U_i^{\mathrm{cal}},\ 0\right\},
\quad i \in \mathcal{I}_{\mathrm{cal}}.
\end{align*}
A score of $R_i=0$ indicates that the uncalibrated interval already covers the true response $y_i$; a positive score measures the extent of the violation. The conformal correction threshold is then computed as the $k$th order statistic of the calibration scores, where
\[
k = \min \left\{ \left\lceil(n_{\mathrm{cal}}+1)(1-\gamma) \right\rceil,\, n_{\mathrm{cal}} \right\}, \qquad \widehat{q}_{1-\gamma} = R_{(k)}.
\]
Finally, the conformalized predictive interval at each new location $\bu_l^*$ is obtained by symmetrically expanding the uncalibrated interval by the conformal correction $\widehat{q}_{1-\gamma}$
\begin{equation}
    \widehat{\mathcal{C}}_{1-\gamma}(\bu_\ell^*) = \left[ L_\ell^* - \widehat{q}_{1-\gamma},\, U_\ell^* + \widehat{q}_{1-\gamma} \right].
    \label{eq:conformal_predictive_interval}
\end{equation}
The standard split conformal coverage guarantee assumes exchangeability of the calibration and test observations, which is not satisfied for spatially dependent data. However, if the model is well-specified, the errors are approximately spatially uncorrelated after conditioning on the estimated coefficient surfaces. This can be assessed via Moran's I or a variogram of the residuals. In such cases, the nonconformality scores $R_i$, which are functions of these residuals, are approximately exchangeable, and the conformal coverage guarantee holds approximately.


\section{Theoretical Study}\label{sec:theoretical_study}
We establish consistency of the GeoVAE estimator for the spatially varying coefficient surfaces and for out-of-sample prediction. For mathematical tractability, we analyze the framework without the shared encoder, retaining
only the coefficient-specific encoders and decoders, so that
$\bPhi = \{\bphi_0, \ldots, \bphi_J\}$, $\bTheta = \{\btheta_0, \ldots, \btheta_J\}$ and $\bgamma_j = \bz_j$. Furthermore, we assume that $\tau^2$ is fixed and known. Since the shared encoder hierarchically synthesizes information from the
coefficient-specific encoders, its inclusion can only improve empirical performance; the consistency results established here therefore remain valid in the presence of the shared encoder.
\subsection{Notations}
Recall from Section~2 that the true data generating process is
\begin{equation*}
y(\bu_i)= \sum_{j=0}^J x_j(\bu_i)\,\beta_j^*(\bu_i)+ \epsilon(\bu_i),
\qquad
\epsilon(\bu_i) \overset{\mathrm{iid}}{\sim} \mathcal{N}(0,\tau^2),
\quad i = 1,\ldots,N,
\end{equation*}
where $\beta_j^*(\cdot)$ denotes the true $j$th spatially varying coefficient function. In vector form, $\by = \boldf^* + \bepsilon$, where $\boldf^* \in \mathbb{R}^N$ has $i$th entry
$f_i^* = \sum_{j=0}^J x_j(\bu_i)\,\beta_j^*(\bu_i)$.
The GeoVAE basis representation approximates each $\beta_j^*(\cdot)$ via representation through the basis functions $\bB_j(\bu)$,
with best-approximation coefficient vector
\begin{equation*}
\balpha_j^{(K_j)}
=
\operatorname*{arg\,min}_{\balpha_j \in \mathbb{R}^{K_j}}
\sup_{\bu \in \mathcal{D}}
\left|\beta_j^*(\bu) - \bB_j(\bu)^\top \balpha_j\right|.
\end{equation*}
The stacked best-approximation vector is
$\balpha^{(K)}
=\bigl(
  \balpha_0^{(K_0)\top},
  \ldots,
  \balpha_J^{(K_J)\top}
\bigr)^\top,$
and the \emph{misspecification bias vector} is
\begin{equation}\label{eq:misclass}
\bg^* = \boldf^* - \bD\balpha^{(K)} \in \mathbb{R}^N,
\qquad
g_i^* = \sum_{j=0}^J x_j(\bu_i)
\Bigl(
  \beta_j^*(\bu_i) - \bB_j(\bu_i)^\top\balpha_j^{(K_j)}
\Bigr).
\end{equation}
The GeoVAE estimated coefficient surface is
\begin{equation*}
\hat{\beta}_j(\bu)
=
\frac{1}{M}\sum_{m=1}^M \bB_j(\bu)^\top \balpha_j^{(m)}
=
\bB_j(\bu)^\top \bar{\balpha}_j,
\end{equation*}
where
$\bar{\balpha}_j
=
M^{-1}\sum_{m=1}^M \balpha_j^{(m)},$ and
$\balpha_j^{(m)} \sim p_{\btheta_j}\bigl(\balpha_j \mid \bgamma_j^{(m)}\bigr)$,
are draws from the $j$th coefficient-specific decoder. We define
$\tilde{\balpha}_j
=
\mathbb{E}_{q_{\hat\bPhi},\, p_{\hat\btheta_j}}
\!\left[\balpha_j^{(1)}\right]$ as
the posterior mean of the decoded basis coefficients under the trained parameters
$(\hat\bPhi,\hat\bTheta)$. The global design matrix $\bD \in \mathbb{R}^{N \times K}$ and the per-coefficient design
matrices $\bD^{(j)} \in \mathbb{R}^{N \times K_j}$ are defined as in Section~\ref{sec:spatial_basis_expansion}.

\subsection{Assumptions}

\noindent\textbf{Assumption 1 (Spatial domain and observation locations).}
The domain $\mathcal{D} \subset \mathbb{R}^d$ is compact. The observation locations $\{\bu_i\}_{i=1}^N$ are drawn i.i.d.\ from a probability measure $\nu$ on $\mathcal{D}$ that
is absolutely continuous with respect to Lebesgue measure, with density satisfying
$\inf_{\bu \in \mathcal{D}} f_\nu(\bu) \geq c_\nu > 0$.

\medskip
\noindent\textbf{Assumption 2 (Sobolev smoothness and basis approximation quality).}
For each $j = 0,1,\ldots,J$, the true coefficient function $\beta_j^*(\cdot) \in \mathcal{H}_j$,
where $\mathcal{H}_j$ is a Sobolev space of smoothness $s_j > d/2$ over $\mathcal{D}$.
The basis system $\{B_{jk}(\cdot)\}_{k=1}^{K_j}$ satisfies the approximation property:
there exist $\balpha_j^{(K_j)} \in \mathbb{R}^{K_j}$ and a constant $C_j > 0$ such that
\[
\sup_{u \in \mathcal{D}}
\bigl|\beta_j^*(\bu) - \bB_j(\bu)^\top \balpha_j^{(K_j)}\bigr|
\;\leq\;
C_j K_j^{-s_j/d}.
\]

\medskip
\noindent\textbf{Assumption 3 (Uniform boundedness of basis functions).}
There exists $C_B < \infty$ such that
$\sup_{\bu \in \mathcal{D}} |B_{jk}(\bu)| \leq C_B,$
for all $j$ and $k$.

\medskip
\noindent\textbf{Assumption 4 (Design matrix eigenvalue condition).}
Let $\bSigma_N = N^{-1} \bD^\top \bD$. There exist constants
$0 < \tilde\lambda_1 \leq \tilde\lambda_2 < \infty$, not depending on $N$, such that
almost surely
\[
\tilde\lambda_1 \leq \lambda_{\min}(\bSigma_N)
\;\leq\;
\lambda_{\max}(\bSigma_N) \leq \tilde\lambda_2.
\]

\medskip
\noindent\textbf{Assumption 5 (Bounded parameters).}
All encoder and decoder parameters are contained in a compact set:
$\|\bphi_j\|_2 \leq R_\phi
\quad\text{and}\quad
\|\btheta_j\|_2 \leq R_\theta,$
for all $j=0,...,J$ and all $N$.

\medskip
\noindent\textbf{Assumption 6 (Reference configuration).}
There exists a parameter configuration $(\bPhi^0, \bTheta^0)$ with
$\|\bphi_j^0\|_2 \leq R_\phi$ and $\|\btheta_j^0\|_2 \leq R_\theta$ such that the following
three conditions hold:

\smallskip
\noindent\emph{\textbf{(A6.1)}} For each $j = 0,\ldots,J$, the reference decoder mean and variance satisfy
$\bmu_{\btheta_j^0}(\bgamma_j)= \balpha_j^{(K_j)},$
$\sup_{\bgamma_j \in \mathcal{C}}
\bigl\|\bsigma_{\theta_j^0}(\bgamma_j)\bigr\|_2
\;\leq\; \delta_N,$
where $\mathcal{C} \subset \mathbb{R}^{d_j}$ is a compact set containing the support of
$\bgamma_j$ under $q_{\bPhi^0}$, and $\delta_N \to 0$ as $N \to \infty$. 

\smallskip
\noindent\emph{\textbf{(A6.2)}} The reference encoder $\bPhi^0$ satisfies
$\mathrm{KL}\!\left\{
  q_{\bphi_j^0}(\bz_j \mid \by)\,\Big\|\,\mathcal{N}(\bzero, \bI_{d_j})
\right\}
\;\leq\; C_{\mathrm{KL}} < \infty,\:\:j=0,...,J,$
uniformly over the training data, for some constant $C_{\mathrm{KL}} > 0$.

\smallskip
\noindent\emph{\textbf{(A6.3)}} The trained parameters $(\hat\bPhi, \hat\bTheta)$ obtained by
minimizing the GeoVAE ELBO objective $\mathcal{L}(\bPhi,\bTheta)$ satisfy $\frac{1}{N}\mathcal{L}(\hat\bPhi, \hat\bTheta)
\;\leq\;
\frac{1}{N}\mathcal{L}(\bPhi^0, \bTheta^0) + \eta_N,$
where $\eta_N \geq 0$ satisfies $\eta_N \to 0$ as $N \to \infty$.

\medskip
\noindent\textbf{Assumption 7 (Bounded covariates).} The covariates are bounded over the domain, i.e., $\sup_{\bu\in\mathcal{D}}|x_j(\bu)|<C_x<\infty$, for $j=0,...,J.$

\noindent The seven assumptions underlying the consistency theory of GeoVAE are standard regularity conditions that are broadly satisfied in spatial statistics applications and are individually mild in nature. Compactness of the domain $\mathcal{D}$ in \textbf{Assumption 1} is a standard requirement in nonparametric function estimation over spatial regions and is satisfied by any bounded geographic study area, while the positivity condition on the sampling density ensures that no subregion of the domain is systematically unobserved, preventing the estimator from extrapolating into regions with no data support. \textbf{Assumption 2} imposes that each coefficient function is at least continuous, which is a minimal requirement for spatially varying effects to be interpretable \citep{van2011information}, and the approximation rate is the classical rate achieved by B-splines, radial basis functions, and thin-plate splines for functions in Sobolev spaces \citep{shen2015adaptive}, making this assumption verifiable through the choice of basis. \textbf{Assumption 3} is immediately satisfied by all standard basis systems including B-splines and compactly supported radial basis functions, since their evaluations are bounded by construction over a compact domain $\mathcal{D}$. \textbf{Assumption 4} imposes a bounded eigenvalue condition on the empirical design matrix which ensures that the design is neither rank-deficient nor ill-conditioned asymptotically. \textbf{Assumption 5} restricts the encoder and decoder parameters to a compact set by bounding their Euclidean norms. This ensures that the neural network outputs remain bounded, which is necessary for the law of large numbers arguments used in the consistency proofs. \textbf{Assumption 6} is the key reference configuration assumption, which postulates the existence of a parameter configuration under which the decoder mean recovers the best-approximation coefficient vector, the decoder variance shrinks to zero, the encoder KL divergences remain bounded, and the trained parameters achieve a loss no larger than that of the reference configuration up to a vanishing slack $\eta_N$. The first two conditions in Assumption 6 are guaranteed by the universal approximation capacity \citep{hanin2019universal} of sufficiently wide and deep neural networks, which can approximate any continuous target function to arbitrary precision; the bounded KL condition reflects the fact that a well-initialized encoder does not need to deviate far from the prior to encode relevant information. Since $(\hat\bTheta, \hat\bPhi)$ is the minimizer of the ELBO over the compact parameter space, by definition
$\frac{1}{N}\mathcal{L}(\hat\bTheta, \hat\bPhi)\leq \frac{1}{N}\mathcal{L}(\bTheta^0, \bPhi^0)$,
which means \textbf{(A6.3)} holds with $\eta_N=0$ if exact minimization is achieved. In practice, optimization via stochastic gradient descent finds an approximate minimizer, so a small slack $\eta_N\geq 0$ is permitted, and the condition $\eta_N\rightarrow 0$ simply requires that this optimization error vanishes asymptotically, which is a standard assumption in the sieve estimation and neural network consistency literature \citep{barron1994approximation}. Collectively, these assumptions impose no more structure than is standard in the nonparametric regression and neural network consistency literatures, and together they provide a complete and coherent set of conditions under which the GeoVAE estimator is consistent for the true spatially varying coefficient surfaces.

\subsection{Main Consistency Results}
This section discusses the consistency results. The first results discuss consistency of coefficient surface estimation. The second result is on prediction accuracy by the proposed model. The proofs of both results can be found in the Appendix.
\begin{theorem}\label{lem:param_consistency}
Let $\widehat{\beta}_j(\bu)=\bB_j(\bu)^\top\bar{\balpha}_j$, where $\bar{\balpha}_j=\frac{1}{M}\sum_{m=1}^M \balpha_j^{(m)}$, with $\balpha_j^{(m)}$ drawn from the decoder 
$p_{\hat{\btheta}_j}(\balpha_j|\bgamma_j^{(m)})$.
Under Assumptions 1--7, as $N\rightarrow\infty$, with $K_j=K_{j,N}=O(N^{d/(2s_j+d)})$ and $M=M_N\geq N^{2d/(2s_j+d)}$, we have 
\begin{align*}
 E_{\nu}[(\widehat{\beta}_j(\bu)-\beta_j^*(\bu))^2]\stackrel{P}{\rightarrow} 0,\:\:\mbox{for}\:\:j=1,...,J. 
\end{align*}
\end{theorem}

\begin{theorem}\label{lemma: prediction}
Let $\bu^*$ be a new spatial location drawn independently from $\nu$. Under Assumptions~1--7, with $K_{j,N} = O\!\bigl(N^{d/(2s_j+d)}\bigr)$ and
$M_N \geq N^{2d/(2s_j+d)}$, we have as $N\rightarrow\infty$
\[
E_\nu\!\left[\bigl(\hat{y}(\bu^*) - \mu^*(\bu^*)\bigr)^2\right]
\stackrel{P}{\rightarrow} 0,
\]
where
$\mu^*(\bu^*) = \sum_{j=0}^J x_j(\bu^*)\,\beta_j^*(\bu^*)$
is the true conditional mean and $\hat{y}(\bu^*) = \sum_{j=0}^J x_j(\bu^*)\,\hat{\beta}_j(\bu^*)$ is the estimated conditional mean.
\end{theorem}

\section{Simulation Study}\label{sec:simulation}

We designed a simulation study to systematically assess the performance of GeoVAE in recovering multiple correlated spatially varying coefficient surfaces and predictive inference across a range of spatial smoothness regimes for true varying coefficient surfaces. GeoVAE is implemented in Python using \texttt{PyTorch}. Each coefficient surface is represented by a tensor-product cubic B-spline basis with eight interior knots per coordinate, yielding 144 basis functions. Coefficient-specific convolutional encoders produce 8-dimensional local latent vectors, which are combined by a global encoder into a 16-dimensional shared representation. Separate decoders output Gaussian means and log-variances for the basis coefficients. The model is trained using full-batch Adam with a learning rate of $10^{-4}$ for up to 3000 epochs, with early stopping based on a held-out validation set comprising $15\%$ of the training data (patience: 250 epochs). Posterior summaries use 500 Monte Carlo draws, and $95\%$ predictive intervals are calibrated using split conformal prediction.


\noindent\underline{\textbf{Competitor.}} As competitors, we consider five methods spanning nonspatial deep generative models, spatial neural networks, scalable Gaussian process approximations, and Bayesian adaptive spline models. First, we include a nonspatial conditional variational autoencoder (CVAE), following the conditional generative framework of \citet{sohn2015learning}, but excluding spatial coordinates and spatial basis representations. This comparison isolates the contribution of explicitly incorporating spatial structure into the proposed GeoVAE. DeepKriging \citep{chen2024deepkriging} represents a spatial deep learning approach that embeds spatial coordinates through multiresolution basis functions and uses a deep neural network for nonlinear spatial prediction. We also consider two scalable Gaussian process competitors. The local approximate Gaussian process (LaGP)  constructs a separate local GP approximation around each prediction location using adaptively selected neighborhoods \citep{gramacy2015local}, whereas the Vecchia GP approximates the joint Gaussian likelihood through a product of low-dimensional conditional distributions \citep{vecchia1988estimation,katzfuss2021general}. Finally, Bayesian adaptive spline surfaces (BASS) provide a fully Bayesian nonparametric regression benchmark based on adaptively selected spline basis functions \citep{francom2020bass}.

\subsection{Simulation Design}
\label{sec:simulation_design}

We consider the spatial domain $\mathcal{D} = [0,1]^2$ and the sample size $N=3000$, and let 
$\{\bu_i = (u_{i1}, u_{i2})^\top : i = 1,\ldots, N\}$ 
denote a regular grid on $\mathcal{D}$. Data are generated from the varying coefficient model \eqref{eq:VCM_basic}, where $\bx(\bu_i) = (x_{i1}, x_{i2}, x_{i3})^\top$ and $\bbeta(\bu_i) = (\beta_1(\bu_i), \beta_2(\bu_i), \beta_3(\bu_i))^\top$. We fix $J=3$ predictor-specific coefficient surfaces. The predictors are generated independently as $x_{ij} \overset{\mathrm{iid}}{\sim} N(0,1)$ for $j=1,2,3$, and the measurement errors are generated as 
$\epsilon_i \overset{\mathrm{iid}}{\sim} N(0,\tau^2)$, independently of the covariates and coefficient processes, with $\tau = 0.25$.

\noindent\underline{\textbf{Simulating the spatially varying intercept $\beta_0(\bu)$.}}
The spatially varying intercept $\beta_0(\bu)$ is generated independently of the slope processes from a 
univariate Gaussian process with a Matern covariance kernel
$\beta_0(\bu) \sim \operatorname{GP}\big\{\mu_0, \sigma_0^2 \mathcal{M}(\|\bu - \bu'\|; \nu_0, \kappa_0)\big\},$
where $\mu_0 = 1$, $\sigma_0 = 0.5$, and $\nu_0 = 1.5$, ensuring that the intercept surface is once mean-square differentiable. The Matérn correlation function is given by
\begin{equation}
\mathcal{M}(h; \nu, \kappa)
= \frac{2^{1-\nu}}{\Gamma(\nu)} (\kappa h)^\nu K_\nu(\kappa h), 
\qquad h = \|\bu - \bu'\|,
\label{eq:matern_correlation}
\end{equation}
where $K_\nu(\cdot)$ denotes the modified Bessel function of the second kind, $\nu > 0$ controls local smoothness, and $\kappa > 0$ governs the rate of spatial correlation decay. We set $\kappa_0 = 11.11$, so that the spatial correlation of the intercept process decreases to 0.10 at a distance of $0.35$. A single realization of 
$\beta_0(\bu)$ is generated and then held fixed across all simulation cases, so that differences in model performance can be attributed primarily to the properties of the predictor-specific coefficient surfaces.

\noindent\underline{\textbf{Simulating the trivariate slope process $\bbeta(\bu)$.}}
The three predictor-specific coefficient surfaces are generated as
\begin{equation}
\bbeta(\bu) 
= \bmu_{\beta} + \bw(\bu),
\qquad
\bmu_{\beta} = (0.60, -0.40, 0.30)^\top,
\label{eq:sim_slope_process}
\end{equation}
where $\bw(\bu) = (w_1(\bu), w_2(\bu), w_3(\bu))^\top$ is a mean-zero multivariate Gaussian process, and the 
marginal standard deviations are fixed at $(\sigma_1, \sigma_2, \sigma_3) = (0.70, 0.55, 0.60)$. We construct 
five simulation cases, each involving the same observation model but different covariance structures for $\bw(\bu)$, in order to assess the impact of spatial smoothness and cross-coefficient dependence on model performance. A brief summary of the simulation scenarios is given in Table~\ref{tab:smoothness_scenarios}.

\noindent\underline{\emph{\textbf{Cases 1--3}: Common smoothness, increasing regularity.}}
In Cases 1--3, the three slope processes are generated together from a GP with a multivariate Matérn covariance kernel given by
\begin{equation}
\operatorname{Cov}\{w_j(\bu), w_k(\bu')\}
= \sigma_j \sigma_k \rho_{jk} \, \mathcal{M}(h; \nu_{jk}, \kappa),
\qquad
\nu_{jk} = \frac{\nu_j + \nu_k}{2},
\qquad j,k \in \{1,2,3\},
\label{eq:parsimonious_multivariate_matern}
\end{equation}
with common inverse-range parameter $\kappa$. They share the same smoothness parameter $\nu=\nu_1=\nu_2=\nu_3$, and the pairwise collocated correlation between any two slope processes is fixed at $\rho_{jk}=0.65$, inducing nonzero dependence among the three predictor-specific coefficient surfaces. To isolate the effect of smoothness from spatial range, $\kappa$ is 
chosen in each case to satisfy
\begin{equation}
\mathcal{M}(r_0; \nu, \kappa) = 0.10,
\qquad r_0 = 0.35,
\label{eq:matched_practical_range}
\end{equation}
so that the spatial correlation of each slope process decays to $0.10$ at distance $r_0$. Thus, the marginal variances, the collocated correlations between processes, and the practical correlation range are fixed across Cases 1--3, and only the smoothness parameter $\nu$ varies, yielding rough (exponential), moderately smooth (once mean-square differentiable), and more smooth (twice mean-square differentiable) coefficient surfaces respectively.

\noindent\underline{\emph{\textbf{Case 4:} Heterogeneous smoothness across coefficient surfaces.}}
Case 4 evaluates the ability of the coefficient-specific branches to recover three correlated predictor effects with 
different smoothness levels. Here, the three slope processes are modeled via the parsimonious multivariate Matérn 
construction,
with a common inverse-range parameter $\kappa$. We set 
$(\nu_1, \nu_2, \nu_3) = (0.5, 1.5, 2.5)$, so that the three coefficient surfaces exhibit distinctly different 
smoothness levels. Define
$c_d(\nu) = \frac{\Gamma(\nu + d/2)}{\pi^{d/2} \Gamma(\nu)}.$
Using a positive definite base correlation matrix $\bR$ with unit diagonal and off-diagonal 
entries $0.65$, the collocated cross-correlations are constructed as
\begin{equation}
\rho_{jk}
= R_{jk}
\frac{\{c_d(\nu_j) c_d(\nu_k)\}^{1/2}}{c_d(\nu_{jk})},
\qquad j,k \in \{1,2,3\}.
\label{eq:valid_cross_correlation}
\end{equation}
With the chosen smoothness parameters, 
the resulting collocated correlations are approximately $\rho_{12} = 0.563$, $\rho_{13} = 0.484$, and 
$\rho_{23} = 0.629$, ensuring that all three coefficient surfaces remain correlated while exhibiting distinct 
degrees of spatial smoothness.

\noindent\underline{\emph{\textbf{Case 5:} Nonstationary dependence.}} To evaluate 
performance under nonstationarity, Case~5 allows the dependence structure among the 
three predictor-specific coefficient surfaces to vary across space. Following the 
spatially varying coregionalization framework of \citet{gelfand2004nonstationary}, 
slope deviations are generated as $\bw(\bu) = \bA(\bu)\bv(\bu)$, where 
$\bv(\bu) = \{v_1(\bu), v_2(\bu), v_3(\bu)\}^\top$ comprises independent, mean-zero, 
unit-variance Mat\'{e}rn Gaussian processes, and the spatially varying loading matrix 
$\bA(\bu)$ governs how these latent processes contribute to each coefficient surface, 
inducing location-dependent correlations among coefficients. Each latent process shares 
the smoothness and range specifications of Case~2, with $\nu = 1.5$ and inverse-range 
parameter $\kappa_{1.5}$ satisfying $\mathcal{M}(0.35; 1.5, \kappa_{1.5}) = 0.10$. 
To construct $A(\bu)$, let $\bL\bL^\top = \bR$, where $\bR$ has unit diagonal and 
constant off-diagonal entries of $0.65$, and define 
$\bB(\bu) = \bL + 0.5\,\bG(\bu)$, where the entries of $\bG(\bu)$ are independent, 
mean-zero, unit-variance Mat\'{e}rn Gaussian processes with smoothness $2.5$ and 
practical correlation range $0.35$, generated independently of $\bv(\bu)$. To ensure 
that each coefficient surface retains its prescribed marginal variance $\sigma_j^2$ 
at every location, the rows of $\bB(\bu)$ are normalized as
\begin{equation}
    a_{jk}(\bu) \;=\; \sigma_j\,
    \frac{b_{jk}(\bu)}
    {\left\{\displaystyle\sum_{\ell=1}^3 b_{j\ell}(\bu)^2\right\}^{1/2}},
    \qquad j,k = 1,2,3,
    \label{eq:normalization}
\end{equation}
which guarantees $\operatorname{Var}\{w_j(\bu) \mid \bA(\cdot)\} = \sigma_j^2$ 
pointwise while permitting the inter-coefficient correlations to vary spatially. 
For each simulation replicate, $\bA(\cdot)$ is drawn once and held fixed when 
sampling $\bw(\bu)$. Conditional on these loading surfaces, $\bw(\bu)$ is a 
nonstationary multivariate Gaussian process whose covariance structure depends 
on the specific pair of locations rather than solely on their separation distance. 
The five simulation scenarios are summarized in Table~\ref{tab:smoothness_scenarios}.

\begin{table}[t]
  \centering
  \caption{Simulation scenarios for the three correlated
  predictor-specific coefficient surfaces, including common and
  coefficient-specific smoothness and spatially varying dependence.}
  \label{tab:smoothness_scenarios}
  \small
  \begin{tabular}{clcccc}
    \toprule
    Case & Description
    & $\nu_1$ & $\nu_2$ & $\nu_3$ & Construction \\
    \midrule
    1 & Rough, exponential
      & $0.5$ & $0.5$ & $0.5$ & Separable Mat\'ern \\
    2 & Moderately smooth
      & $1.5$ & $1.5$ & $1.5$ & Separable Mat\'ern \\
    3 & Smooth
      & $2.5$ & $2.5$ & $2.5$ & Separable Mat\'ern \\
    4 & Coefficient-specific smoothness
      & $0.5$ & $1.5$ & $2.5$ & Parsimonious Mat\'ern \\
    5 & Spatially varying dependence
      & $1.5^{\ast}$ & $1.5^{\ast}$ & $1.5^{\ast}$
      & Spatially varying loadings \\
    \bottomrule
  \end{tabular}
\end{table}

\subsection{Simulation Results}\label{sec:sim_res}
Figure~\ref{fig:simulation-true-coefficients} displays the true spatially-varying coefficient surfaces, and Figure~\ref{fig:sim-svc-rep01} illustrates how well the proposed GeoVAE framework recovers them across simulation scenarios. Under Case $C_1$, where the true surfaces are non-differentiable with minute local variations, GeoVAE achieves root mean squared errors (RMSE) of $0.373, 0.379, 0.369$, and $0.368$ for $\beta_0(\bu), \beta_1(\bu), \beta_2(\bu)$ and $\beta_3(\bu)$, respectively. As smoothness increases across Cases $C_2$ and $C_3$, RMSE decreases steadily, confirming that smoothness is the primary determinant of estimation accuracy. Case $C_4$, which assigns heterogeneous smoothness levels across coefficients, yields the best recovery for the smoothest surface $\beta_3(\bu)$ while exhibiting progressively inferior recovery for $\beta_2(\bu)$ and $\beta_1(\bu)$ as their smoothness decreases, a pattern entirely consistent with theoretical expectations. Case $C_5$ presents the most challenging setting, introducing nonstationary cross-correlations among the varying coefficients. GeoVAE handles this complexity remarkably well, owing to its flexible, assumption-lean architecture that imposes neither stationarity nor spatially invariant correlation structures, assumptions commonly adopted in existing methods purely for computational convenience. Since Case $C_5$ also features smooth coefficient surfaces, it achieves strong recovery despite the added modeling complexity, further emphasizing that smoothness drives performance even under nonstationarity. Compared to BASS, GeoVAE achieves approximately a two-fold reduction in RMSE for smooth coefficient settings and a $1.5$-fold reduction for the rough surfaces of Case $C_1$, as reported in Table~\ref{tab:simulation_results}, demonstrating consistent and substantial gains in estimation accuracy across all simulation scenarios.

\begin{table*}[t]
\centering
\caption{
Simulation results for spatially varying coefficient recovery,
response prediction, uncertainty quantification, and computation time.
}
\label{tab:simulation_results}
\resizebox{\textwidth}{!}{
\begin{tabular}{llcccccccc}
\toprule
& &
\multicolumn{4}{c}{SVC recovery: RMSE}
&
\multicolumn{3}{c}{Response prediction and UQ}
&
\\
\cmidrule(lr){3-6}
\cmidrule(lr){7-9}
Case
& Method
& $\beta_0(\bs)$
& $\beta_1(\bs)$
& $\beta_2(\bs)$
& $\beta_3(\bs)$
& $y$ RMSE
& Coverage
& Interval width
& Time (min)
\\
\midrule

\multirow{6}{*}{$C_1$}
& GeoVAE
& 0.342 (0.024) & 0.361 (0.018)
& 0.364 (0.028) & 0.348 (0.027)
& \textbf{0.874} (0.057) & 0.953 (0.017)
& \textbf{3.502} (0.207) & 0.120 (0.030) \\
& Nonspatial CVAE
& -- & -- & -- & --
& 1.288 (0.051) & 0.955 (0.013)
& 4.826 (0.150) & 0.040 (0.010) \\
& DeepKriging
& -- & -- & -- & --
& 1.190 (0.054) & 0.951 (0.013)
& 5.048 (0.413) & 0.240 (0.040) \\
& LaGP
& -- & -- & -- & --
& 1.235 (0.070) & 0.954 (0.016) & 4.584 (0.176) & 0.240 (0.010) \\
& Vecchia GP
& -- & -- & -- & --
& 1.201 (0.035) & 0.951 (0.012)
& 5.110 (0.341) & 0.110 (0.030) \\
& BASS
& 0.450 (0.038) & 0.561 (0.052)
& 0.570 (0.035) & 0.570 (0.052)
& 1.088 (0.048) & 0.956 (0.011)
& 4.632 (0.300) & 0.080 (0.010) \\
\midrule

\multirow{6}{*}{$C_2$}
& GeoVAE
& 0.199 (0.022) & 0.222 (0.017)
& 0.211 (0.016) & 0.210 (0.018)
& \textbf{0.665} (0.031) & 0.941 (0.014)
& \textbf{2.519} (0.156) & 0.150 (0.040) \\
& Nonspatial CVAE
& -- & -- & -- & --
& 1.302 (0.045) & 0.952 (0.008)
& 4.848 (0.184) & 0.040 (0.010) \\
& DeepKriging
& -- & -- & -- & --
& 1.087 (0.051) & 0.951 (0.023)
& 4.758 (0.524) & 0.270 (0.040) \\
& LaGP
& -- & -- & -- & --
& 1.198 (0.050) & 0.949 (0.009) & 4.389 (0.231) & -- \\
& Vecchia GP
& -- & -- & -- & --
& 1.176 (0.053) & 0.952 (0.015)
& 4.937 (0.353) & 0.100 (0.020) \\
& BASS
& 0.360 (0.055) & 0.477 (0.067)
& 0.507 (0.089) & 0.503 (0.062)
& 0.914 (0.081) & 0.956 (0.014)
& 3.860 (0.398) & 0.090 (0.010) \\
\midrule

\multirow{6}{*}{$C_3$}
& GeoVAE
& 0.166 (0.011) & 0.190 (0.013)
& 0.179 (0.012) & 0.178 (0.007)
& \textbf{0.610} (0.025) & 0.949 (0.022)
& \textbf{2.423} (0.140) & 0.150 (0.030) \\
& Nonspatial CVAE
& -- & -- & -- & --
& 1.287 (0.044) & 0.953 (0.010)
& 4.887 (0.275) & 0.040 (0.010) \\
& DeepKriging
& -- & -- & -- & --
& 1.095 (0.057) & 0.954 (0.013)
& 4.789 (0.206) & 0.240 (0.040) \\
& LaGP
& -- & -- & -- & --
& 1.194 (0.073) & 0.953 (0.009) & 4.380 (0.291) & 0.250 (0.010) \\
& Vecchia GP
& -- & -- & -- & --
& 1.161 (0.048) & 0.955 (0.011)
& 5.064 (0.352) & 0.100 (0.020) \\
& BASS
& 0.316 (0.046) & 0.486 (0.075)
& 0.521 (0.076) & 0.460 (0.021)
& 0.891 (0.074) & 0.950 (0.018)
& 3.596 (0.248) & 0.090 (0.010) \\
\midrule

\multirow{6}{*}{$C_4$}
& GeoVAE
& 0.210 (0.017) & 0.408 (0.023)
& 0.238 (0.017) & 0.188 (0.012)
& \textbf{0.738} (0.029) & 0.960 (0.008)
& \textbf{3.040} (0.165) & 0.140 (0.040) \\
& Nonspatial CVAE
& -- & -- & -- & --
& 1.259 (0.060) & 0.958 (0.014)
& 4.926 (0.245) & 0.040 (0.010) \\
& DeepKriging
& -- & -- & -- & --
& 1.103 (0.055) & 0.958 (0.012)
& 4.883 (0.314) & 0.240 (0.040) \\
& LaGP
& -- & -- & -- & --
& 1.189 (0.078) & 0.961 (0.013) & 4.574 (0.276) & 0.250 (0.010) \\
& Vecchia GP
& -- & -- & -- & --
& 1.136 (0.053) & 0.960 (0.014)
& 5.072 (0.286) & 0.140 (0.100) \\
& BASS
& 0.346 (0.040) & 0.592 (0.049)
& 0.539 (0.071) & 0.461 (0.078)
& 0.974 (0.059) & 0.958 (0.014)
& 4.142 (0.375) & 0.080 (0.010) \\
\midrule

\multirow{6}{*}{$C_5$}
& GeoVAE
& 0.136 (0.005) & 0.200 (0.017)
& 0.180 (0.014) & 0.178 (0.013)
& \textbf{0.434} (0.020) & 0.966 (0.016)
& \textbf{1.902} (0.119) & 0.220 (0.060) \\
& Nonspatial CVAE
& -- & -- & -- & --
& 1.152 (0.069) & 0.958 (0.011)
& 4.384 (0.257) & 0.040 (0.010) \\
& DeepKriging
& -- & -- & -- & --
& 1.003 (0.075) & 0.947 (0.018)
& 4.101 (0.338) & 0.270 (0.040) \\
& LaGP
& -- & -- & -- & --
& 1.091 (0.066) & 0.949 (0.016) & 4.010 (0.303) & 0.240 (0.020) \\
& Vecchia GP
& -- & -- & -- & --
& 1.060 (0.074) & 0.949 (0.015)
& 4.333 (0.342) & 0.130 (0.090) \\
& BASS
& 0.279 (0.022) & 0.569 (0.091)
& 0.495 (0.140) & 0.495 (0.050)
& 0.836 (0.036) & 0.950 (0.014)
& 3.422 (0.219) & 0.080 (0.010) \\
\bottomrule
\end{tabular}
}
\end{table*}
Table~\ref{tab:simulation_results} presents predictive inference for GeoVAE in comparison with the competing approaches across all simulation scenarios. GeoVAE achieves superior root mean squared prediction error over all competing methods in every scenario considered. While the performance gap in point prediction is relatively modest for Case $C_1$, which features rough spatially varying coefficient surfaces, the advantage of GeoVAE becomes increasingly pronounced for more smoothly varying coefficient surfaces, irrespective of complex nonstationarity (Cases $C_2$, $C_3$ and $C_5$). This pattern can plausibly be attributed to the hierarchical encoder-decoder architecture of GeoVAE, which constructs a dedicated latent representation for each coefficient function and integrates cross-coefficient information through a shared synthesis layer.

The relative underperformance of the competing methods can be understood in terms of their respective structural assumptions. Nonspatial VAE, while computationally efficient and expressive as a deep generative model, does not explicitly account for the spatial dependence structure of the varying coefficients, limiting its ability to borrow spatial strength across locations and recover spatially coherent coefficient surfaces. DeepKriging is a powerful and scalable approach that leverages basis function representations and deep networks for spatial prediction; however, it is designed primarily for spatial interpolation of a single response surface rather than joint estimation of multiple spatially varying coefficient functions with distinct smoothness levels, which places it at a structural disadvantage in the VCM setting. Vecchia GP approximations represent a highly principled and computationally efficient class of sparse GP methods that have demonstrated impressive scalability for large spatial datasets; in the VCM context, however, the Vecchia approximation is applied independently to each coefficient surface without explicitly modeling cross-coefficient dependencies. This places Vecchia-GP in a little disadvantageous position, which is further exacerbated by the fact that its MCMC-based inference engine, while rigorous, does not benefit from the joint end-to-end optimization that drives GeoVAE's coefficient surfaces toward mutual consistency and predictive fidelity. Regarding uncertainty quantification, all competing methods achieve nominal or near-nominal empirical coverage of the 95\% predictive intervals across simulation cases, reflecting well-calibrated uncertainty estimates, a commendable property shared across the competing approaches. Nevertheless, GeoVAE produces strictly narrower 95\% predictive intervals in all simulation cases while maintaining nominal coverage, indicating that its hierarchical joint estimation strategy extracts more information from the data and translates this into sharper, better-resolved posterior predictive distributions. Finally, since the sample sizes $N$ considered in the simulation scenarios is not moderate ($N=3000$), all competing methods deliver predictive inference in under a minute, confirming their practical computational efficiency at moderate scales. Figure~\ref{fig:sim-y-rep01} displays the true and GeoVAE-estimated response surfaces alongside the corresponding 2.5\% and 97.5\% posterior predictive interval limits. The estimated surface faithfully recovers the local spatial features of the true surface, while the predictive intervals remain narrow and closely centered around the truth, reflecting well-calibrated and sharp uncertainty quantification.

\begin{figure*}[htbp]
\centering

\begin{subfigure}[t]{0.95\textwidth}
    \centering
    \includegraphics[width=\linewidth]
    {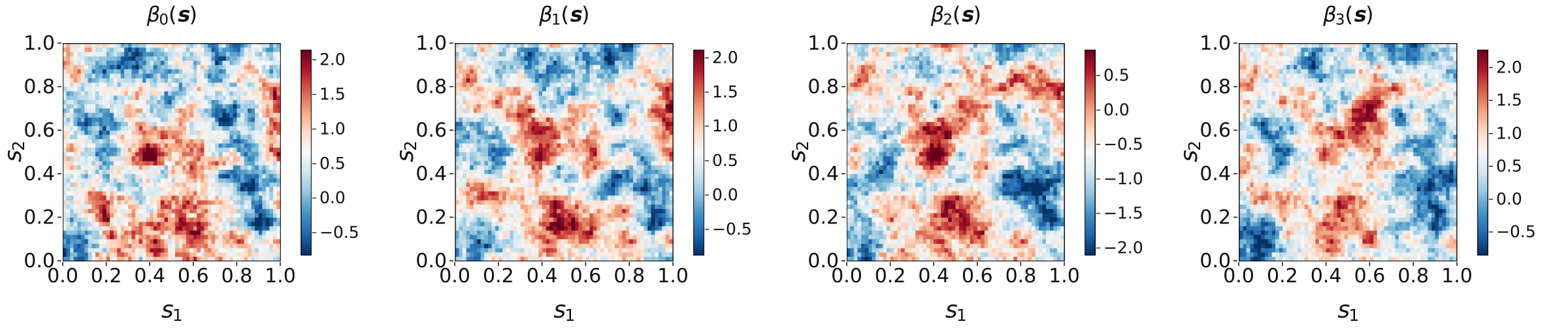}
    \caption{Case 1: Rough.}
    \label{fig:sim-true-s1}
\end{subfigure}

\smallskip

\begin{subfigure}[t]{0.95\textwidth}
    \centering
    \includegraphics[width=\linewidth]
    {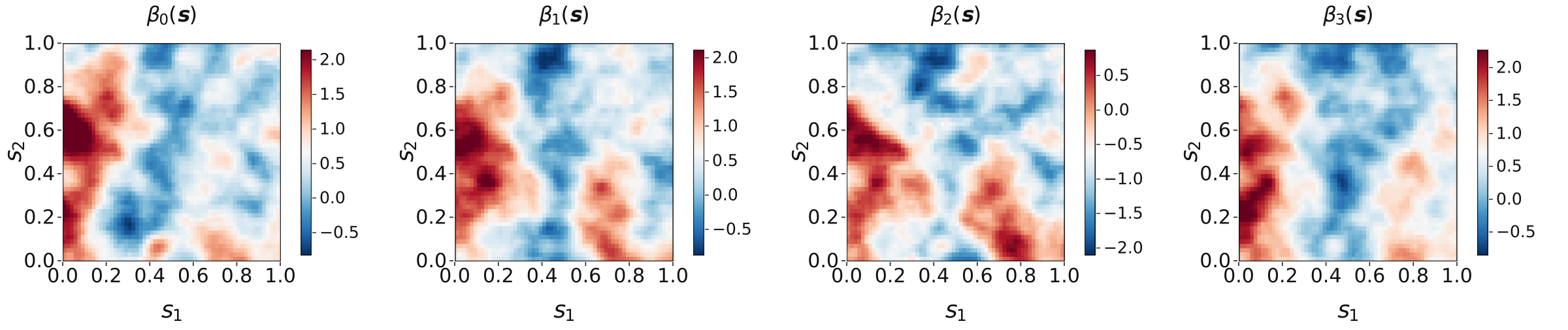}
    \caption{Case 2: Moderate.}
    \label{fig:sim-true-s2}
\end{subfigure}

\smallskip

\begin{subfigure}[t]{0.95\textwidth}
    \centering
    \includegraphics[width=\linewidth]
    {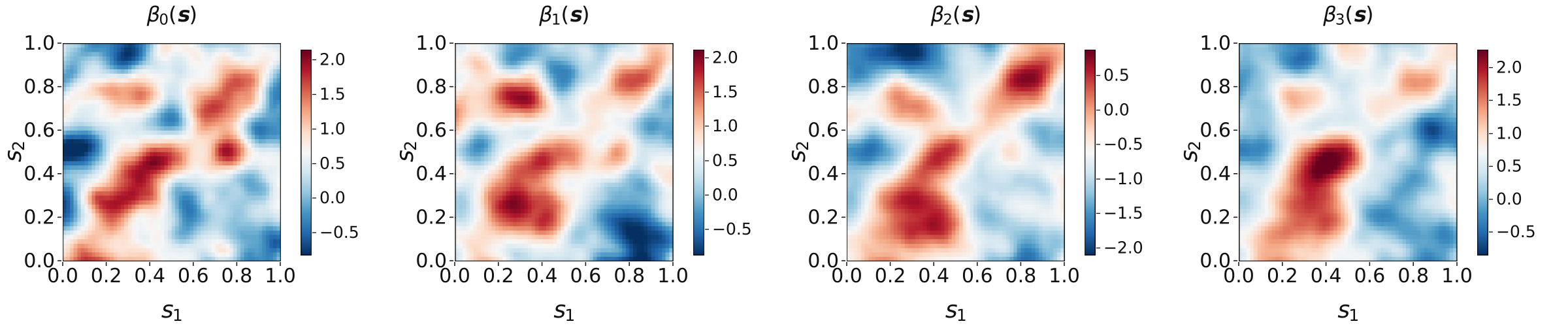}
    \caption{Case 3: Smooth.}
    \label{fig:sim-true-s3}
\end{subfigure}

\smallskip

\begin{subfigure}[t]{0.95\textwidth}
    \centering
    \includegraphics[width=\linewidth]
    {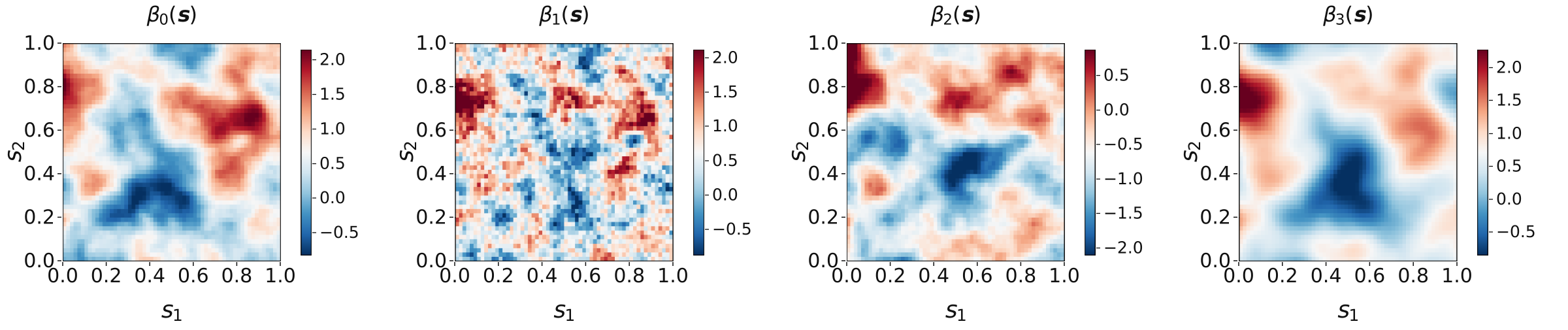}
    \caption{Case 4: Heterogeneous.}
    \label{fig:sim-true-s4}
\end{subfigure}

\smallskip

\begin{subfigure}[t]{0.95\textwidth}
    \centering
    \includegraphics[width=\linewidth]
    {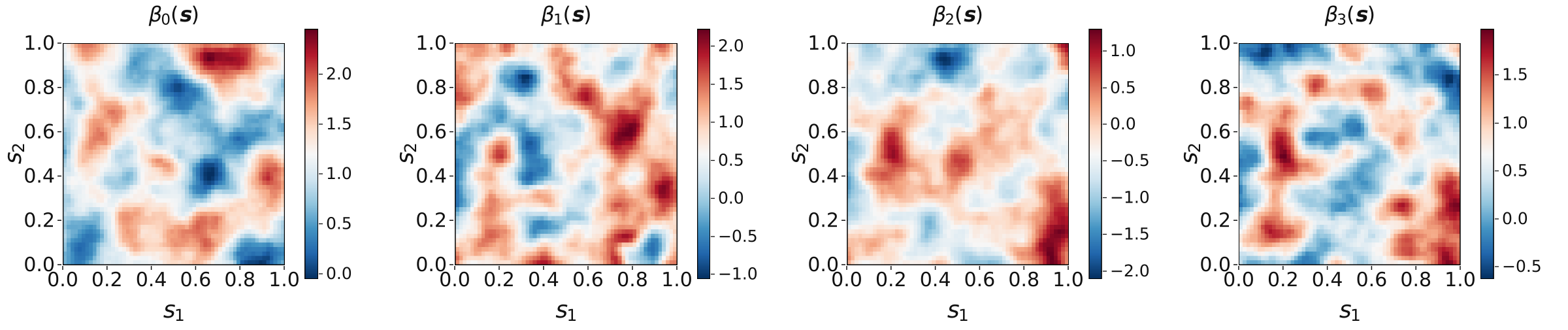}
    \caption{Case 5: Nonstationary.}
    \label{fig:sim-true-s5}
\end{subfigure}

\caption{
True spatially varying coefficient surfaces used to generate
the simulated data under the five scenarios.
}
\label{fig:simulation-true-coefficients}
\end{figure*}

\begin{figure}[htbp]
\centering
\captionsetup{font=small,skip=3pt}
\captionsetup[subfigure]{font=small,skip=1pt}
\begin{subfigure}[t]{\linewidth}
  \centering
  \includegraphics[width=\linewidth]{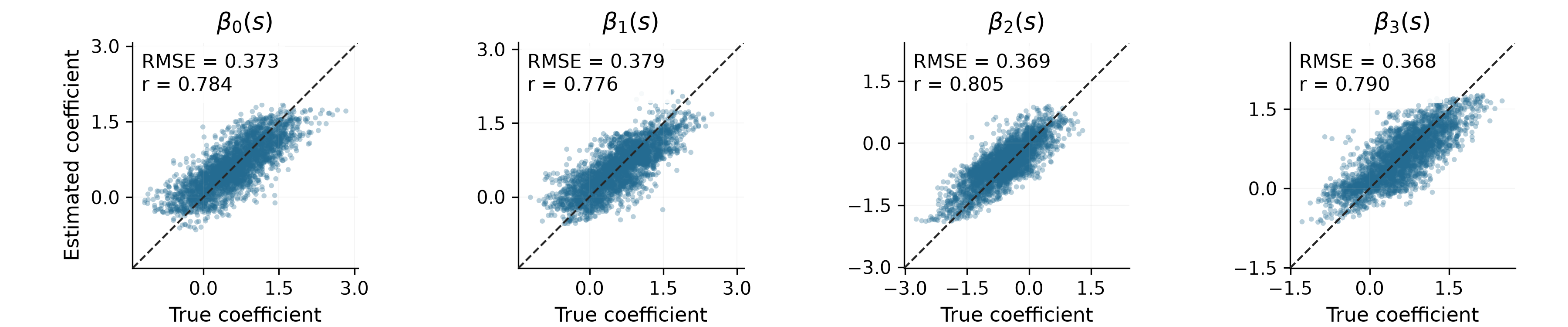}
  \caption{Case $C_1$: rough.}
  \label{fig:sim-svc-c1-rep01}
\end{subfigure}
\par\vspace{1mm}
\begin{subfigure}[t]{\linewidth}
  \centering
  \includegraphics[width=\linewidth]{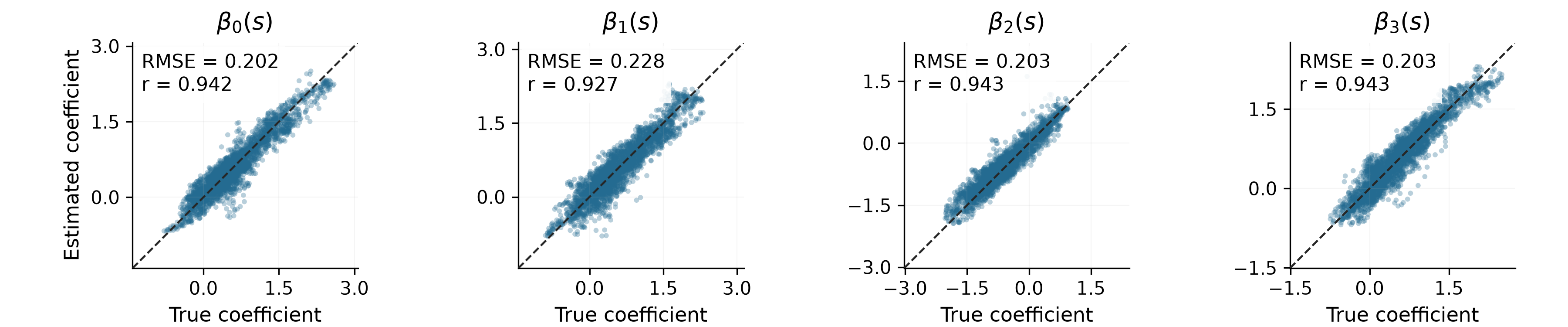}
  \caption{Case $C_2$: moderate.}
  \label{fig:sim-svc-c2-rep01}
\end{subfigure}
\par\vspace{1mm}
\begin{subfigure}[t]{\linewidth}
  \centering
  \includegraphics[width=\linewidth]{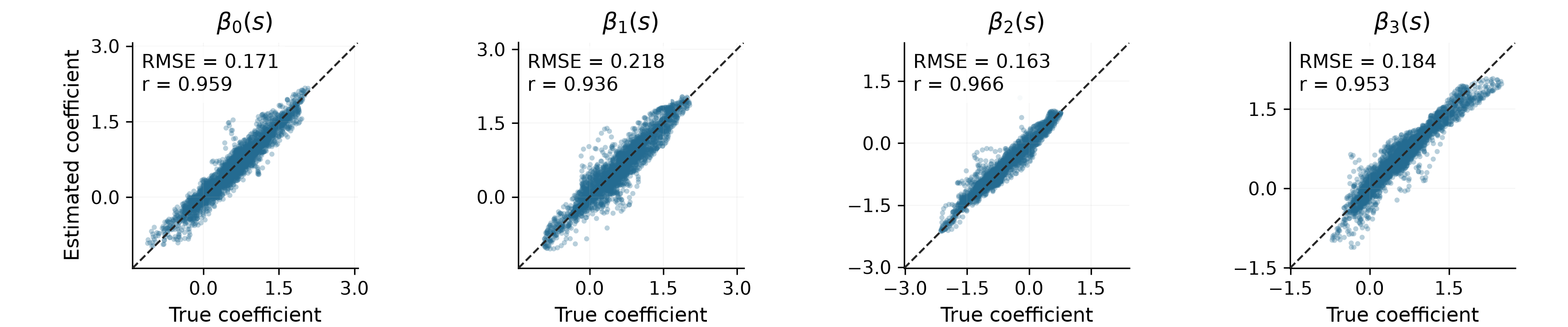}
  \caption{Case $C_3$: smooth.}
  \label{fig:sim-svc-c3-rep01}
\end{subfigure}
\par\vspace{1mm}
\begin{subfigure}[t]{\linewidth}
  \centering
  \includegraphics[width=\linewidth]{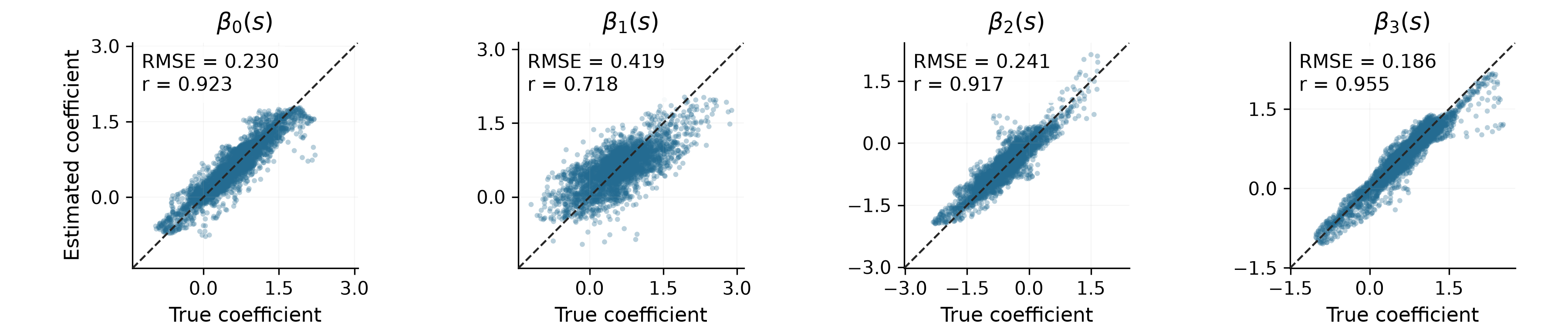}
  \caption{Case $C_4$: heterogeneous.}
  \label{fig:sim-svc-c4-rep01}
\end{subfigure}
\par\vspace{1mm}
\begin{subfigure}[t]{\linewidth}
  \centering
  \includegraphics[width=\linewidth]{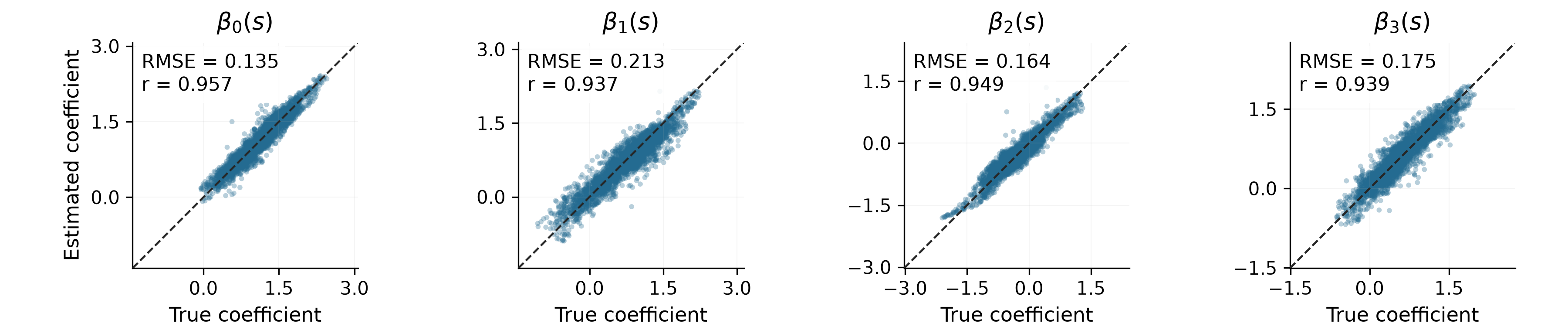}
  \caption{Case $C_5$: nonstationary.}
  \label{fig:sim-svc-c5-rep01}
\end{subfigure}
\caption{Recovery of spatially varying coefficients in five simulation scenarios. Columns correspond to $\beta_0(s)$ through $\beta_3(s)$. Dashed lines indicate equality between true and estimated values. RMSE and Pearson correlation ($r$) use the original saved values.}
\label{fig:sim-svc-rep01}
\end{figure}

\begin{figure}[htbp]
\centering
\captionsetup{font=small,skip=3pt}
\captionsetup[subfigure]{font=small,skip=1pt}
\begin{subfigure}[t]{\linewidth}
  \centering
  \includegraphics[width=\linewidth]{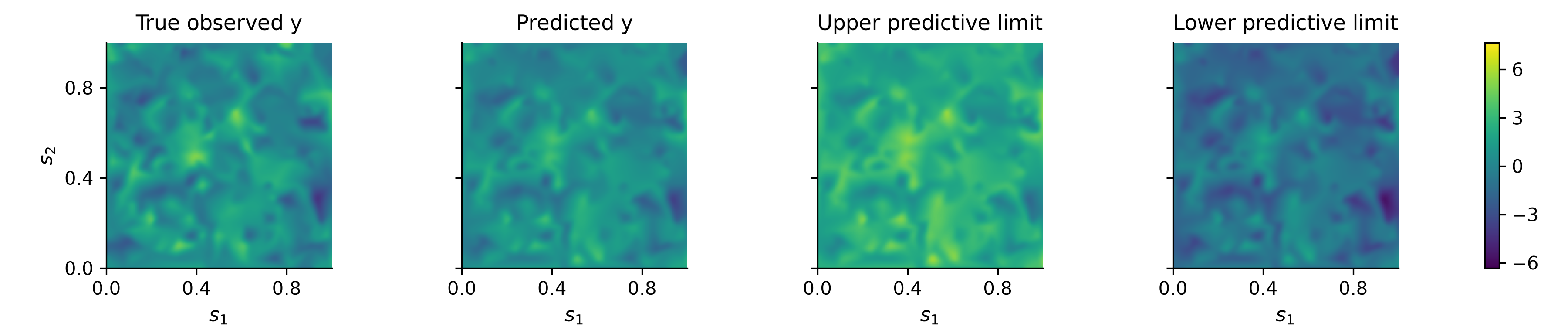}
  \caption{Case $C_1$: rough.}
  \label{fig:sim-y-c1-rep01}
\end{subfigure}
\par\vspace{1mm}
\begin{subfigure}[t]{\linewidth}
  \centering
  \includegraphics[width=\linewidth]{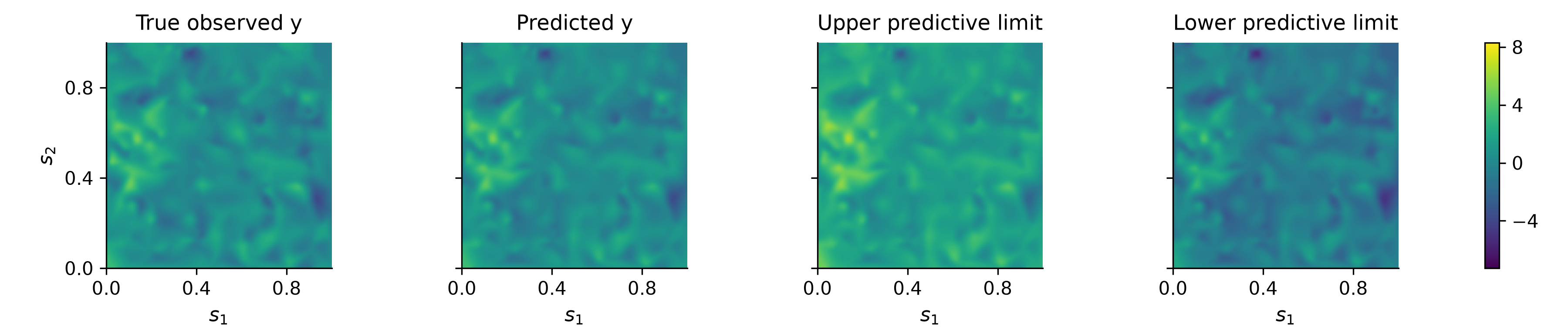}
  \caption{Case $C_2$: moderate.}
  \label{fig:sim-y-c2-rep01}
\end{subfigure}
\par\vspace{1mm}
\begin{subfigure}[t]{\linewidth}
  \centering
  \includegraphics[width=\linewidth]{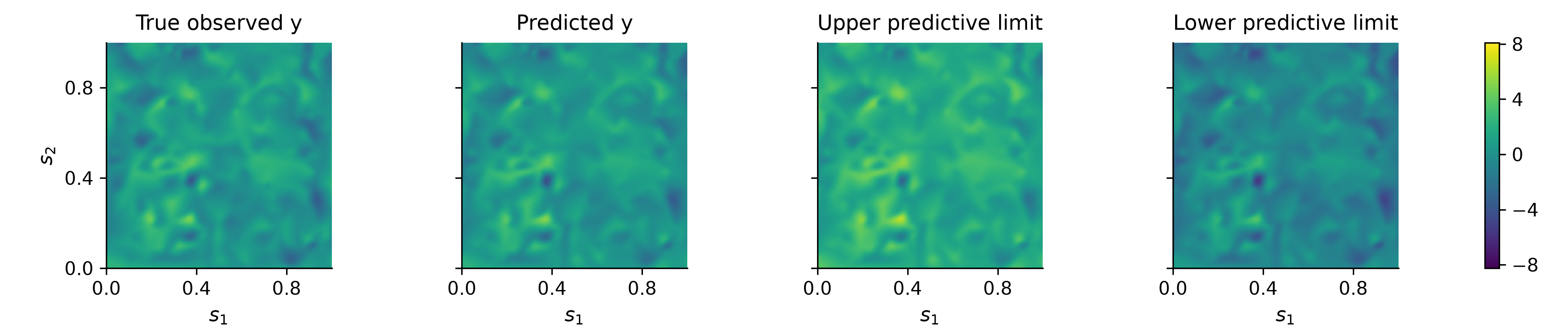}
  \caption{Case $C_3$: smooth.}
  \label{fig:sim-y-c3-rep01}
\end{subfigure}
\par\vspace{1mm}
\begin{subfigure}[t]{\linewidth}
  \centering
  \includegraphics[width=\linewidth]{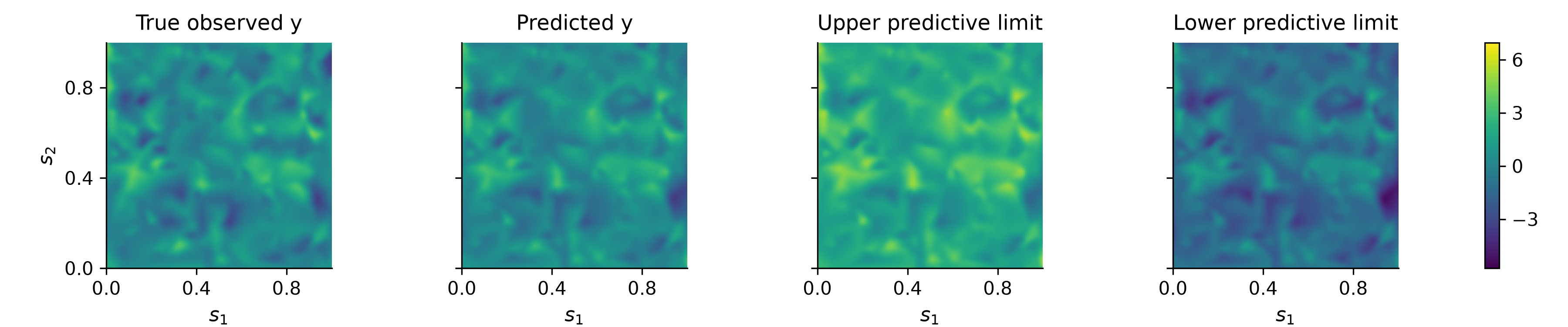}
  \caption{Case $C_4$: heterogeneous.}
  \label{fig:sim-y-c4-rep01}
\end{subfigure}
\par\vspace{1mm}
\begin{subfigure}[t]{\linewidth}
  \centering
  \includegraphics[width=\linewidth]{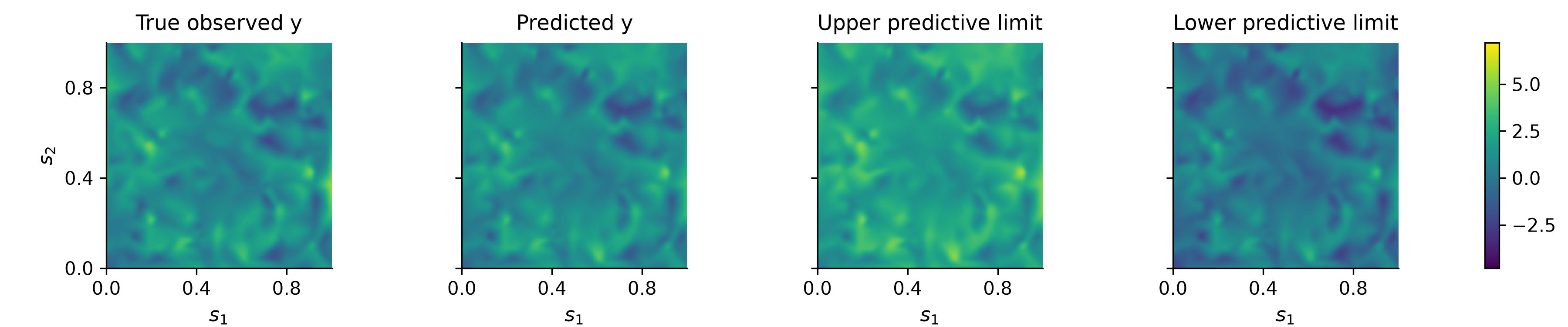}
  \caption{Case $C_5$: nonstationary.}
  \label{fig:sim-y-c5-rep01}
\end{subfigure}
\caption{Test-response prediction in five simulation scenarios. From left to right: observed test responses, posterior mean predictions, upper predictive limits, and lower predictive limits. All four panels share a common color scale within each scenario.}
\label{fig:sim-y-rep01}
\end{figure}


\section{Remote Sensing Vegetation Data Analysis}\label{sec:realdata}

This section illustrates our approach with a remote sensing vegetation data analysis. The western United States is undergoing substantial ecological changes associated with increasing drought stress and wildfire activity~\citep{abatzoglou2016impact,williams2019observed,westerling2006warming}. Understanding spatial variation in vegetation condition is therefore important for monitoring ecosystem responses and informing land management~\citep{running2004continuous,pettorelli2005using}. The region encompasses diverse landscapes, ranging from coastal forests to arid interior environments, where the factors limiting vegetation productivity vary considerably. This environmental heterogeneity motivates an analysis that allows both baseline vegetation conditions and their association with urban land cover to vary across space.

Satellite remote sensing provides extensive spatial coverage for examining these patterns. We use the normalized difference vegetation index (NDVI), a widely used indicator of vegetation greenness and productivity \citep{pettorelli2005using}, to investigate how vegetation conditions differ between urban and nonurban areas across the study region. Our analysis utilizes MODIS Enhanced Vegetation Index data from the western United States, focusing on sinusoidal grid tile h08v05, which spans approximately 30°N–40°N latitude and 104°W–130°W longitude, covering around $N=1,020,000$ observed locations to demonstrate the method. We divide the data into training (60\%), validation (20\%), and testing (20\%) sets to evaluate out-of-sample predictive performance.  We consider the transformed response $y(\bu)=\log{\operatorname{NDVI}(\bu)+1}$ and a binary covariate $x(\bu)$ indicating urban land cover, with $x(\bu)=1$ for urban locations and $x(\bu)=0$ for nonurban locations.

Figure~\ref{fig:beta} presents the estimated spatially varying intercept and urban coefficient surfaces, together with their posterior uncertainty. The estimated intercept surface $\beta_0(\bu)$ describes spatial variation in expected $\log{\operatorname{NDVI}(\bu)+1}$ for nonurban land cover. Posterior means reach approximately $0.4$--$0.5$ in the upper band of the mapped domain, compared with approximately $0.1$--$0.2$ across much of the interior. Pointwise credible interval widths are generally below $0.03$, indicating that posterior uncertainty at individual locations is small relative to these differences in the estimated baseline surface. The posterior mean of $\beta_1(\bu)$ reaches approximately $-0.2$ to $-0.3$ in several localized patches, indicating lower expected $\log{\operatorname{NDVI}(\bs)+1}$ for urban relative to nonurban land cover. In parts of these patches, the upper pointwise credible limit is also below zero, providing posterior support for a negative urban association. Positive posterior means are generally smaller in magnitude, and their credible intervals include zero in many locations. Credible interval widths reach approximately $0.35$ in some areas, compared with a maximum of approximately $0.06$ for the intercept surface. Thus, the urban--nonurban contrast is estimated less precisely than the nonurban baseline, with clearer evidence for negative associations in specific locations than for a uniformly positive or negative association across the domain. Figure~\ref{fig:prediction} contrasts the observed test responses with the GeoVAE posterior mean predictions and the associated 95\% prediction intervals obtained via split conformal calibration. The predicted surface closely matches the log-transformed NDVI field, and the resulting intervals are narrow and closely envelop the observed values.



\begin{figure}[htbp]
\centering
\begin{subfigure}[t]{0.245\textwidth}
  \centering
  \includegraphics[width=\linewidth]{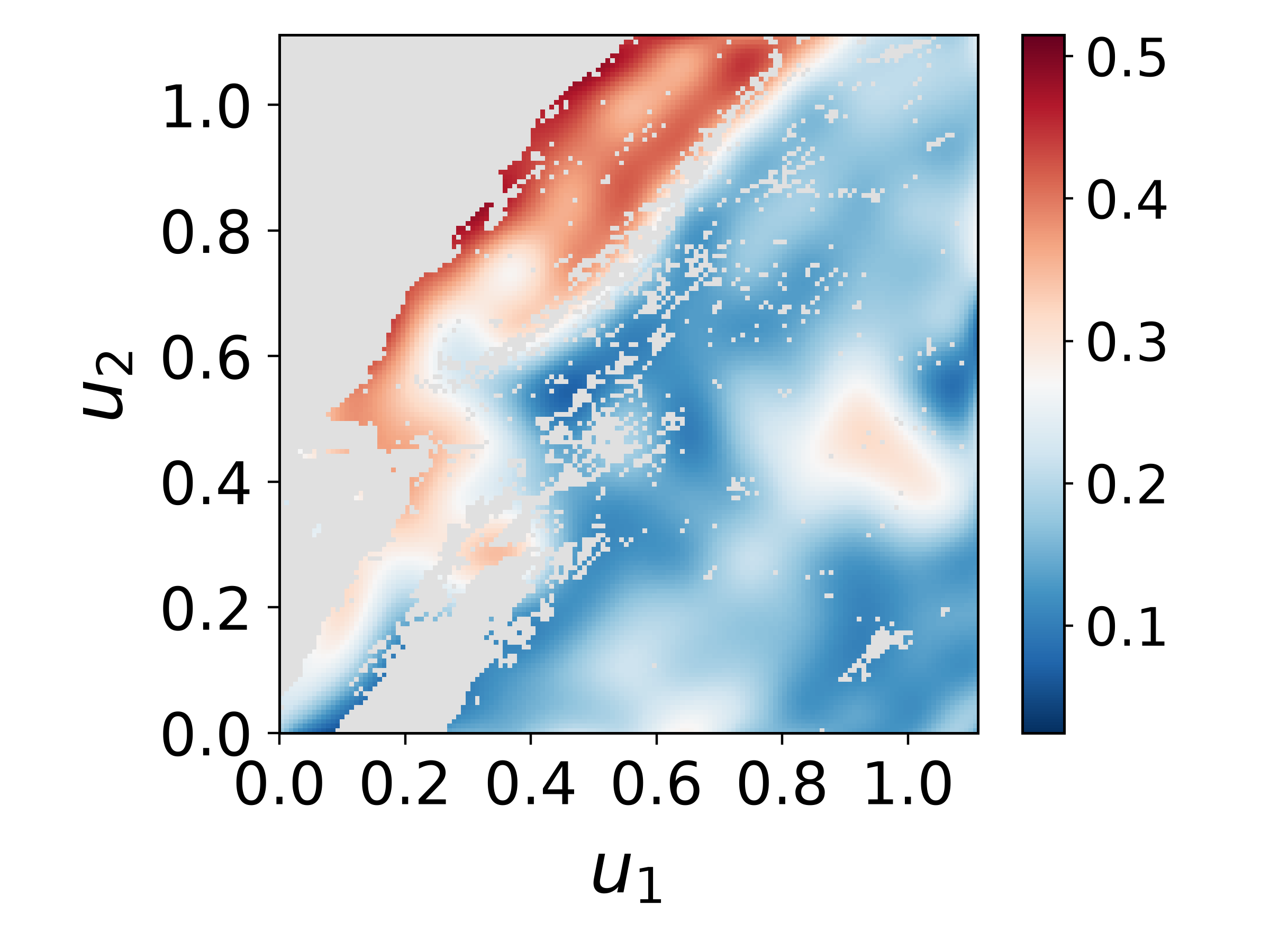}
  \caption{$\beta_0$: posterior mean.}
  \label{fig:beta0-mean}
\end{subfigure}\hfill
\begin{subfigure}[t]{0.245\textwidth}
  \centering
  \includegraphics[width=\linewidth]{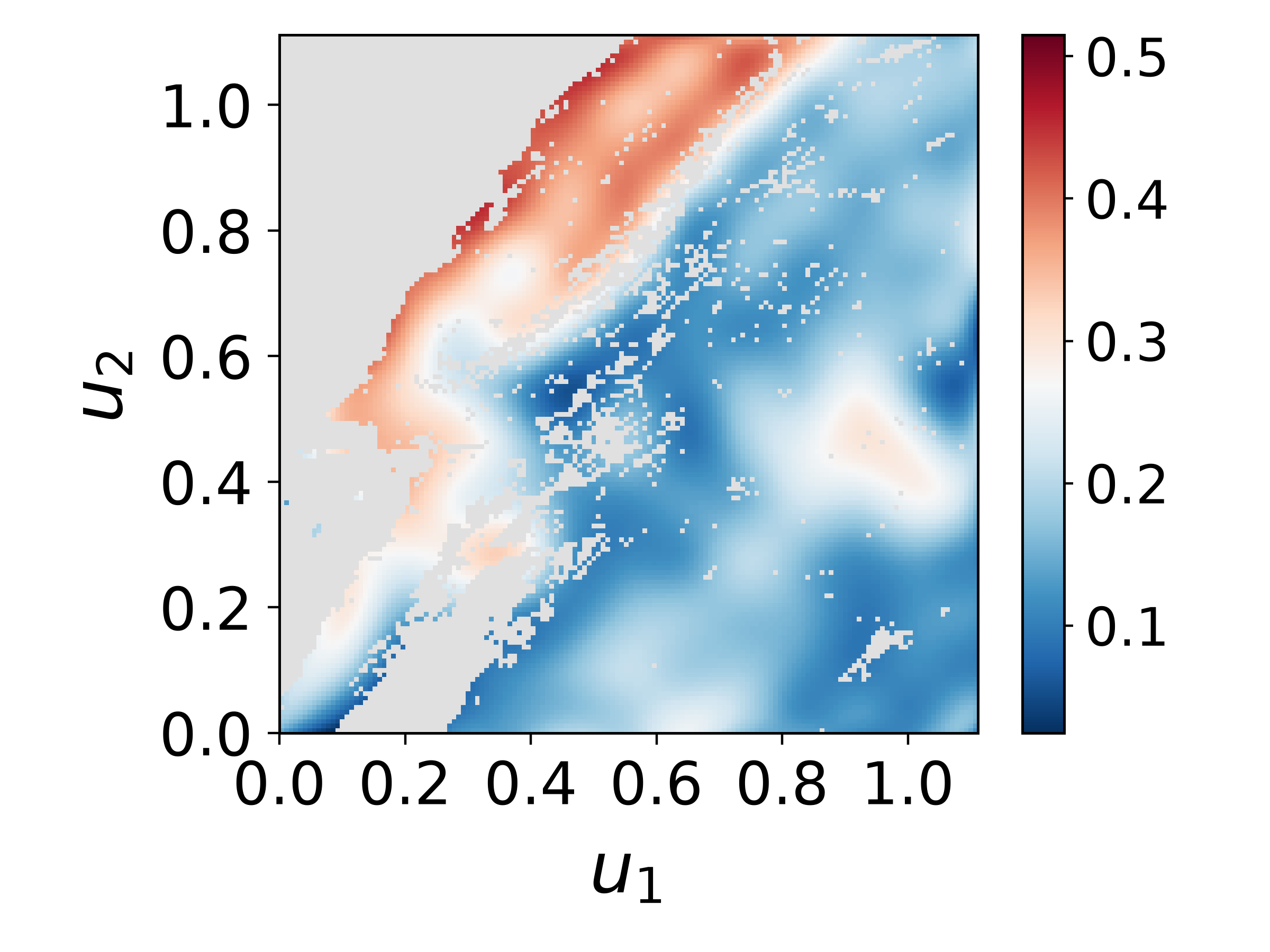}
  \caption{$\beta_0$: lower credible limit.}
  \label{fig:beta0-lower}
\end{subfigure}\hfill
\begin{subfigure}[t]{0.245\textwidth}
  \centering
  \includegraphics[width=\linewidth]{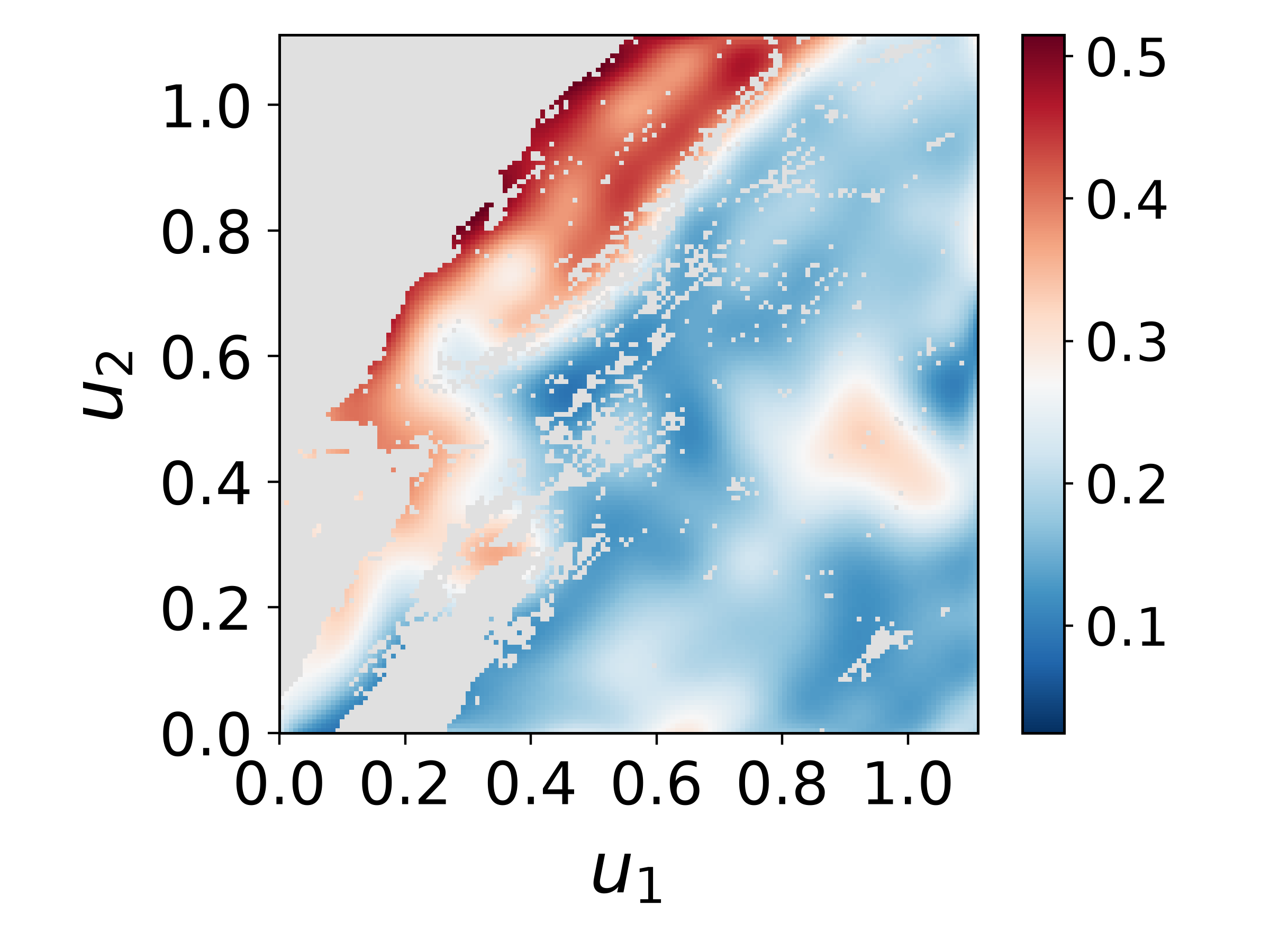}
  \caption{$\beta_0$: upper credible limit.}
  \label{fig:beta0-upper}
\end{subfigure}\hfill
\begin{subfigure}[t]{0.245\textwidth}
  \centering
  \includegraphics[width=\linewidth]{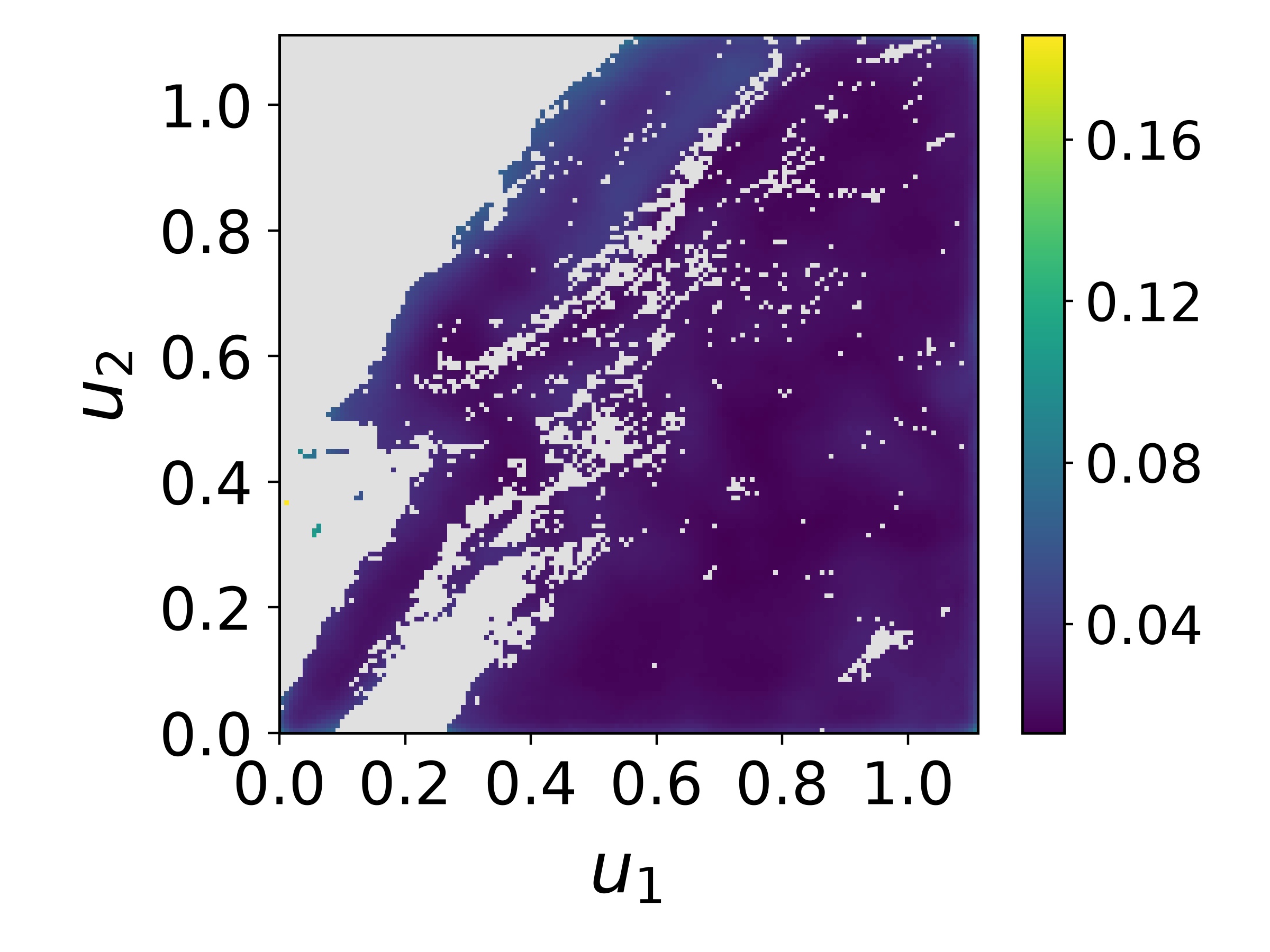}
  \caption{$\beta_0$: credible interval width.}
  \label{fig:beta0-width}
\end{subfigure}

\medskip
\begin{subfigure}[t]{0.245\textwidth}
  \centering
  \includegraphics[width=\linewidth]{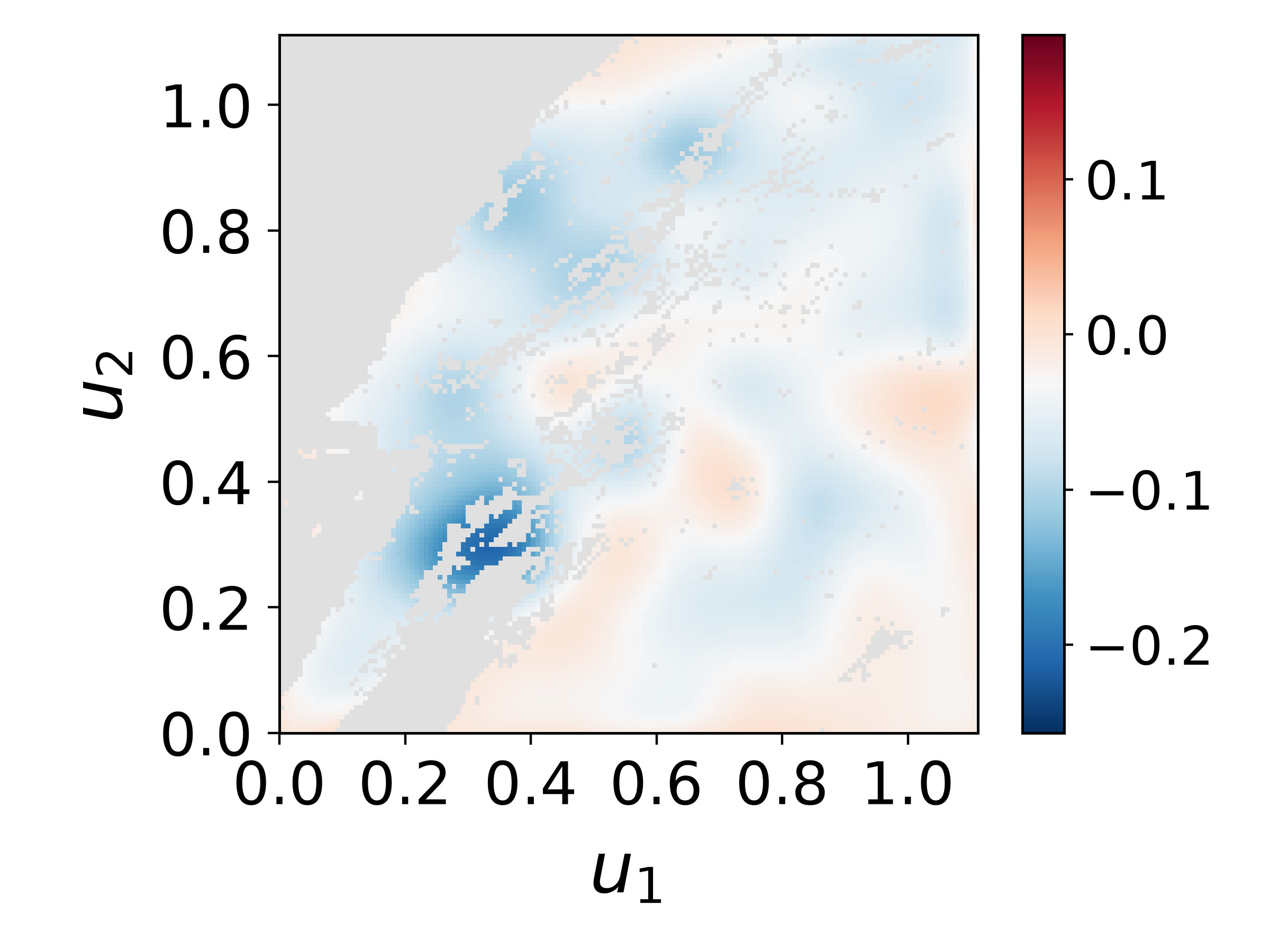}
  \caption{$\beta_1$: posterior mean.}
  \label{fig:beta1-mean}
\end{subfigure}\hfill
\begin{subfigure}[t]{0.245\textwidth}
  \centering
  \includegraphics[width=\linewidth]{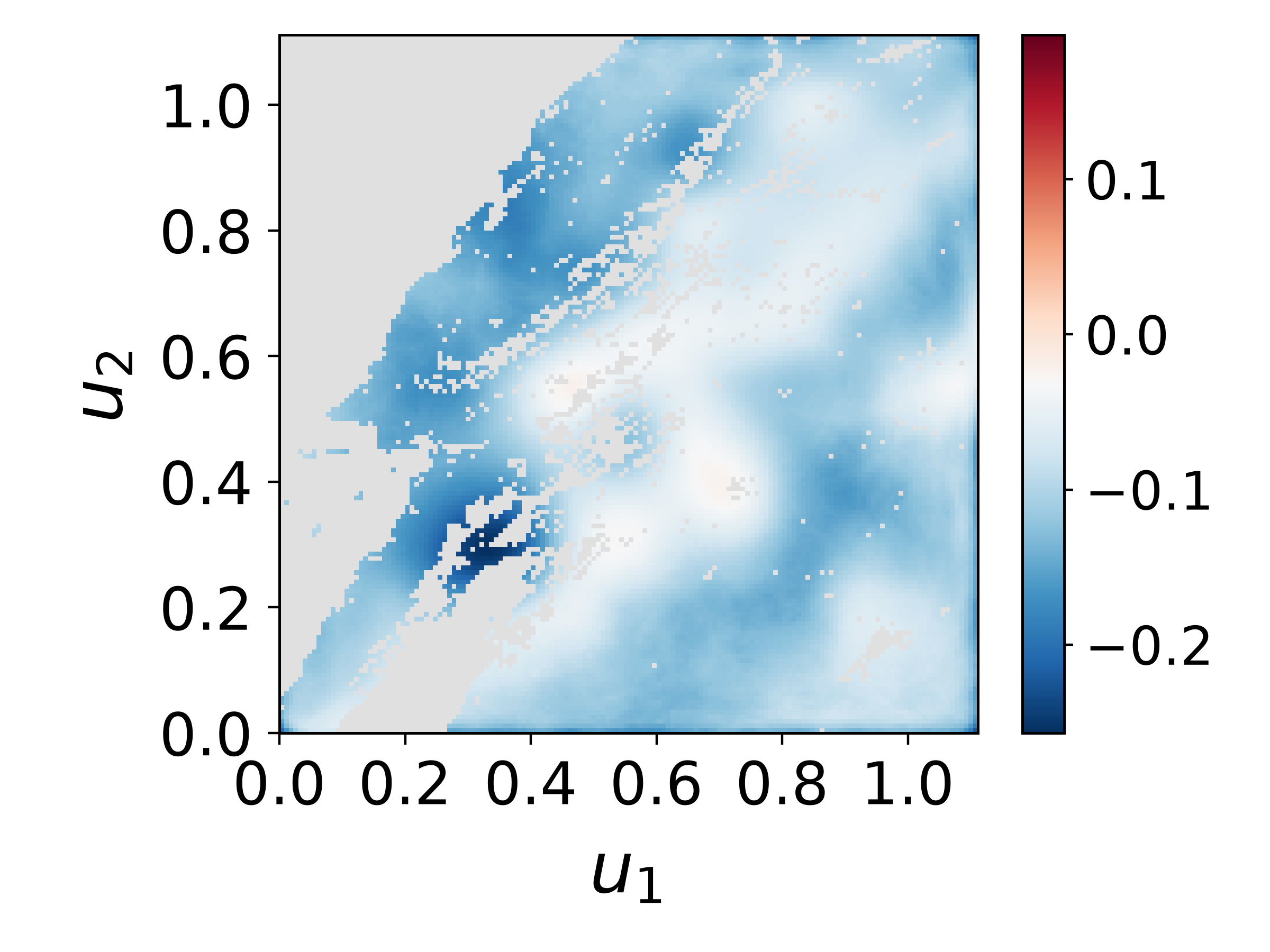}
  \caption{$\beta_1$: lower credible limit.}
  \label{fig:beta1-lower}
\end{subfigure}\hfill
\begin{subfigure}[t]{0.245\textwidth}
  \centering
  \includegraphics[width=\linewidth]{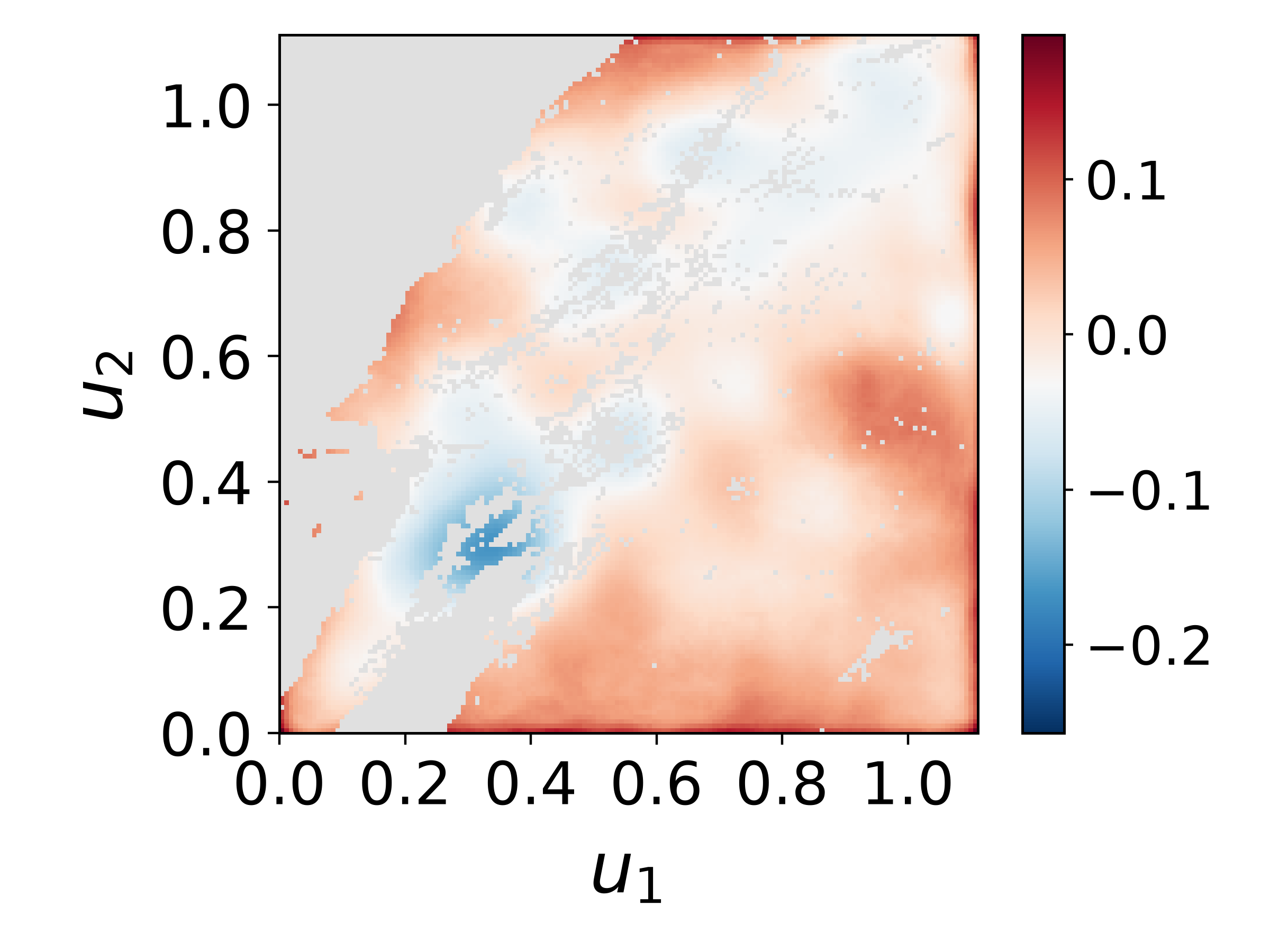}
  \caption{$\beta_1$: upper credible limit.}
  \label{fig:beta1-upper}
\end{subfigure}\hfill
\begin{subfigure}[t]{0.245\textwidth}
  \centering
  \includegraphics[width=\linewidth]{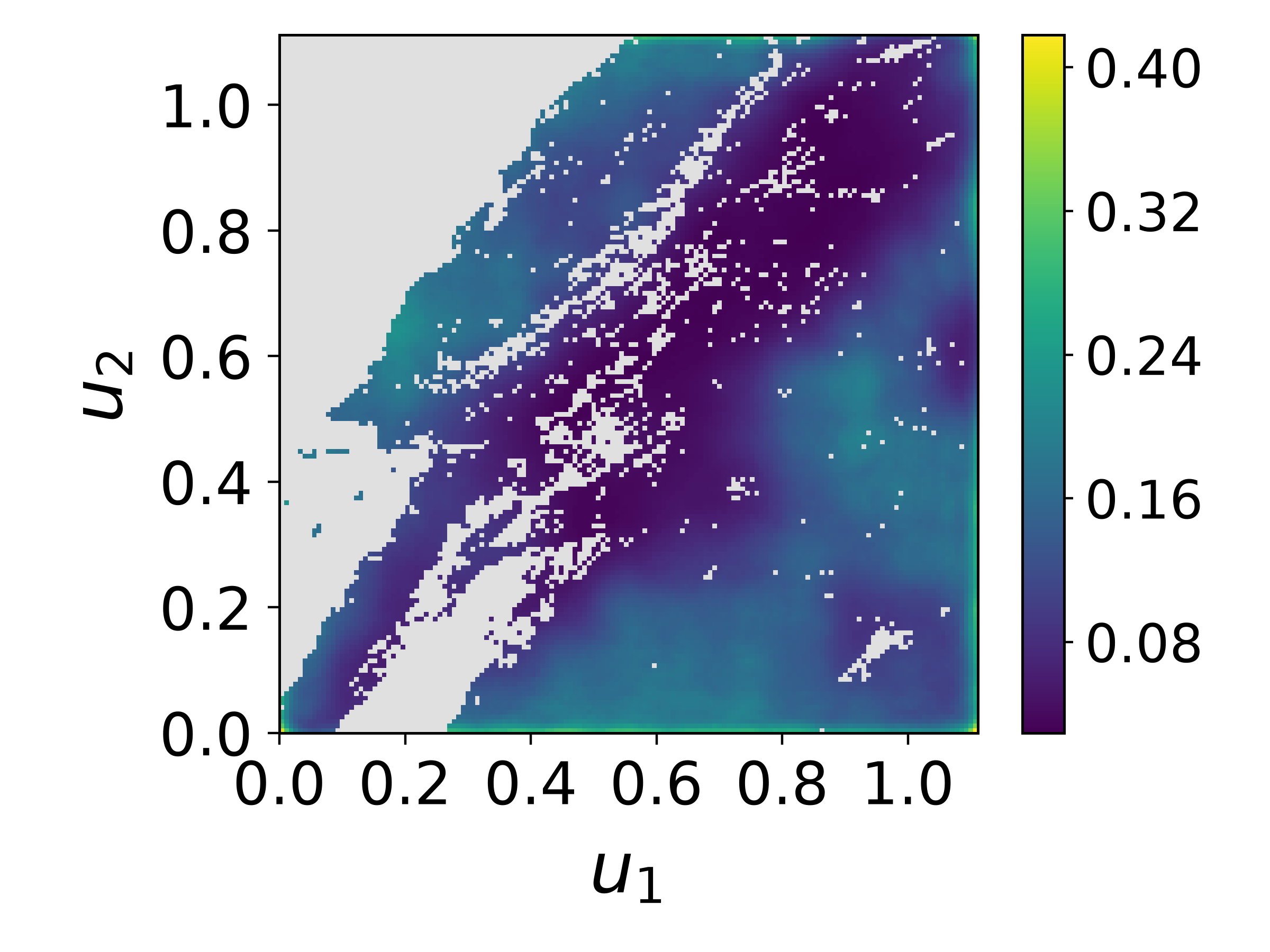}
  \caption{$\beta_1$: credible interval width.}
  \label{fig:beta1-width}
\end{subfigure}
\caption{GeoVAE estimates of the spatially varying intercept
$\beta_0(\mathbf{u})$ (top row) and urban coefficient
$\beta_1(\mathbf{u})$ (bottom row) for
$y(\mathbf{u})=\log\{\operatorname{NDVI}(\mathbf{u})+1\}$.
Columns show the posterior mean, lower and upper pointwise credible limits,
and credible interval width. Within each row, the first three panels share
a common color scale; the interval-width panel uses a separate scale.
Gray areas indicate grid locations excluded by the data-support mask.}
\label{fig:beta}
\end{figure}

\begin{figure}[htbp]
\centering
\begin{subfigure}[t]{0.245\textwidth}
  \centering
  \includegraphics[width=\linewidth]{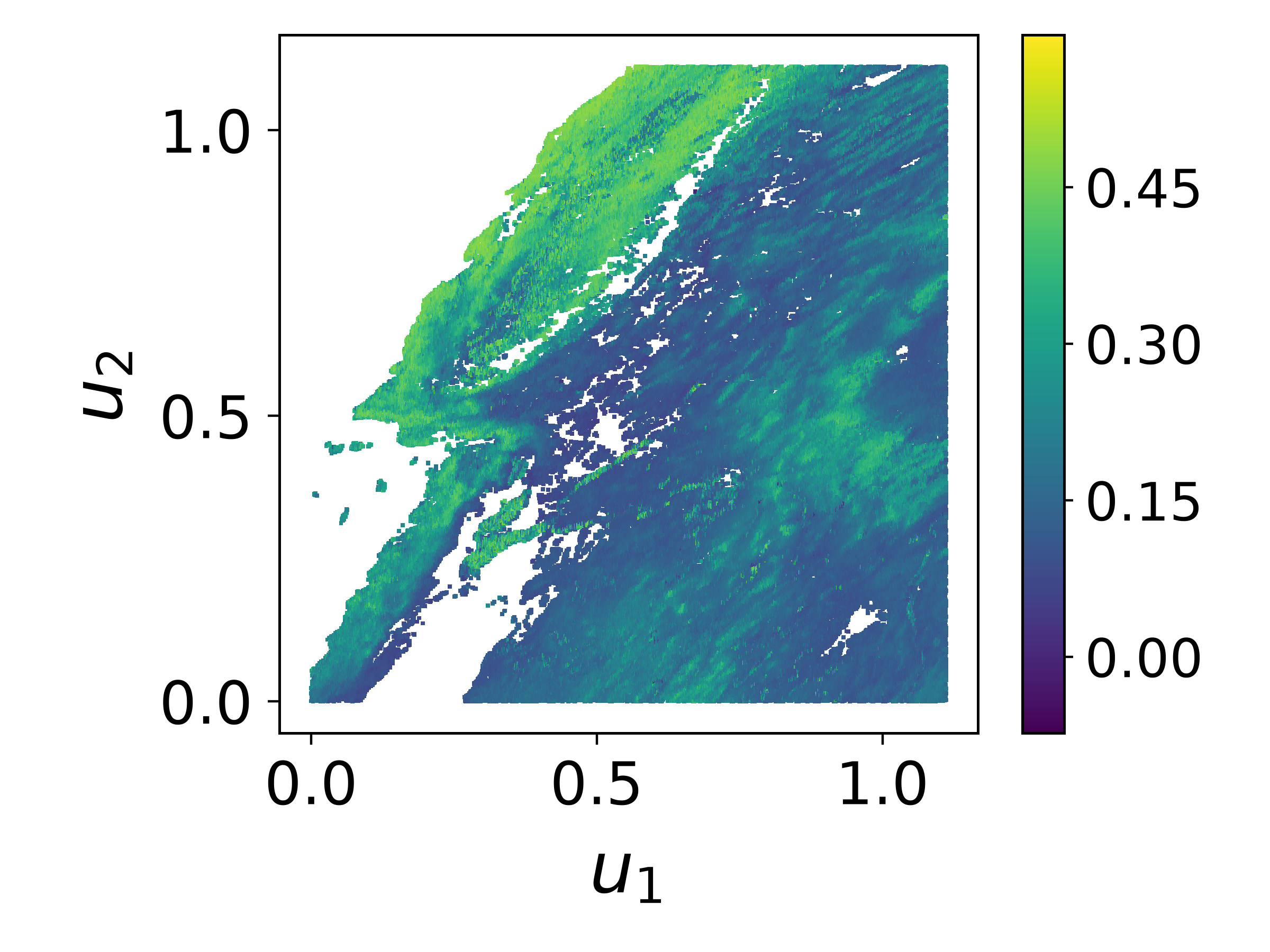}
  \caption{Observed response.}
  \label{fig:prediction-observed}
\end{subfigure}\hfill
\begin{subfigure}[t]{0.245\textwidth}
  \centering
  \includegraphics[width=\linewidth]{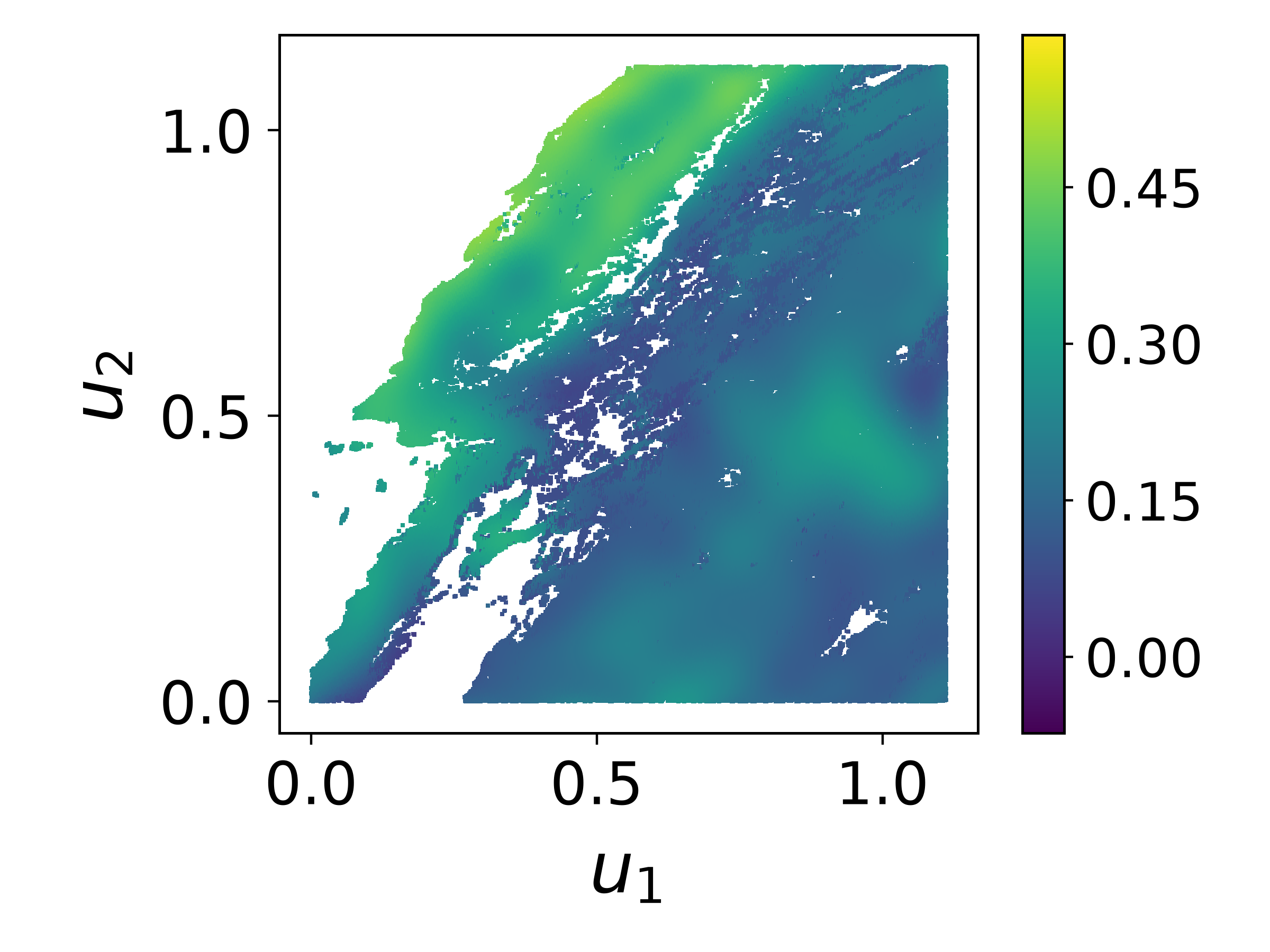}
  \caption{Posterior mean prediction.}
  \label{fig:prediction-mean}
\end{subfigure}\hfill
\begin{subfigure}[t]{0.245\textwidth}
  \centering
  \includegraphics[width=\linewidth]{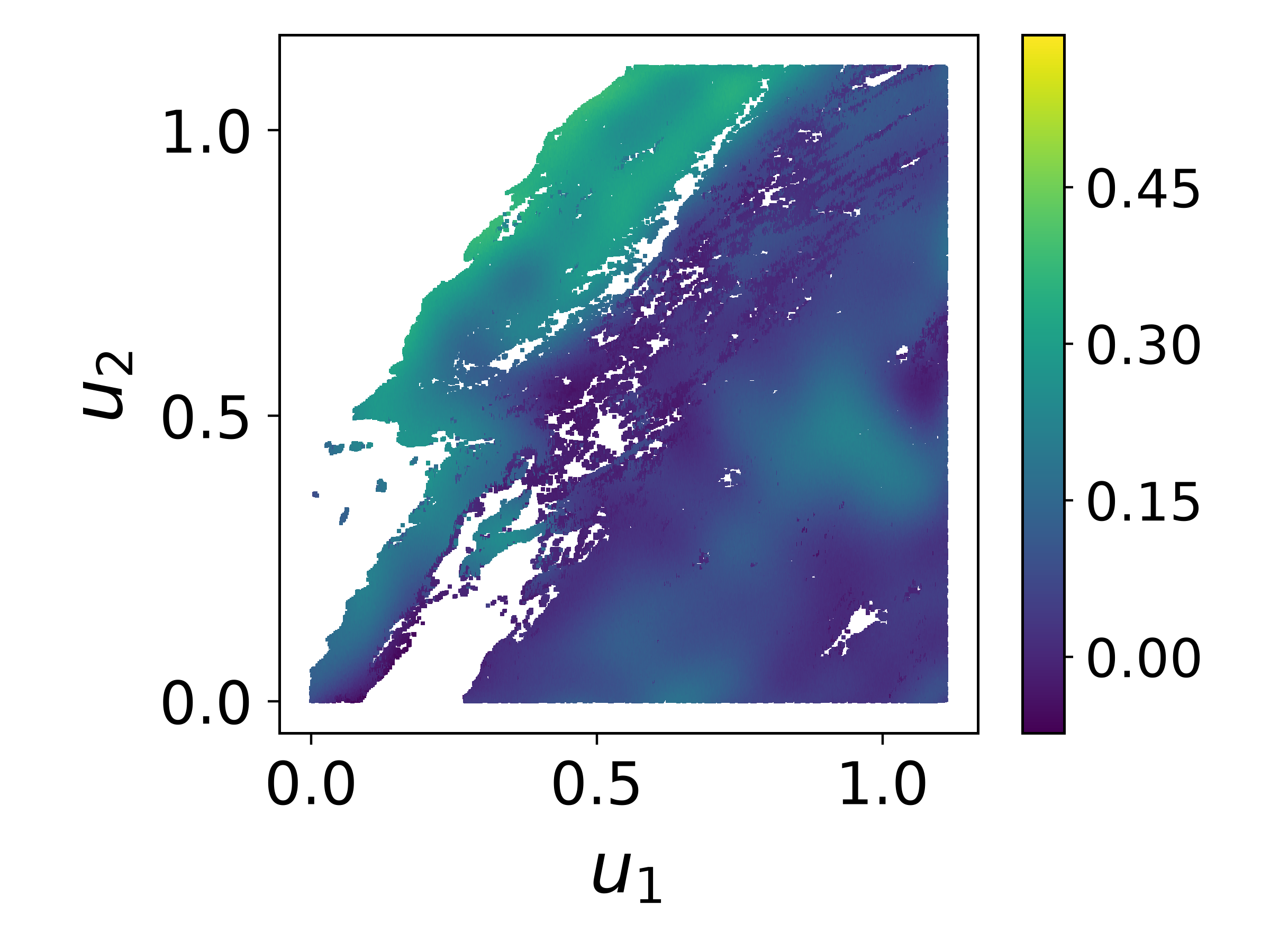}
  \caption{Lower conformal prediction limit.}
  \label{fig:prediction-lower}
\end{subfigure}\hfill
\begin{subfigure}[t]{0.245\textwidth}
  \centering
  \includegraphics[width=\linewidth]{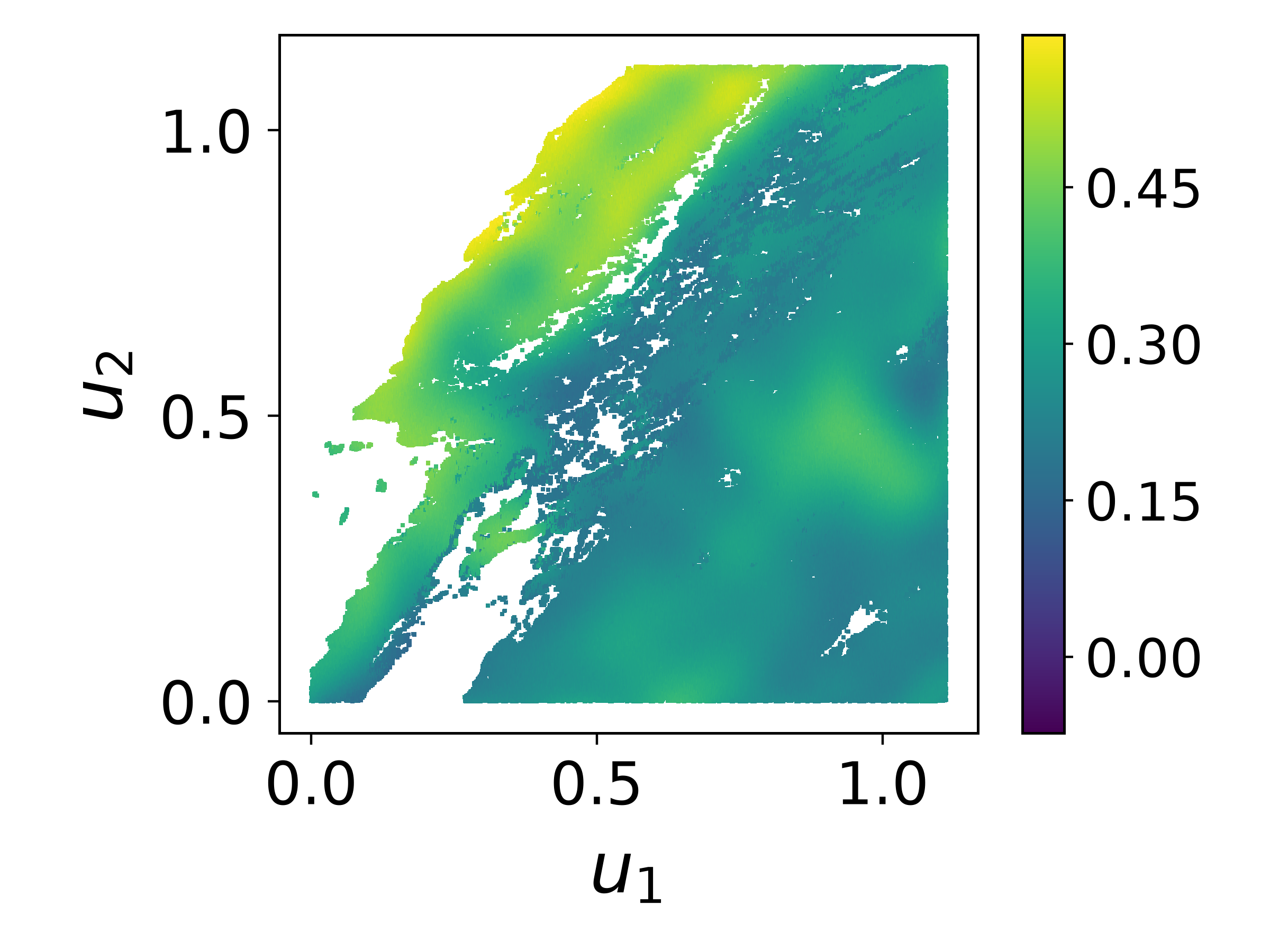}
  \caption{Upper conformal prediction limit.}
  \label{fig:prediction-upper}
\end{subfigure}
\caption{GeoVAE predictions and predictive uncertainty at held-out test
locations. Panels show the observed response
$\log\{\operatorname{NDVI}(\mathbf{u})+1\}$, posterior mean predictions,
and lower and upper conformal prediction limits.
All panels use a common color scale on the transformed response scale.}
\label{fig:prediction}
\end{figure}

Table~\ref{tab:model_comparison} reports predictive performance and computation time for GeoVAE and all competing methods. GeoVAE attains a test MSE of $0.0016$ and empirical coverage of $94.92$\%, very close to the nominal $95$\% level, with an average prediction interval width of $0.1674$. Relative to the nonspatial CVAE, this corresponds to an approximately $83$\% reduction in MSE and a $49$\% reduction in average interval width, while using only about one-seventh of the computation time. The laGP and Vecchia GP approaches yield slightly lower prediction errors and somewhat narrower intervals than GeoVAE. However, laGP requires roughly $1,357$ minutes of computation, compared with $3.35$ minutes for GeoVAE, and Vecchia GP shows a mild underestimation of uncertainty, and it is still about seven times slower than GeoVAE. Overall, the results highlight a clear trade-off between predictive accuracy and computational cost: the MCMC-free GeoVAE delivers the shortest runtime, scales favorably to very large datasets, and achieves coverage close to the nominal level upon calibration through split conformal strategy, whereas the MCMC-based GP approximation methods offer marginal gains in point prediction accuracy at substantially higher computational expense. BASS and DeepKriging are omitted from the comparison because their implementations could not be fitted at this data scale within the available computational resources.

\begin{table}[ht]
\centering
\caption{Out-of-sample predictive performance and computation time for the MODIS vegetation data. MSE denotes mean squared prediction error, Coverage is the empirical coverage of nominal 95\% prediction intervals, and Interval Width is their average width. Prediction errors and interval widths are evaluated on the $\log(\operatorname{NDVI}+1)$ scale. Computation time is reported in minutes.}
\label{tab:model_comparison}
\begin{tabular}{lcccc}
\toprule
Model & MSE & Test Coverage & Interval Width & Time (min) \\
\midrule
GeoVAE            & 0.0016 & 0.9492 & 0.1674 & 3.3460 \\
Nonspatial CVAE    & 0.0093       & 0.9550    & 0.3254     & 22.6210     \\
laGP               & 0.0008       & 0.9506    & 0.0887     & 1356.8178     \\
Vecchia GP     & 0.0005       & 0.9371    & 0.0813     &  22.8091    \\
\bottomrule
\end{tabular}
\end{table}
\section{Conclusion and Future Work}\label{sec:conclusion}
Remote sensing technologies have led to an explosion of massive geocoded datasets, pushing traditional Gaussian process–based spatial models beyond their computational limits. Even carefully crafted GP approximation methods, while capable of capturing complex spatial variability, remain constrained by their reliance on MCMC-based inference. In this article, we introduce a fundamentally different, MCMC-free approach: a hierarchical variational auto-encoder tailored to spatially varying coefficient models for big spatial data. The proposed GeoVAE framework assigns a dedicated encoder to the basis coefficients associated with each predictor-specific coefficient surface, aggregates the coefficient-specific latent representations from each encoder through a shared encoder, and then reconstructs the coefficient surfaces via corresponding decoders. An end-to-end training objective that directly reconstructs the observed responses ensures that all predictor-specific encoders and decoders, along with the shared encoder, are learned jointly in a coherent manner. Under mild regularity conditions, our theoretical results establish consistency of the estimated coefficient surfaces and of out-of-sample predictions, providing formal accuracy guarantees for the proposed estimator. Empirically, GeoVAE delivers predictive performance comparable to state-of-the-art deep learning methods and MCMC-based GP approximations, while achieving substantially lower computation times and thus offering a practical route to scalable spatial inference on large-scale geospatial datasets.

As a direction for future work, we plan to extend GeoVAE to multi-scale spatial models, in which encoders and decoders operating at different spatial resolutions collaboratively approximate fine- and coarse-scale spatial structure. Such a multi-resolution GeoVAE would further enhance flexibility in capturing spatial features across scales, while retaining the computational advantages of variational inference.

\section*{Appendix}
\subsection{Auxiliary Results}
We state an auxiliary lemma which will be used in the proofs of main consistency results.
\begin{lemma}\label{lemma: bound1}
Let $\mu(\bu)$ be parameterized as an $L$-layer feedforward neural network with ReLU activation $\rho(\cdot)$ and parameter vector ${\boldsymbol \upsilon}$ collecting all vectorized weight matrices and bias vectors. Let $||\cdot||$ and $||\cdot||_F$ represent the Euclidean and Frobenius norms, respectively.\\
(i) If $||{\boldsymbol \upsilon}||\leq R_{\upsilon}<\infty$, then
$||\bmu(\bu)||\leq R_{\upsilon}^L(||\bu||+1)^L$.\\
(ii) If additionally
$E[||\bu||^r]\leq C_u<\infty$, for all $r>1$, for some constant $C_u>0$, then there exists a constant $C_\mu>0$ such that $E[|\bmu(\bu)||^2]\leq C_\mu<\infty$.
\end{lemma}
\begin{proof}
(i) Let $\bW_1$ and $\bb_1$ represent the weight matrix and bias vector for the first layer, respectively. At the first layer, $||\bW_1\bu+\bb_1||\leq ||\bW_1||_F||\bu||+||\bb_1||\leq R_{\upsilon}(||\bu||+1).$ Since $|\rho(x)|=|max(x,0)\leq |x|$ for ReLU, applying the activation function $\rho(\cdot)$ does not increase the norm. Iterating over all L layers yields part (i).\\
(ii) $E[||\bmu(\bu)||^2]\leq R_{\upsilon}^{2L} \sum_{l=0}^{2L}{2L\choose l}E[||\bu||^{l}]\leq C_\mu$ for some constant $C_\mu>0$, where the last inequality uses the assumed finiteness of all moments of $||\bu||$.
\end{proof}

\subsection{Proofs of Consistency Results}
\underline{\textbf{Proof of Theorem~\ref{lem:param_consistency}}}
\begin{proof}
Decompose the error as
\begin{align*}
\hat\beta_j(\bu) - \beta_j^*(\bu)
&=
\left\{\bB_j(\bu)^\top\bigl(\bar\balpha_j - \tilde\balpha_j\bigr)\right\}
+
\left\{\bB_j(\bu)^\top\bigl(\tilde\balpha_j - \balpha_j^{(K_j)}\bigr)\right\}
+
\left\{\bB_j(\bu)^\top\balpha_j^{(K_j)} - \beta_j^*(\bu)\right\}\\
&=T_1(\bu)+T_2(\bu)+T_3(\bu)
\end{align*}
so that by the Cauchy inequality,
\begin{equation}\label{eq:original}
E_\nu\!\left[\bigl(\hat\beta_j(\bu) - \beta_j^*(\bu)\bigr)^2\right]
\leq 3\left\{
E_\nu\bigl[T_1(\bu)^2\bigr]
+
E_\nu\bigl[T_2(\bu)^2\bigr]
+
E_\nu\bigl[T_3(\bu)^2\bigr]
\right\}.
\end{equation}

\paragraph{\underline{Bound on $E_\nu[T_3(\bu)^2]$}}
By Assumption~2,
\begin{equation*}
|T_3(\bu)| = \bigl|\beta_j^*(\bu) - \bB_j(\bu)^\top\balpha_j^{(K_j)}\bigr|
\leq C_j K_{j,N}^{-s_j/d},
\end{equation*}
uniformly in $\bu \in \mathcal{D}$. Therefore,
$E_\nu\bigl[T_3(\bu)^2\bigr]
\leq C_j^2 K_{j,N}^{-2s_j/d}.$
With $K_{j,N} = O\bigl(N^{d/(2s_j+d)}\bigr)$,
\begin{equation}\label{eq:third_term_conv}
E_\nu\bigl[T_3(\bu)^2\bigr]
= O\!\left(N^{-2s_j/(2s_j+d)}\right) \to 0.
\end{equation}

\paragraph{\underline{Bound on $E_\nu[T_1(\bu)^2]$}}

From equation (\ref{eq:local_decoder}), we have the decomposition
\begin{equation*}
\balpha_j^{(m)}
=
\bmu_{\hat{\btheta}_j}\bigl(\bgamma_j^{(m)}\bigr)
+
\bsigma_{\hat{\btheta}_j}\bigl(\bgamma_j^{(m)}\bigr)
\odot
\bxi_{\balpha_j}^{(m)},
\end{equation*}
so that
\begin{equation*}
\mathrm{Var}\bigl(\balpha_j^{(m)}\bigr)
=
\mathrm{Var}_{q_{\hat\bPhi}}\!\left(\bmu_{\hat{\btheta}_j}\bigl(\bgamma_j^{(m)}\bigr)\right)
+
E_{q_{\hat\bPhi}}\!\left[
  \mathrm{diag}\left\{\bsigma_{\hat{\btheta}_j}^2\bigl(\bgamma_j^{(m)}\bigr)\right\}
\right].
\end{equation*}
Since the GeoVAE ELBO enforces
$\sum_{j=0}^J
\mathrm{KL}\!\left\{
  q_{\hat{\bphi}_j}(\bz_j \mid \by)
  \,\Big\|\,
  \mathcal{N}(\bzero, \bI_{d_j})
\right\}
< \infty,$
using the closed-form KL expression from equation (\ref{eq:hier_kl}), this implies
\begin{equation*}
\|\bmu_{\hat{\bphi}_j}(\by)\| < \infty
\quad\text{and}\quad
\|\bsigma_{\hat{\bphi}_j}(\by)\| < \infty.
\end{equation*}
For the $r$th component
$\gamma_{j,r}
= \mu_{\hat{\bphi}_j,r}(\by) + \sigma_{\hat{\bphi}_j,r}(\by)\,\xi_{j,r}$,
with $\xi_{j,r} \sim \mathcal{N}(0,1)$, applying the inequality
$(a+b)^R \leq 2^{R-1}(a^R + b^R)$ for any $a,b \ge 0$ and $R>1$ gives
\begin{equation*}
|\gamma_{j,r}|^R
\leq
2^{R-1}\!\left(
  |\mu_{\hat{\bphi}_j,r}(\by)|^R
  +
  \sigma_{\hat{\bphi}_j,r}(\by)^R\,|\xi_{j,r}|^R
\right).
\end{equation*}
Taking expectations and using the fact that $\xi_{j,r} \sim \mathcal{N}(0,1)$ has all finite moments,
\begin{equation*}
E_{q_{\hat{\bPhi}}}\!\left[|\gamma_{j,r}|^R\right]
\leq
2^{R-1}\!\left(
  |\mu_{\hat{\bphi}_j,r}(\by)|^R
  +
  \sigma_{\hat{\bphi}_j,r}(\by)^R\,\mathbb{E}[|\xi_{j,r}|^R]
\right)
\leq
C_\gamma
< \infty,
\end{equation*}
for some constant $C_\gamma > 0$. Hence all moments of $\|\bgamma_j\|$ are finite. Invoking
Assumption~5 and Lemma~1.1(ii), we obtain
\begin{equation*}
E_{q_{\hat{\bPhi}}}\!\left[
  \bigl\|
    \bmu_{\hat{\btheta}_j}\bigl(\bgamma_j^{(m)}\bigr)
  \bigr\|^2
\right]
\leq
C_\mu
< \infty,
\end{equation*}
for some constant $C_\mu > 0$.

By the same argument applied to the decoder variance network, and using Assumption~5,
\begin{equation*}
E_{q_{\hat{\bPhi}}}\!\left[
  \sigma_{\hat{\btheta}_j,r}\bigl(\bgamma_j^{(m)}\bigr)
\right]
\leq
\sigma_{\max}
< \infty,
\end{equation*}
where $\sigma_{\hat{\btheta}_j,r}(\bgamma_j^{(m)})$ is the $r$th element of $\sigma_{\hat{\btheta}_j}(\bgamma_j^{(m)})$,
and $\sigma_{\max} > 0$ depends only on the bounded parameter space from Assumption~5.
Combining both bounds, for the $r$th element $\alpha_{j,r}^{(m)}$ of $\balpha_j^{(m)}$,
\begin{equation*}
\mathrm{Var}\bigl(\alpha_{j,r}^{(m)}\bigr)
\leq
C_\mu^2 + \sigma_{\max}^2
=: V_j
< \infty,
\quad
r = 1,\ldots,K_j.
\end{equation*}
Summing over $r = 1,\ldots,K_j$,
\begin{equation*}
E\!\left[
  \bigl\|
    \balpha_j^{(m)} - \tilde\balpha_j
  \bigr\|_2^2
\right]
=
\sum_{r=1}^{K_j}
\mathrm{Var}\bigl(\alpha_{j,r}^{(m)}\bigr)
\leq
K_j V_j.
\end{equation*}
By the i.i.d.\ structure of
$\{\balpha_j^{(m)}\}_{m=1}^M$ given the observed data and the independence of Monte Carlo draws
across samples,
\begin{equation}\label{eqinq0}
E\!\left[
  \bigl\|
    \bar\balpha_j - \tilde\balpha_j
  \bigr\|_2^2
\right]
=
\frac{1}{M}
E\!\left[
  \bigl\|
    \balpha_j^{(1)} - \tilde\balpha_j
  \bigr\|_2^2
\right]
\leq
\frac{K_j V_j}{M}.
\end{equation}
Finally, applying the Cauchy–Schwarz inequality and Assumption~3
($\|\bB_j(\bu)\| \leq C_B\sqrt{K_j}$), we obtain
\begin{align*}
E_\nu[T_1(\bu)^2]
&=
E_\nu\!\left[
  \bigl(
    \bB_j(\bu)^\top(\bar\balpha_j - \tilde\balpha_j)
  \bigr)^2
\right]
\leq
E_\nu\!\left[\|\bB_j(\bu)\|^2\right]
\cdot
E\!\left[
  \bigl\|
    \bar\balpha_j - \tilde\balpha_j
  \bigr\|_2^2
\right] \\
&\leq
C_B^2 K_j \cdot \frac{K_j V_j}{M}
=
O_P\!\left(\frac{K_j^2}{M}\right),
\end{align*}
since $V_j$ does not depend on $N$. With the choice $K_{j,N} = O\bigl(N^{d/(2s_j+d)}\bigr)$ and
$M = M_N \geq N^{2d/(2s_j+d)}$,
\begin{equation}\label{eq:first_term_conv}
E_\nu[T_1(\bu)^2]
=
O_P\!\left(\frac{K_j^2}{M}\right)
=
O_P\!\left(\frac{N^{2d/(2s_j+d)}}{M_N}\right)
\stackrel{P}{\to} 0
\quad\text{as } N \to \infty.
\end{equation}

\paragraph{\underline{Bound on $\mathbb{E}_\nu[T_2(\bu)^2]$}}
By the Cauchy--Schwarz inequality and Assumption~3 (i.e.,
$\sup_{\bu \in \mathcal{D}} |B_{jk}(\bu)| \leq C_B$, so that
$\|\bB_j(\bu)\|^2 \leq C_B^2 K_j$), we have
\begin{align}\label{eq:mother}
E_\nu[T_2(\bu)^2]
=
E_\nu\!\left[
  \bigl( \bB_j(\bu)^\top (\tilde{\balpha}_j - \balpha_j^{(K_j)}) \bigr)^2
\right]
\leq
E_\nu\!\left[\|\bB_j(\bu)\|^2\right]
\cdot
\|\tilde{\balpha}_j - \balpha_j^{(K_j)}\|^2
\leq
C_B^2 K_j \,\|\tilde{\balpha}_j - \balpha_j^{(K_j)}\|^2. 
\end{align}
It therefore suffices to establish an upper bound on
$\|\tilde{\balpha}_j - \balpha_j^{(K_j)}\|^2$.

Under the reference configuration $(\bPhi^0, \bTheta^0)$ from Assumption~6, the
GeoVAE ELBO decomposes as
\begin{align*}
\mathcal{L}(\bPhi^0, \bTheta^0)
=
\underbrace{
  \frac{1}{2\tau^2}
  \mathbb{E}_{q_{\bPhi^0}, q_{\bTheta^0}}
  \Bigl\|
    \by - \sum_{j=0}^J \bD^{(j)}\balpha_j
  \Bigr\|_2^2
}_{\text{reconstruction term}}
+
\underbrace{
  \lambda_{\mathrm{local}}
  \sum_{j=1}^J
  \mathrm{KL}\!\left(
    q_{\bPhi^0}(\bz_j \mid \by)
    \,\big\|\,
    \mathcal{N}(\bzero, \bI_{d_j})
  \right)
}_{\text{regularization term}}.
\end{align*}
By Assumption~(A6.1), $\bmu_{\btheta_j^0}(\bgamma_j) = \balpha_j^{(K_j)}$, so
$\balpha_j = \bmu_{\btheta_j^0} + \bsigma_{\btheta_j^0} \odot \bxi_{\balpha_j}$.
Using the inequality $(a+b)^2 \leq 2(a^2 + b^2)$, we obtain
\begin{align*}
\frac{1}{2\tau^2}
E_{q_{\bPhi^0}, q_{\bTheta^0}}
\Bigl\|
  \by - \sum_{j=0}^J \bD^{(j)}\balpha_j
\Bigr\|^2
\leq
\frac{1}{\tau^2}
E
\Bigl\|
  \by - \sum_{j=0}^J \bD^{(j)}\balpha_j^{(K_j)}
\Bigr\|^2
+
\frac{1}{\tau^2}
E
\Bigl\|
  \sum_{j=0}^J \bD^{(j)}\bsigma_{\btheta_j^0} \odot \bxi_{\balpha_j}
\Bigr\|^2.
\end{align*}
\underline{\textit{Bounding the first part.}}
Under the true data-generating process
$\by = \sum_{j=0}^J \bD^{(j)}\balpha_j^{(K_j)} + \bg^* + \bepsilon$, where $\bg^*$ is the
misspecification bias vector defined in (\ref{eq:misclass}), we have
\begin{align}\label{eqinq1}
\frac{1}{\tau^2}
E
\Bigl\|
  \by - \sum_{j=0}^J \bD^{(j)}\balpha_j^{(K_j)}
\Bigr\|^2
=
\frac{1}{\tau^2}\|\bg^*\|^2
+
\frac{1}{\tau^2}\mathrm{Var}(\bepsilon)
=
\frac{1}{\tau^2}\|\bg^*\|^2 + N.
\end{align}
 $\|\bg^*\|^2=\sum_{i=1}^N g_i^{*2}\leq (J+1)\sum_{i=1}^N\sum_{j=0}^J |x_j(\bu_i)|^2|\beta_j^*(\bu_i)-\bB_j(\bu_i)^\top\balpha_j^{(K_j)}|^2\leq C_x^2N\sum_{j=0}^JC_j^2K_j^{-2s_j/d}$, By Assumptions~2 and 7.

\noindent\underline{\textit{Bounding the second part.}}
By Assumption~4 (bounded eigenvalues of $\bSigma_N$) and Assumption~A6.1,
\begin{equation}\label{eqinq2}
\frac{1}{\tau^2}
E
\Bigl\|
  \sum_{j=0}^J \bD^{(j)}\bsigma_{\btheta_j^0} \odot \bxi_{\balpha_j}
\Bigr\|^2
\leq
\frac{(J+1)}{\tau^2}
\sum_{j=0}^J
E
\Bigl[
  \bigl\|
    \mathrm{diag}(\bsigma_{\btheta_j^0})
    \bD^{(j)\top} \bD^{(j)}
  \bigr\|_F
\Bigr]
=
O_P(N\delta_N^2).
\end{equation}
\underline{\textit{Bounding the regularization term.}}
By Assumption~A6.2,
\begin{equation}\label{eqinq3}
\lambda_{\mathrm{local}}
\sum_{j=1}^J
\mathrm{KL}\!\left(
  q_{\bPhi^0}(\bz_j \mid \by)
  \,\big\|\,
  \mathcal{N}(\bzero, \bI_{d_j})
\right)
\leq
(J+1)\lambda_{\mathrm{local}} C_{\mathrm{KL}}
=
O_P(\Lambda_N),
\end{equation}
where $\Lambda_N = (J+1)\lambda_{\mathrm{local}} C_{\mathrm{KL}}$.

Combining (\ref{eqinq1}), (\ref{eqinq2}), (\ref{eqinq3}) and dividing by $N$ gives (also using $K_{j,N}=O(N^{-d/(2s_j+d)})$)
\begin{equation}\label{eqinq4}
\frac{1}{N}\mathcal{L}(\bPhi^0,\bTheta^0)
\leq
O_P\!\bigl(\sum_{j=0}^J K_{j,N}^{-2s_j/d}\bigr)
+ O_P(\delta_N^2)
+ O_P(\Lambda_N)\leq O\!\bigl(\sum_{j=0}^J N^{-2s_j/(2s_j+d)}\bigr)
+ O_P(\delta_N^2)
+ O_P(\Lambda_N).
\end{equation}
By Assumption~(A6.3), the trained parameters satisfy
$\frac{1}{N}\mathcal{L}(\widehat{\bPhi},\widehat{\bTheta})
\leq
\frac{1}{N}\mathcal{L}(\bPhi^0,\bTheta^0) + \eta_N.$
Since the KL terms in $\mathcal{L}(\widehat{\bPhi},\widehat{\bTheta})$ are nonnegative, we have
\begin{equation}\label{eqinq5}
\frac{1}{2N\tau^2}
E_{q_{\hat{\bPhi}},q_{\hat{\bTheta}}}
\Bigl\|
  \by - \sum_{j=0}^J \bD^{(j)}\balpha_j
\Bigr\|_2^2
\leq
\frac{1}{N}\mathcal{L}(\bPhi^0,\bTheta^0)
\leq
O_P\!\bigl(\sum_{j=0}^J N^{-2s_j/(2s_j+d)}\bigr)
+ O_P(\delta_N^2)
+ O_P(\Lambda_N)
+ O_P(\eta_N).
\end{equation}
Since $\tilde{\balpha}_j = E_{q_{\hat{\bPhi}},p_{\hat{\btheta}_j}}[\balpha_j]$ and
$\|\cdot\|^2$ is convex, Jensen’s inequality yields
\[
\Bigl\|
  \by - \sum_{j=0}^J \bD^{(j)}\tilde{\balpha}_j
\Bigr\|^2
=
\Bigl\|
  E\bigl[\by - \sum_j \bD^{(j)}\balpha_j\bigr]
\Bigr\|^2
\leq
E
\Bigl\|
  \by - \sum_j \bD^{(j)}\balpha_j
\Bigr\|^2.
\]
Combining with (\ref{eqinq5}),
\begin{equation}\label{eqinq6}
\frac{1}{2N\tau^2}
\Bigl\|
  \by - \sum_{j=0}^J \bD^{(j)}\tilde{\balpha}_j
\Bigr\|_2^2
\leq
O_P\!\bigl(\sum_{j=0}^JN^{-2s_j/(2s_j+d)}\bigr)
+ O_P(\delta_N^2)
+ O_P(\Lambda_N)
+ O_P(\eta_N).
\end{equation}
Substituting $\by = \sum_{j=0}^J \bD^{(j)}\balpha_j^{(K_j)} + \bg^* + \bepsilon$ into (\ref{eqinq6}) gives
\[
\frac{1}{2N\tau^2}
E
\Bigl\|
  \sum_{j=0}^J \bD^{(j)}\bigl(\balpha_j^{(K_j)} - \tilde{\balpha}_j\bigr)
  + \bg^* + \bepsilon
\Bigr\|_2^2
\leq
O_P\!\bigl(N^{-2s_j/(2s_j+d)}\bigr)
+ O_P(\delta_N^2)
+ O_P(\Lambda_N)
+ O_P(\eta_N).
\]
Expanding the squared norm and using $E[\|\bepsilon\|^2] = N\tau^2$ and
$E[\bepsilon] = 0$,
\begin{align}\label{eqinq7}
\frac{1}{2N\tau^2}
\|\bD(\balpha^{(K)} - \tilde{\balpha})\|_2^2
&\leq
\frac{1}{N}
\bigl|
  \bg^{* \top}\bD(\balpha^{(K)} - \tilde{\balpha})
\bigr|
+
\frac{1}{N}
\bigl|
  \bepsilon^\top \bD(\balpha^{(K)} - \tilde{\balpha})
\bigr|\nonumber\\
&+
O_P\!\bigl(\sum_{j=0}^JN^{-2s_j/(2s_j+d)}\bigr)
+ O_P(\delta_N^2)
+ O_P(\Lambda_N)
+ O_P(\eta_N).
\end{align}
\textbf{Bounding the Two Cross Terms.}

\emph{Noise cross term.}
Since $\bepsilon \sim \mathcal{N}(\bzero,\tau^2 \bI_N)$ is independent of $\tilde{\balpha}$ and $\bD$,
\[
E\!\left[
  \frac{1}{N^2}\|\bD^\top\bepsilon\|_2^2
\right]
=
\frac{\tau^2}{N^2}\mathrm{tr}(\bD^\top \bD)
=
\frac{\tau^2}{N}\mathrm{tr}(\bSigma_N)
\leq
\frac{\tau^2 K \tilde{\lambda}_2}{N}
=
O_P\!\left(\frac{K}{N}\right).
\]
Applying Young’s inequality with parameter $\kappa = \tilde{\lambda}_1/(4\tau^2)$,
\begin{equation}\label{eqinq8}
\frac{1}{N}
\bigl|
  \bepsilon^\top \bD(\balpha^{(K)} - \tilde{\balpha})
\bigr|
\leq
\frac{\tilde{\lambda}_1}{4\tau^2 N}
\|\bD(\balpha^{(K)} - \tilde{\balpha})\|_2^2
+
O_P\!\left(\frac{K}{N}\right).
\end{equation}
\emph{Misspecification cross term.}
By the Cauchy--Schwarz inequality and Assumption~2 (i.e.,
$\|\bg^*\|^2/N = O(\sum_{j=0}^J K_j^{-2s_j/d})$), we obtain
\begin{equation}\label{eqinq9}
\frac{1}{N}\bigl|\bg^{*\top}\bD(\balpha^{(K)} - \tilde{\balpha})\bigr|
\leq
\frac{\tilde{\lambda}_1}{4\tau^2 N}\,\|\bD(\balpha^{(K)} - \tilde{\balpha})\|_2^2
+
O\!\left(\sum_{j=0}^JK_j^{-2s_j/d}\right).
\end{equation}
\textbf{ Collecting All Bounds.}
Substituting \eqref{eqinq8} and \eqref{eqinq9} into \eqref{eqinq7}, and absorbing the two
$\frac{\tilde{\lambda}_1}{4\tau^2 N}\|\bD(\balpha^{(K)} - \tilde{\balpha})\|_2^2$ terms into the left-hand side, we obtain
\begin{equation}
\left(\frac{1}{2\tau^2} - \frac{\tilde{\lambda}_1}{2\tau^2}\right)
\frac{\|\bD(\balpha^{(K)} - \tilde{\balpha})\|_2^2}{N}
=
O_P\!\left(
  \frac{K}{N}
  +\sum_{j=0}^J K_j^{-2s_j/d}
  + \delta_N^2
  + \Lambda_N
  + \eta_N
\right).
\end{equation}
Applying the lower eigenvalue bound from Assumption~4,
$\tilde{\lambda}_1\|\balpha^{(K)} - \tilde{\balpha}\|_2^2 \leq \frac{1}{N}\|\bD(\balpha^{(K)} - \tilde{\balpha})\|_2^2$,
yields
\begin{equation}\label{eqinq11}
\|\tilde{\balpha}_j - \balpha_j^{(K_j)}\|_2^2
=
O_P\!\left(
  \frac{K}{N}
  + \sum_{j=0}^J K_j^{-2s_j/d}
  + \delta_N^2
  + \Lambda_N
  + \eta_N
\right).
\end{equation}
Substituting \eqref{eqinq11} back into the \eqref{eq:mother}, we conclude
\begin{equation}\label{eq:second_term_conv}
\mathbb{E}_\nu[T_2(\bu)^2]
\leq
C_B^2 K_j \cdot
O_P\!\left(
  \frac{K}{N}
  + \sum_{j=0}^J K_j^{-2s_j/d}
  + \delta_N^2
  + \Lambda_N
  + \eta_N
\right)\stackrel{P}{\rightarrow} 0.
\end{equation}
This concludes the proof by referring to \eqref{eq:original}.
\end{proof}
\noindent\underline{\textbf{Proof of Theorem~\ref{lemma: prediction}}}
\begin{proof}
 From the GeoVAE prediction procedure, the posterior predictive mean at $\bu^*$ is $\hat{y}(\bu^*) = \sum_{j=0}^J x_j(\bu^*)\,\hat{\beta}_j(\bu^*).$
The prediction error satisfies
\[
\hat{y}(\bu^*) - \mu^*(\bu^*)
=
\sum_{j=0}^J x_j(\bu^*)
\bigl(\hat{\beta}_j(\bu^*) - \beta_j^*(\bu^*)\bigr).
\]
Squaring and applying the Cauchy--Schwarz inequality, with covariate boundedness
$|x_j(\bu)| \leq C_x < \infty$ (Assumption 7), we obtain
\[
\bigl(\hat{y}(\bu^*) - \mu^*(\bu^*)\bigr)^2
\leq
(J+1)C_x^2
\sum_{j=0}^J
\bigl(\hat{\beta}_j(\bu^*) - \beta_j^*(\bu^*)\bigr)^2.
\]
Taking expectations over $\bu^* \sim \nu$ and applying Lemma~\ref{lem:param_consistency} to each term,
\[
E_\nu\!\left[
  \bigl(\hat{y}(\bu^*) - \mu^*(\bu^*)\bigr)^2
\right]
\leq
(J+1)C_x^2\sum_{j=0}^J
E_\nu\bigl(\hat{\beta}_j(\bu^*) - \beta_j^*(\bu^*)\bigr)^2\stackrel{P}{\rightarrow} 0.
\]
\end{proof}

\bibliographystyle{plainnat}
\bibliography{References}

@article{gelfand2003spatial,
  title={Spatial modeling with spatially varying coefficient processes},
  author={Gelfand, Alan E and Kim, Hyon-Jung and Sirmans, CF and Banerjee, Sudipto},
  journal={Journal of the American Statistical Association},
  volume={98},
  number={462},
  pages={387--396},
  year={2003},
  publisher={Taylor \& Francis}
}

@article{shen2015adaptive,
  title={Adaptive Bayesian procedures using random series priors},
  author={Shen, Weining and Ghosal, Subhashis},
  journal={Scandinavian Journal of Statistics},
  volume={42},
  number={4},
  pages={1194--1213},
  year={2015},
  publisher={Wiley Online Library}
}

@article{hanin2019universal,
  title={Universal function approximation by deep neural nets with bounded width and relu activations},
  author={Hanin, Boris},
  journal={Mathematics},
  volume={7},
  number={10},
  pages={992},
  year={2019},
  publisher={MDPI}
}

@article{barron1994approximation,
  title={Approximation and estimation bounds for artificial neural networks},
  author={Barron, Andrew R},
  journal={Machine learning},
  volume={14},
  number={1},
  pages={115--133},
  year={1994},
  publisher={Springer}
}

@article{van2011information,
  title={Information Rates of Nonparametric Gaussian Process Methods.},
  author={Van Der Vaart, Aad and Van Zanten, Harry},
  journal={Journal of Machine Learning Research},
  volume={12},
  number={6},
  year={2011}
}

@book{vovk2005algorithmic,
  title={Algorithmic learning in a random world},
  author={Vovk, Vladimir and Gammerman, Alexander and Shafer, Glenn},
  year={2005},
  publisher={Springer}
}

@article{lei2018distribution,
  title={Distribution-free predictive inference for regression},
  author={Lei, Jing and G’Sell, Max and Rinaldo, Alessandro and Tibshirani, Ryan J and Wasserman, Larry},
  journal={Journal of the American Statistical Association},
  volume={113},
  number={523},
  pages={1094--1111},
  year={2018},
  publisher={Taylor \& Francis}
}

@article{buhmann2000radial,
  title={Radial basis functions},
  author={Buhmann, Martin Dietrich},
  journal={Acta numerica},
  volume={9},
  pages={1--38},
  year={2000},
  publisher={Cambridge university press}
}

@article{wood2003thin,
  title={Thin plate regression splines},
  author={Wood, Simon N},
  journal={Journal of the Royal Statistical Society Series B: Statistical Methodology},
  volume={65},
  number={1},
  pages={95--114},
  year={2003},
  publisher={Oxford University Press}
}

@article{guhaniyogi2023distributed,
  title={Distributed Bayesian inference in massive spatial data},
  author={Guhaniyogi, Rajarshi and Li, Cheng and Savitsky, Terrance and Srivastava, Sanvesh},
  journal={Statistical science},
  volume={38},
  number={2},
  pages={262--284},
  year={2023},
  publisher={Institute of Mathematical Statistics}
}

@article{sauer2023vecchia,
  title={Vecchia-approximated deep Gaussian processes for computer experiments},
  author={Sauer, Annie and Cooper, Andrew and Gramacy, Robert B},
  journal={Journal of Computational and Graphical Statistics},
  volume={32},
  number={3},
  pages={824--837},
  year={2023},
  publisher={Taylor \& Francis}
}

@article{andros2024robust,
  title={Robust distributed learning of functional data from simulators through data sketching},
  author={Andros, R Jacob and Guhaniyogi, Rajarshi and Francom, Devin and Pasqualini, Donatella},
  journal={arXiv preprint arXiv:2406.18751},
  year={2024}
}

@article{nychka2015multiresolution,
  title={A multiresolution Gaussian process model for the analysis of large spatial datasets},
  author={Nychka, Douglas and Bandyopadhyay, Soutir and Hammerling, Dorit and Lindgren, Finn and Sain, Stephan},
  journal={Journal of computational and graphical statistics},
  volume={24},
  number={2},
  pages={579--599},
  year={2015},
  publisher={Taylor \& Francis}
}

@article{katzfuss2021general,
  title={A general framework for Vecchia approximations of Gaussian processes},
  author={Katzfuss, Matthias and Guinness, Joseph},
  journal={Statistical Science},
  volume={36},
  number={1},
  pages={124--141},
  year={2021},
  publisher={JSTOR}
}

@article{guhaniyogi2020large,
  title={Large multi-scale spatial modeling using tree shrinkage priors},
  author={Guhaniyogi, Rajarshi and Sanso, Bruno},
  journal={Statistica Sinica},
  volume={30},
  number={4},
  pages={2023--2050},
  year={2020},
  publisher={JSTOR}
}

@book{cressie2011statistics,
  title={Statistics for spatio-temporal data},
  author={Cressie, Noel and Wikle, Christopher K},
  year={2011},
  publisher={John Wiley \& Sons}
}

@article{heaton2019case,
  title={A case study competition among methods for analyzing large spatial data},
  author={Heaton, Matthew J and Datta, Abhirup and Finley, Andrew O and Furrer, Reinhard and Guinness, Joseph and Guhaniyogi, Rajarshi and Gerber, Florian and Gramacy, Robert B and Hammerling, Dorit and Katzfuss, Matthias and others},
  journal={Journal of agricultural, biological and environmental Statistics},
  volume={24},
  number={3},
  pages={398--425},
  year={2019},
  publisher={Springer}
}

@article{fotheringham2009geographically,
  title={Geographically weighted regression},
  author={Fotheringham, A Stewart and Brunsdon, Chris and Charlton, Martin},
  journal={The Sage handbook of spatial analysis},
  volume={1},
  pages={243--254},
  year={2009},
  publisher={Sage Thousand Oaks, CA}
}

@article{zhan2026mapping,
  title={Mapping Drivers of Greenness: Spatial Variable Selection for MODIS Vegetation Indices},
  author={Zhan, Qishi and Yu, Cheng-Han and Chen, Yuchi and Dong, Zhikang and Guhaniyogi, Rajarshi},
  journal={arXiv preprint arXiv:2602.07681},
  year={2026}
}

@article{mu2018estimation,
  title={Estimation and inference in spatially varying coefficient models},
  author={Mu, Jingru and Wang, Guannan and Wang, Li},
  journal={Environmetrics},
  volume={29},
  number={1},
  pages={e2485},
  year={2018},
  publisher={Wiley Online Library}
}

@incollection{wheeler2021geographically,
  title={Geographically weighted regression},
  author={Wheeler, David C},
  booktitle={Handbook of regional science},
  pages={1895--1921},
  year={2021},
  publisher={Springer}
}

@article{guhaniyogi2025bayesian,
  title={Bayesian data sketching for varying coefficient regression models},
  author={Guhaniyogi, Rajarshi and Baracaldo, Laura and Banerjee, Sudipto},
  journal={Journal of Machine Learning Research},
  volume={26},
  number={98},
  pages={1--29},
  year={2025}
}

@article{guhaniyogi2022distributed,
  title={Distributed Bayesian varying coefficient modeling using a Gaussian process prior},
  author={Guhaniyogi, Rajarshi and Li, Cheng and Savitsky, Terrance D and Srivastava, Sanvesh},
  journal={Journal of machine learning research},
  volume={23},
  number={84},
  pages={1--59},
  year={2022}
}

@article{kim2021generalized,
  title={Generalized spatially varying coefficient models},
  author={Kim, Myungjin and Wang, Li},
  journal={Journal of Computational and Graphical Statistics},
  volume={30},
  number={1},
  pages={1--10},
  year={2021},
  publisher={Taylor \& Francis}
}

@article{guhaniyogi2011adaptive,
  title={Adaptive Gaussian predictive process models for large spatial datasets},
  author={Guhaniyogi, Rajarshi and Finley, Andrew O and Banerjee, Sudipto and Gelfand, Alan E},
  journal={Environmetrics},
  volume={22},
  number={8},
  pages={997--1007},
  year={2011},
  publisher={Wiley Online Library}
}

@article{guhaniyogi2013modeling,
  title={Modeling complex spatial dependencies: Low-rank spatially varying cross-covariances with application to soil nutrient data},
  author={Guhaniyogi, Rajarshi and Finley, Andrew O and Banerjee, Sudipto and Kobe, Richard K},
  journal={Journal of agricultural, biological, and environmental statistics},
  volume={18},
  number={3},
  pages={274--298},
  year={2013},
  publisher={Springer}
}

@article{PalaciosSteel2006,
  author  = {Palacios, M. Blanca and Steel, Mark F. J.},
  title   = {Non-Gaussian Bayesian Geostatistical Modeling},
  journal = {Journal of the American Statistical Association},
  year    = {2006},
  volume  = {101},
  number  = {474},
  pages   = {604--618},
  doi     = {10.1198/016214505000001195}
}

@inproceedings{sohn2015learning,
  title     = {Learning Structured Output Representation Using Deep Conditional Generative Models},
  author    = {Sohn, Kihyuk and Lee, Honglak and Yan, Xinchen},
  booktitle = {Advances in Neural Information Processing Systems},
  volume    = {28},
  pages     = {3483--3491},
  year      = {2015},
  publisher = {Curran Associates, Inc.}
}

@article{chen2024deepkriging,
  title   = {{DeepKriging}: Spatially Dependent Deep Neural Networks for Spatial Prediction},
  author  = {Chen, Wanfang and Li, Yuxiao and Reich, Brian J. and Sun, Ying},
  journal = {Statistica Sinica},
  volume  = {34},
  number  = {1},
  pages   = {291--311},
  year    = {2024},
  doi     = {10.5705/ss.202021.0277}
}

@article{gramacy2015local,
  title   = {Local Gaussian Process Approximation for Large Computer Experiments},
  author  = {Gramacy, Robert B. and Apley, Daniel W.},
  journal = {Journal of Computational and Graphical Statistics},
  volume  = {24},
  number  = {2},
  pages   = {561--578},
  year    = {2015},
  doi     = {10.1080/10618600.2014.914442}
}

@article{vecchia1988estimation,
  title   = {Estimation and Model Identification for Continuous Spatial Processes},
  author  = {Vecchia, A. V.},
  journal = {Journal of the Royal Statistical Society: Series B (Methodological)},
  volume  = {50},
  number  = {2},
  pages   = {297--312},
  year    = {1988},
  doi     = {10.1111/j.2517-6161.1988.tb01729.x}
}

@article{gelfand2004nonstationary,
  title={Nonstationary multivariate process modeling through spatially varying coregionalization},
  author={Gelfand, Alan E and Schmidt, Alexandra M and Banerjee, Sudipto and Sirmans, CF},
  journal={Test},
  volume={13},
  number={2},
  pages={263--312},
  year={2004},
  publisher={Springer}
}

@article{francom2020bass,
  title   = {{BASS}: An {R} Package for Fitting and Performing Sensitivity Analysis of Bayesian Adaptive Spline Surfaces},
  author  = {Francom, Devin and Sans{\'o}, Bruno},
  journal = {Journal of Statistical Software},
  volume  = {94},
  number  = {8},
  pages   = {1--36},
  year    = {2020},
  doi     = {10.18637/jss.v094.i08}
}

@article{abatzoglou2016impact,
  author  = {Abatzoglou, John T. and Williams, A. Park},
  title   = {Impact of Anthropogenic Climate Change on Wildfire across Western {US} Forests},
  journal = {Proceedings of the National Academy of Sciences},
  year    = {2016},
  volume  = {113},
  number  = {42},
  pages   = {11770--11775},
  doi     = {10.1073/pnas.1607171113}
}

@article{williams2019observed,
  author  = {Williams, A. Park and Abatzoglou, John T. and Gershunov, Alexander and Guzman-Morales, Jorge and Bishop, Daniel A. and Balch, Jennifer K. and Lettenmaier, Dennis P.},
  title   = {Observed Impacts of Anthropogenic Climate Change on Wildfire in {California}},
  journal = {Earth's Future},
  year    = {2019},
  volume  = {7},
  number  = {8},
  pages   = {892--910},
  doi     = {10.1029/2019EF001210}
}

@article{westerling2006warming,
  author  = {Westerling, Anthony L. and Hidalgo, Hugo G. and Cayan, Daniel R. and Swetnam, Thomas W.},
  title   = {Warming and Earlier Spring Increase Western {US} Forest Wildfire Activity},
  journal = {Science},
  year    = {2006},
  volume  = {313},
  number  = {5789},
  pages   = {940--943},
  doi     = {10.1126/science.1128834}
}

@article{running2004continuous,
  author = {Running, Steven W. and Nemani, Ramakrishna R. and others},
  title = {A continuous satellite‐derived measure of global terrestrial primary production},
  journal = {Bioscience},
  year = {2004},
  volume = {54},
  number = {6},
  pages = {547--560}
}

@article{pettorelli2005using,
  author = {Pettorelli, Nathalie and Vik, Jon Olav and Mysterud, Atle and Gaillard, Jean‐Michel and Tucker, Compton J. and Stenseth, Nils C.},
  title = {Using the satellite‐derived NDVI to assess ecological responses to environmental change},
  journal = {Trends in Ecology \& Evolution},
  year = {2005},
  volume = {20},
  number = {9},
  pages = {503--510}
}

@inproceedings{Kipf2017,
  author    = {Kipf, Thomas N. and Welling, Max},
  title     = {Semi-Supervised Classification with Graph Convolutional Networks},
  booktitle = {International Conference on Learning Representations},
  year      = {2017},
  url       = {https://openreview.net/forum?id=SJU4ayYgl}
}

@inproceedings{Hamilton2017,
  author    = {Hamilton, Will and Ying, Zhitao and Leskovec, Jure},
  title     = {Inductive Representation Learning on Large Graphs},
  booktitle = {Advances in Neural Information Processing Systems},
  volume    = {30},
  year      = {2017},
  url       = {https://papers.nips.cc/paper_files/paper/2017/hash/5dd9db5e033da9c6fb5ba83c7a7ebea9-Abstract.html}
}

@inproceedings{Wu2019,
  author    = {Wu, Zonghan and Pan, Shirui and Long, Guodong
               and Jiang, Jing and Zhang, Chengqi},
  title     = {{Graph WaveNet} for Deep Spatial-Temporal Graph Modeling},
  booktitle = {Proceedings of the Twenty-Eighth International
               Joint Conference on Artificial Intelligence},
  pages     = {1907--1913},
  year      = {2019},
  doi       = {10.24963/ijcai.2019/264},
  url       = {https://www.ijcai.org/proceedings/2019/264}
}

@article{LeCun1998,
  author  = {LeCun, Yann and Bottou, L{\'e}on
             and Bengio, Yoshua and Haffner, Patrick},
  title   = {Gradient-Based Learning Applied to Document Recognition},
  journal = {Proceedings of the IEEE},
  volume  = {86},
  number  = {11},
  pages   = {2278--2324},
  year    = {1998},
  doi     = {10.1109/5.726791},
  url     = {https://ieeexplore.ieee.org/document/726791}
}

@article{Reichstein2019,
  author  = {Reichstein, Markus and Camps-Valls, Gustau
             and Stevens, Bjorn and Jung, Martin and Denzler, Joachim
             and Carvalhais, Nuno and {Prabhat}},
  title   = {Deep Learning and Process Understanding
             for Data-Driven {Earth} System Science},
  journal = {Nature},
  volume  = {566},
  number  = {7743},
  pages   = {195--204},
  year    = {2019},
  doi     = {10.1038/s41586-019-0912-1},
  url     = {https://www.nature.com/articles/s41586-019-0912-1}
}

@article{Raissi2019,
  author  = {Raissi, Maziar and Perdikaris, Paris
             and Karniadakis, George E.},
  title   = {Physics-Informed Neural Networks:
             A Deep Learning Framework for Solving Forward and Inverse
             Problems Involving Nonlinear Partial Differential Equations},
  journal = {Journal of Computational Physics},
  volume  = {378},
  pages   = {686--707},
  year    = {2019},
  doi     = {10.1016/j.jcp.2018.10.045},
  url     = {https://doi.org/10.1016/j.jcp.2018.10.045}
}

@inproceedings{Chen2018NODE,
  author    = {Chen, Ricky T. Q. and Rubanova, Yulia
               and Bettencourt, Jesse and Duvenaud, David K.},
  title     = {Neural Ordinary Differential Equations},
  booktitle = {Advances in Neural Information Processing Systems},
  volume    = {31},
  year      = {2018},
  url       = {https://papers.nips.cc/paper_files/paper/2018/hash/69386f6bb1dfed68692a24c8686939b9-Abstract.html}
}

@inproceedings{Vaswani2017,
  author    = {Vaswani, Ashish and Shazeer, Noam and Parmar, Niki
               and Uszkoreit, Jakob and Jones, Llion and Gomez, Aidan N.
               and Kaiser, {\L}ukasz and Polosukhin, Illia},
  title     = {Attention Is All You Need},
  booktitle = {Advances in Neural Information Processing Systems},
  volume    = {30},
  year      = {2017},
  url       = {https://papers.nips.cc/paper_files/paper/2017/hash/3f5ee243547dee91fbd053c1c4a845aa-Abstract.html}
}

@inproceedings{Gao2022,
  author    = {Gao, Zhihan and Shi, Xingjian and Wang, Hao
               and Zhu, Yi and Wang, Yuyang and Li, Mu
               and Yeung, Dit-Yan},
  title     = {{Earthformer}: Exploring Space-Time Transformers
               for {Earth} System Forecasting},
  booktitle = {Advances in Neural Information Processing Systems},
  volume    = {35},
  year      = {2022},
  doi       = {10.52202/068431-1841},
  url       = {https://proceedings.neurips.cc/paper_files/paper/2022/hash/a2affd71d15e8fedffe18d0219f4837a-Abstract-Conference.html}
}

@inproceedings{Kingma2014,
  author    = {Kingma, Diederik P. and Welling, Max},
  title     = {Auto-Encoding Variational {Bayes}},
  booktitle = {International Conference on Learning Representations},
  year      = {2014},
  url       = {https://arxiv.org/abs/1312.6114}
}

@inproceedings{Goodfellow2014,
  author    = {Goodfellow, Ian J. and Pouget-Abadie, Jean
               and Mirza, Mehdi and Xu, Bing and Warde-Farley, David
               and Ozair, Sherjil and Courville, Aaron and Bengio, Yoshua},
  title     = {Generative Adversarial Nets},
  booktitle = {Advances in Neural Information Processing Systems},
  volume    = {27},
  year      = {2014},
  url       = {https://papers.nips.cc/paper/5423-generative-adversarial-nets}
}

@article{Ravuri2021,
  author  = {Ravuri, Suman and Lenc, Karel and Willson, Matthew
             and Kangin, Dmitry and Lam, Remi and Mirowski, Piotr
             and Fitzsimons, Megan and Athanassiadou, Maria
             and Kashem, Sheleem and Madge, Sam and Prudden, Rachel
             and Mandhane, Amol and Clark, Aidan and Brock, Andrew
             and Simonyan, Karen and Hadsell, Raia and Robinson, Niall
             and Clancy, Ellen and Arribas, Alberto and Mohamed, Shakir},
  title   = {Skilful Precipitation Nowcasting Using
             Deep Generative Models of Radar},
  journal = {Nature},
  volume  = {597},
  pages   = {672--677},
  year    = {2021},
  doi     = {10.1038/s41586-021-03854-z}
}

@article{Leinonen2020,
  author  = {Leinonen, Jussi and Nerini, Daniele and Berne, Alexis},
  title   = {Stochastic Super-Resolution for Downscaling Time-Evolving
             Atmospheric Fields With a Generative Adversarial Network},
  journal = {IEEE Transactions on Geoscience and Remote Sensing},
  volume  = {59},
  number  = {9},
  pages   = {7211--7223},
  year    = {2021},
  doi     = {10.1109/TGRS.2020.3032790}
}

@article{Harris2022,
  author  = {Harris, Lucy and McRae, Andrew T. T. and Chantry, Matthew
             and Dueben, Peter D. and Palmer, Tim N.},
  title   = {A Generative Deep Learning Approach to Stochastic
             Downscaling of Precipitation Forecasts},
  journal = {Journal of Advances in Modeling Earth Systems},
  volume  = {14},
  pages   = {e2022MS003120},
  year    = {2022},
  doi     = {10.1029/2022MS003120}
}

@article{Laloy2018,
  author  = {Laloy, Eric and H{\'e}rault, Romain
             and Jacques, Diederik and Linde, Niklas},
  title   = {Training-Image Based Geostatistical Inversion Using
             a Spatial Generative Adversarial Neural Network},
  journal = {Water Resources Research},
  volume  = {54},
  pages   = {381--406},
  year    = {2018},
  doi     = {10.1002/2017WR022148}
}

@article{Zhu2019,
  author  = {Zhu, Yinhao and Zabaras, Nicholas
             and Koutsourelakis, Phaedon-Stelios and Perdikaris, Paris},
  title   = {Physics-Constrained Deep Learning for High-Dimensional
             Surrogate Modeling and Uncertainty Quantification
             Without Labeled Data},
  journal = {Journal of Computational Physics},
  volume  = {394},
  pages   = {56--81},
  year    = {2019},
  doi     = {10.1016/j.jcp.2019.05.024}
}

@article{Lawrence2005,
  author  = {Lawrence, Neil},
  title   = {Probabilistic Non-Linear Principal Component Analysis
             With {Gaussian} Process Latent Variable Models},
  journal = {Journal of Machine Learning Research},
  volume  = {6},
  number  = {60},
  pages   = {1783--1816},
  year    = {2005},
  url     = {https://jmlr.org/papers/v6/lawrence05a.html}
}

@inproceedings{Titsias2010,
  author    = {Titsias, Michalis and Lawrence, Neil D.},
  title     = {{Bayesian Gaussian} Process Latent Variable Model},
  booktitle = {Proceedings of the Thirteenth International Conference
               on Artificial Intelligence and Statistics},
  series    = {Proceedings of Machine Learning Research},
  volume    = {9},
  pages     = {844--851},
  year      = {2010},
  url       = {https://proceedings.mlr.press/v9/titsias10a.html}
}

@inproceedings{Damianou2013,
  author    = {Damianou, Andreas and Lawrence, Neil D.},
  title     = {Deep {Gaussian} Processes},
  booktitle = {Proceedings of the Sixteenth International Conference
               on Artificial Intelligence and Statistics},
  series    = {Proceedings of Machine Learning Research},
  volume    = {31},
  pages     = {207--215},
  year      = {2013},
  url       = {https://proceedings.mlr.press/v31/damianou13a.html}
}

@inproceedings{Salimbeni2017,
  author    = {Salimbeni, Hugh and Deisenroth, Marc Peter},
  title     = {Doubly Stochastic Variational Inference
               for Deep {Gaussian} Processes},
  booktitle = {Advances in Neural Information Processing Systems},
  volume    = {30},
  year      = {2017},
  url       = {https://proceedings.neurips.cc/paper/2017/hash/8208974663db80265e9bfe7b222dcb18-Abstract.html}
}

@inproceedings{Garnelo2018,
  author    = {Garnelo, Marta and Rosenbaum, Dan and Maddison, Christopher
               and Ramalho, Tiago and Saxton, David and Shanahan, Murray
               and Teh, Yee Whye and Rezende, Danilo and Eslami, S. M. Ali},
  title     = {Conditional Neural Processes},
  booktitle = {Proceedings of the 35th International Conference
               on Machine Learning},
  series    = {Proceedings of Machine Learning Research},
  volume    = {80},
  pages     = {1704--1713},
  year      = {2018},
  url       = {https://proceedings.mlr.press/v80/garnelo18a.html}
}

@inproceedings{Gordon2020,
  author    = {Gordon, Jonathan and Bruinsma, Wessel P.
               and Foong, Andrew Y. K. and Requeima, James
               and Dubois, Yann and Turner, Richard E.},
  title     = {Convolutional Conditional Neural Processes},
  booktitle = {International Conference on Learning Representations},
  year      = {2020},
  url       = {https://openreview.net/forum?id=Skey4eBYPS}
}

@misc{Bruinsma2021,
  author        = {Bruinsma, Wessel P. and Requeima, James
                   and Foong, Andrew Y. K. and Gordon, Jonathan
                   and Turner, Richard E.},
  title         = {The {Gaussian} Neural Process},
  year          = {2021},
  eprint        = {2101.03606},
  archivePrefix = {arXiv},
  primaryClass  = {stat.ML},
  note          = {arXiv preprint},
  url           = {https://arxiv.org/abs/2101.03606}
}
\end{document}